\documentclass[11pt]{article}
\usepackage[letterpaper, margin=1in]{geometry}

\usepackage{amsthm}
\usepackage{amssymb}
\usepackage{amsmath}
\usepackage{thmtools}

\usepackage{pgfplots}

\usepackage{hyperref}
\usepackage{cleveref}
\usepackage{bm}

\usepackage{soul}
\usepackage{algorithm}
\usepackage{algpseudocode}
\algblockdefx[WithProb]{WithProb}{EndWithProb}[1]{\textbf{with probability} #1}{}
\algcblockdefx[WithProb]{WithProb}{Otherwise}{EndWithProb}{\textbf{otherwise}}{\textbf{end}}

\title{Sampling Sphere Packings with Continuum Glauber Dynamics}

\author{Aiya Kuchukova\\akuchukova3@gatech.edu\\Georgia Institute of Technology\and Santosh S. Vempala\\vempala@gatech.edu\\Georgia Institute of Technology\and Daniel J. Zhang\\dzhang381@gatech.edu\\Georgia Institute of Technology}

\newtheoremstyle{bfnote}
{}{}
{\slshape}{}
{\bfseries}{\bfseries.}
{ }
{\thmname{#1}\thmnumber{ #2}\thmnote{ \normalfont{}#3}}

\newtheoremstyle{Claim}
{}{}
{\slshape}{}
{\itshape}{.}
{ }
{\thmname{#1}\thmnumber{ #2}\thmnote{ \normalfont{}#3}}

\newtheorem{theorem}{Theorem}[section]
\newtheorem*{theorem*}{Theorem}
\newtheorem{proposition}[theorem]{Proposition}
\newtheorem{lemma}[theorem]{Lemma}
\newtheorem*{lemma*}{Lemma}

\newtheorem{corollary}[theorem]{Corollary}

\newtheorem{fact}[theorem]{Fact}
\newtheorem{observation}[theorem]{Observation}
\newtheorem*{corollary*}{Corollary}

\theoremstyle{definition}
\newtheorem{definition}[theorem]{Definition}
\newtheorem*{definition*}{Definition}
\newtheorem{example}[theorem]{Example}
\newtheorem{remark}[theorem]{Remark}

\newtheorem*{example*}{Example}

\theoremstyle{Claim}
\newtheorem{claim}{Claim}[theorem]

\newcommand{\N}{{\mathbb{N}}}

\newcommand{\R}{\mathbb{R}}
\newcommand{\C}{\mathbb{C}}
\newcommand{\Q}{\mathbb{Q}}
\newcommand{\Z}{\mathbb{Z}}
\renewcommand{\P}{\mathbb{P}}
\newcommand{\E}{\mathbb{E}}

\newcommand{\bcd}{\text{BCD}}
\newcommand{\si}{\text{SI}}
\renewcommand{\epsilon}{\varepsilon}
\newcommand{\eps}{\varepsilon}

\DeclareMathOperator{\vol}{vol}

\DeclareMathOperator{\dist}{dist}
\DeclareMathOperator{\supp}{supp}
\DeclareMathOperator{\poly}{poly}
\DeclareMathOperator{\Var}{Var}

\DeclareMathOperator{\Exp}{Exp}

\DeclareMathOperator{\Pois}{Pois}

\foreach \x in {A,...,Z}{
	\expandafter\xdef\csname mc\x\endcsname{\noexpand\mathcal{\x}}
}

\newcommand{\defeq}{:=}

\allowdisplaybreaks

\begin{document}

\maketitle

\begin{abstract}
Continuum Glauber dynamics is a spatial birth-death process whose stationary distribution is a Gibbs distribution. We establish a spectral gap for Continuum Glauber dynamics applied to Gibbs point processes with repulsive pair potentials, a well-known special case of which is the hard sphere model.
For arbitrary-range repulsive pair potentials, we show that a continuous version of Spectral Independence suffices to establish a spectral gap.
This extends the regime of activity for which Continuum Glauber dynamics is known to mix, yielding a simple efficient sampling algorithm for arbitrary-range pair potentials that matches the known efficient sampling regime for finite-range pair potentials currently based on specialized algorithms.
As a consequence, we also improve the threshold up to which packings of fixed size/density can be efficiently sampled from a bounded domain, the first improvement since Kannan, Mahoney and Montenegro (2003).

To prove these results, we develop continuous analogs of Spectral Independence and negative fields localization. We show that a stronger variant of zero-freeness implies Spectral Independence, which in turn allows us to run the localization scheme to boost the spectral gap of Continuum Glauber dynamics from smaller activity to larger activity. While this follows the high-level blueprint of Chen and Eldan (2022) for the discrete setting, we have to address several novel difficulties due to the continuous setting. Notably, we avoid discretization in the algorithm and the analysis and work directly in the continuous setting.
\end{abstract}

\thispagestyle{empty}

\newpage
\tableofcontents
\thispagestyle{empty}

\newpage

\section{Introduction} 
Sphere packing is a classical problem of great interest to mathematicians, physicists and computer scientists. Kepler's conjecture posed in 1611 led to a body of work including upper and lower bounds on the optimal density of sphere packings in any dimension leading to breakthroughs by Hales \cite{Hales2005Kepler} on optimal packings in dimension 3, Viazovska \cite{Viazovska2017} in dimension 8 and by Cohn, Kumar, Miller, Radchenko, and Viazovska in dimension 24 \cite{Cohn_2017}. Recent  advances by   
Campos, Jenssen, Michelen, Sahasrabudhe \cite{campos2023newlowerboundsphere} and Klartag \cite{klartag2025latticepackingsphereshigh} provide general lower bounds in high dimension. 
\setcounter{page}{1}

A closely related problem, motivated in part by statistical physics, is efficiently {\em sampling} sphere packings of high density. Sampling is a rich direction of theoretical computer science on its own, and in the case of sphere packings, it also has connections to statistical physics phenomena such as phase transitions and correlation decay. 
In fact, the classical Metropolis algorithm~\cite{metropolis} was originally introduced  to address the question of sampling from the hard sphere model.

\paragraph{Sampling Sphere Packings.} 
There are two closely related but distinct versions of the hard sphere model. In one, called the \emph{canonical ensemble},
the number of spheres is fixed and the probability density is uniform over all sphere packings with the prescribed number of spheres. In the other, called the \emph{grand canonical ensemble}, 
the weight of each sphere packing is proportional to the activity parameter $\lambda$ raised to the power of the number of spheres in the packing. That is, the weight of a configuration $S$ is proportional to $\lambda^{|S|}$ (for a formal definition, see \Cref{def: hard sphere}). 

In the canonical ensemble, as the number of spheres increases, sampling generally becomes harder. In the grand canonical ensemble, as $\lambda$ increases, configurations with more spheres become more likely and sampling generally becomes harder. 
The problem of sampling sphere packings in both models has been addressed by various Markov chains.

For the canonical ensemble, Kannan, Mahoney and Montenegro \cite{KMM03} showed that a simple Markov chain that moves a random sphere to a random location if available (Kawasaki-type dynamics)
converges rapidly to the uniform distribution over fixed-cardinality sphere packings, provided the density satisfies $\rho\le (1-\delta)/2^{d+1}$ for some $\delta>0$ (here we assume that spheres have volume 1). The work of Boudou, Caputo, Pra, and Posta \cite{Boudou2005SpectralGE} gives a more general result on the spectral gap for repulsive particle interactions (which we will discuss later), but their result for Kawasaki-type dynamics requires slightly smaller density $\rho\le (1-\delta)\frac{1}{3\cdot 2^{d+1}}$.

For the grand canonical ensemble, Michelen and Perkins \cite{MP21Potential,MP22Analyticity,MP22Strong}
showed that ``block dynamics" on $\Lambda_n=[-n,n]^d$ mixes in $O(|\Lambda_n|\log|\Lambda_n|)$ 
steps where $|\Lambda_n|$ is the volume of the box domain. This result gives the widest known range 
of activity for which efficient sampling is known.

The hard sphere model is a special case of Gibbs point processes with repulsive pair potentials (where a point does not necessarily forbid points nearby, but only makes them less likely; see \Cref{def:gibbs_point_process}). The main results of this paper apply to this general model with explicit corollaries stated for the hard spheres model. 

\paragraph{Gibbs point processes with repulsive pair potentials.}
In this paper, we analyze more general Gibbs point processes with repulsive pair potentials. For now, let $\mathbb{X}$ be $\R^d$ equipped with the Euclidean metric and Lebesgue measure, though some of our results hold for more general spaces.
A pair potential is a symmetric function $\phi:\mathbb{X}\times\mathbb{X}\to(-\infty,\infty]$.

\begin{definition}[Gibbs point process]\label{def:gibbs_point_process}
    Let $\Lambda\subseteq\mathbb{X}$ have finite measure. A Gibbs point process on $\Lambda$ with activity function $\bm\lambda:\Lambda\to\R_{\ge 0}$ and pair potential $\phi:\mathbb{X}\times\mathbb{X}\to(-\infty,\infty]$ can be defined using its partition function
    \begin{align*}
    Z(\bm{\lambda})&\defeq \sum_{k\ge 0}\dfrac{1}{k!}\int_{\Lambda^k}\bm{\lambda}(x_1)\cdots\bm{\lambda}(x_k)e^{-H(x_1,\ldots,x_k)}dx_1\cdots dx_k
    \end{align*}
    where
    \begin{align*}
    H(x_1,\ldots,x_k)&=\sum_{1\le i<j\le k}\phi(x_i,x_j).
    \end{align*}
    Its \textit{Gibbs measure }is defined by
    \begin{align*}
    \mu_{\bm{\lambda}}(A)&\defeq \sum_{k\ge 0}\dfrac{1}{k!}\int_{\Lambda^k}\mathbf{1}_{\{x_1,\ldots,x_k\}\in A}\dfrac{\bm{\lambda}(x_1)\cdots\bm{\lambda}(x_k)}{Z(\bm{\lambda})}e^{-H(x_1,\ldots,x_k)}dx_1\cdots dx_k
    \end{align*}
\end{definition}

\begin{definition}
A pair potential $\phi:\mathbb{X}\times\mathbb{X}\to(-\infty,\infty]$ is {\em repulsive} if $\phi(x,y)\ge 0$ for all $x$ and $y$. It is {\em finite-range} if there exists $R$ such that for all $x, y$ with distance at least $R$, $\phi(x, y) = 0$. 
\end{definition}

The following constant will be relevant for our analysis. Roughly speaking, it is a measure of how repulsive a particle is.
\begin{definition}[Temperedness constant]\label{def: temperedness}
    \[ C_\phi \defeq \sup_{x\in \mathbb{X}}\int_{\mathbb{X}}|1 - e^{-\phi(x, y)}| dy\]  
We say that $\phi$ is \emph{tempered} if $C_{\phi}<\infty$. We assume throughout the paper that pair potentials are tempered.
\end{definition}

The hard sphere model on $\Lambda\subseteq\R^d$ with activity $\lambda> 0$ (see Def.~\ref{def: hard sphere}) can be thought of as a Gibbs point process on $\R^d$ with a finite-range repulsive pair potential:
    \begin{align*}
        \bm{\lambda}(x)&= \begin{cases}
            \lambda&x\in\Lambda\\
            0&x\notin\Lambda
        \end{cases}&
        \phi(x,y)&=\begin{cases}
            \infty&\dist(x,y)<2r\\
            0&\text{otherwise.}
        \end{cases}
    \end{align*}

There are many Gibbs point processes of independent interest beyond the hard sphere model, e.g., the Strauss model \cite{strauss1975model}\cite{kelly1976note}, the Widom-Rowlinson model \cite{widom1970new}, certain germ-grain models \cite{sabatini2024special}, the Gaussian overlap model \cite{berne1972gaussian}, the Yukawa model \cite{rowlinson1989yukawa} and the generalized exponential model \cite{bacher2014explaining}.

\paragraph{Continuum Glauber.}
In this paper, we analyze Continuum Glauber, a natural Markov process for repulsive point processes defined by \cite{BCC02}. It is an important special case of birth-death processes that were first defined by Preston \cite{Preston75}.
These processes 
are used in the context of biological, physical and chemical processes \cite{FRENKEL20021}. Continuum Glauber has been previously used to sample from Gibbs point processes (e.g. the perfect simulation algorithm of \cite{ferrari2002perfect}); however, we are not aware of  rigorous analyses of its runtime.

The definition of the process (\Cref{def:cont_glauber}) has points appearing (if possible to place without conflict) at a rate of $\lambda$ per unit volume and disappearing at a rate of 1 (i.e., each point lasts for $\Exp(1)$ time). 
Kondratiev and Lytvynov~\cite{KL03} showed that when $\lambda \le (1-\delta)/C_{\phi}$ (see \Cref{def: temperedness} for the ``temperedness'' constant), Continuum Glauber dynamics has a spectral gap of at least $\delta$, for an infinite domain. In this paper, we will focus on bounded domains, a more natural setting for sampling. 
 Starting from the empty configuration, a standard argument (\cite{BGL14} Theorem 4.2.5) shows that a spectral gap implies convergence in $\chi^2$-distance.
The result of \cite{KL03} implies $O_d(|\Lambda|)$ mixing of Continuum Glauber for activity $\lambda<1/C_{\phi}$ from the empty configuration. 
Dai Pra and Posta~\cite{DP13} showed the stronger result that for $\lambda\le (1-\delta)\dfrac{1}{C_{\phi}}$, Continuum Glauber dynamics satisfies a Modified Log-Sobolev Inequality (MLSI) with constant at least $\delta$.

The usual definition of mixing time is a uniform time bound over all possible initial states to reach a certain distance of the stationary distribution. For general Gibbs point processes with repulsive pair potentials, the number of initial points could be arbitrarily large, and could take an arbitrarily large amount of time to die out, precluding a uniform mixing time bound. We address this problem by using a burn-in argument, and obtain mixing results starting from any state that depend logarithmically on the number of initial points (see \Cref{subsec: burn-in} for details).  

For implementation purposes, we give a discrete-time algorithm (\Cref{alg:birth_death}) that simulates this process. In \Cref{sec:mixing_time}, we show that this is equivalent to Continuum Glauber and the latter's mixing time bound implies a runtime bound for this algorithm. The function $H$ in the algorithm is defined as $H(\eta)=\sum_{x\neq y\in\eta}\phi(x,y)$. 

\begin{algorithm}[H]
\caption{Simulate Continuum Glauber for $T$ time with potential $\phi$}
\label{alg:birth_death}
\begin{algorithmic}[1]
\State $\eta \gets \emptyset$ \Comment{The initial configuration of particles}
\State $t\gets 0$ \Comment{Initialize the simulation clock}
\Loop 
    \State Sample time increment $h \sim \text{Exp}(|\eta| + \bm\lambda(\Lambda))$ \Comment{Time until the next event}
    \If{$t + h > T$}
        \State \Return $\eta$ \Comment{If next event is after total time, return current state}
     \EndIf
    \State $t \gets t+h$
    \State{Pick $r \in [0,1]$ uniformly}
    \If {$r \leq \frac{|\eta|}{|\eta| + \bm\lambda(\Lambda)}$} \Comment{With some probability a particle dies}
        \State Choose a particle $x \in \eta$ uniformly at random
         \State $\eta \gets \eta \setminus \{x\}$ 
     \Else{ }
     \Comment{Otherwise, attempt to add a particle}
         \State Sample a location $y \in \Lambda$ with density proportional to $\bm\lambda$
         \State With probability $e^{-(H(\eta\cup\{y\})-H(\eta))}$,
         $\eta \gets \eta \cup \{y\}$\label{lst:line:birth}
        \EndIf
\EndLoop
\end{algorithmic}
\end{algorithm}

\subsection{Main results} 
Our first main result is a constant lower bound for the spectral gap of Continuum Glauber. The proof uses extensions of Spectral Independence and Negative Fields Localization to continuous state spaces, which we develop here. We will provide a detailed technical overview after stating the results. We note that our result does not need the pair potential to be finite-range.

\begin{restatable}[Spectral Gap of CG]{theorem}{thmspectralgap}\label{thm:spectral_gap}
    Let $\Lambda\subset \mathbb{X}$ be a domain of finite measure, $\phi:\mathbb{X}\times \mathbb{X}\to[0,\infty]$ be a tempered repulsive pair potential, $\lambda < \frac{e}{\Delta_{\phi}}$ and $\bm\lambda:\Lambda \rightarrow [0,\lambda]$ be an activity function. Then, there exists a constant $c = c(\phi, \lambda) > 0$, independent of $\Lambda$ and $\bm\lambda$ such that Continuum Glauber on $\Lambda$ for the Gibbs point process with potential $\phi$ and activity $\bm\lambda$ has spectral gap at least $c$.
\end{restatable}

Here, $\Delta_{\phi}$ is the potential-weighted connective constant introduced in \cite{MP21Potential}, which satisfies $\Delta_\phi\le C_\phi$. This implies improvement in the range of activity $\lambda$ for which we can prove a spectral gap lower bound.
See Section \ref{sec: boosting} for more details. For estimates on $\Delta_{\phi}$, we refer the reader to the discussion in \cite{MP21Potential}.

Our proof is based on extending the notions of Spectral Independence (SI) and Negative Fields Localization to the continuous setting, which we discuss in more detail in the technical overview.

We show that SI follows from a property of the distribution we call {\em Bounded Complex Density (BCD)} (Def.~\ref{def:bcd}), which was shown by \cite{MP21Potential} to hold up to $\lambda<\frac{e}{\Delta_{\phi}}$ for Gibbs point processes with repulsive pair potentials.
BCD is a stronger version of \textit{zero-freeness}, which was used to show spectral independence in the discrete hard core model \cite{chen2024spectral}.

We also show that SI follows from {\em Strong Spatial Mixing (SSM)} for Gibbs point processes with finite-range repulsive pair potentials. The disadvantage of this approach is that it only works for finite-range pair potentials, and SSM is also only known up to
$\lambda < \frac{e}{\Delta_{\phi}}$ \cite{MP22Strong}. However, to our knowledge SSM and BCD are incomparable, so it may be able to improve the regime for SSM without improving the regime for BCD.

For $\lambda<\frac{e}{C_\phi}$, we provide another proof of spectral independence that allows us to compute an explicit lower bound on the spectral gap (see \Cref{subsec: explicit SI}).
Note that $C_\phi$ also has a natural interpretation in the hard sphere model: $C_{\phi} = 2^d \vol(B_d(0, r))$, and with $\vol(B_d(0, r))=1$, we can simplify to $\Delta_{\phi} \leq C_{\phi} = 2^d$.

\begin{restatable}[Spectral Gap of CG for hard spheres]{theorem}{thmspectralgaphardspheres}\label{thm:spectral_gap_hard_spheres}
    For the hard sphere model with balls of unit volume in bounded region $\Lambda\subset\R^d$ at activity $\lambda \le e^{-\delta} \frac{e}{C_{\phi}}$, the spectral gap of Continuum Glauber is at least 
    ${\dfrac{1}{2}\exp\left(-2\left(1+\dfrac{e2^{d+1}d!}{\delta^d}\right)\right)=e^{-\exp(O(d\log\frac{d}{\delta}))}}$.
\end{restatable}

The spectral gap lower bound allows us to bound the mixing time of Continuum Glauber for Gibbs point processes with repulsive pair potentials. We state this in terms of the threshold $\lambda_{\si}\in(0,\infty]$ up to which Spectral Independence holds, so that if the latter is improved in the future, the mixing and sampling results also follow for the wider range of the activity parameter.
This could be done for example by improving the regimes for which BCD or SSM is known to hold.
We remark that spectral gap itself also implies spectral independence (\Cref{appendix:sg_implies_si}), though this fact will not be used in our results.

We also show that continuum Glauber mixes from any initial configuration. This result follows from \Cref{thm:burnin:main} along with the previously mentioned spectral gaps.

\begin{theorem}[Mixing]\label{thm: main mixing}
For any finite starting configuration $S\subseteq\Lambda$, 
    \begin{itemize}
        \item For a Gibbs point process with a repulsive pair potential and $\lambda<e^{-\delta}\lambda_{\si}$, 
        Continuum Glauber dynamics starting from the configuration $S$ produces a sample within $\eps$ TV distance of the target point process distribution in continuous time 
    \begin{align*}      O_{d,\delta}\left(|\Lambda|+\log\left(\frac{1}{\eps}\right) + \log(|S|+1)\right).
    \end{align*}
        \item For the hard spheres model with $\lambda<e^{-\delta}\frac{1}{C_{\phi}}$, continuum Glauber starting from a configuration $S$ produces a sample with $\epsilon$ TV distance of the hard sphere distribution in continuous time
        \[ O\left(e^{\exp(O(d\log\frac{d}{\delta}))}\left(|\Lambda| + \log(\frac{1}{\eps})\right)+\log(|S|+1)\right) \]
    \end{itemize}
\end{theorem}

\newpage
We briefly discuss related known results. 

 For $\Lambda=[0,l)^d$ and $\lambda\le (1-\delta)\dfrac{e}{2^d}$, Friedrich, G\"obel, Krejca, and Pappik~\cite{FGKP21} used discretization to show that for all $\epsilon\in(0,1]$, there is randomized $\epsilon$-approximation of the partition function of the hard sphere model in time polynomial in $|\Lambda|^{1/\delta^2}$ and $\dfrac{1}{\epsilon}$.
In the same setting, Friedrich, G\"obel, Katzmann, Krejca, and Pappik~\cite{FGKKP22} 
showed
that for all $\epsilon\in(0,1]$, 
there is an $\eps$-approximate sampling algorithm that runs in 
$\poly(|\Lambda|/\epsilon)$-time.
The same group \cite{friedrich2022using} later gave a $\poly(|\Lambda|/\epsilon)$-time $\epsilon$-approximate sampling algorithm that works up to $\lambda<e/C_{\phi}$, which also works for repulsive pair potentials. They also give a quasipolynomial-time sampling algorithm for $\lambda<e/\Delta_{\phi}$.

Guo and Jerrum~\cite{GJ21} gave a perfect sampler for the hard sphere model at activity $\lambda<2^{-(d+1/2)}$ with runtime $O(|\Lambda|)$. Subsequently, Anand, G\"obel, Pappik, and Perkins~\cite{AGPP24} showed the existence of a perfect sampling algorithm for the hard sphere model on finite boxes $\Lambda\subseteq \R^d$ whose expected number of iterations is $O(|\Lambda|)$ assuming Strong Spatial Mixing, which is known to hold up to $\lambda<e/\Delta_{\phi}$. This result also applies to Gibbs point processes with finite-range repulsive pair potentials. Note that the runtime constant depends on Strong Spatial Mixing constants, 
the activity $\lambda$, and the dimension. 

Our results for efficient sampling from the grand canonical ensemble improve upon the range of activity for algorithms given by \cite{HPP21}, \cite{GJ21}, \cite{FGKKP22},  \cite{friedrich2022using} and match the range of activity achieved by \cite{MP22Strong}, \cite{AGPP24} using more sophisticated algorithms for the hard sphere model. 
We note that our results (a) do not require pairwise interactions to have finite-range and (b) are for Continuum Glauber, a simpler and more natural process of independent interest.

Establishing mixing from an arbitrary starting configuration requires an additional argument (compared to starting from the empty configuration). 
An arbitrary starting distribution does not have bounded relative density with respect to the target distribution. 
To handle this, we use a ``burn-in'' argument showing that after a small amount of time we reach a distribution that is close to one with bounded relative density (see \Cref{subsec: burn-in}). 
We also bound the time complexity of \Cref{alg:birth_death} for Continuum Glauber.

\begin{restatable}[Runtime of CG]{theorem}{thmruntime}
    Given Spectral Independence for a repulsive point process up to $\lambda_{\si}$, 
    for $\lambda<\lambda_{\si}$, \Cref{alg:birth_death}
    with the appropriate choice of $T$ produces a sample within $\eps\in(0,1)$ TV distance of the Gibbs distribution with the expected number of iterations bounded by  
    \begin{align*}      O_{\phi,\lambda}\left(|\Lambda|^2+|\Lambda|\log\left(\frac{1}{\eps}\right) + \log\left(\frac{1}{\eps}\right)\right).
    \end{align*}
    Moreover, the expected runtime for the hard sphere model has the same upper bound, while for general repulsive pair potentials it is bounded by the square of the above.
\end{restatable}

\begin{proof}
    The result follows from \Cref{cor:mixing:sampling_algorithm_runtime} and \Cref{thm:runtime:hard_spheres} together with the spectral gap lower bound for $\lambda<\lambda_{\si}$ (\Cref{thm:spectral_gap}).
\end{proof}

Lastly, in \Cref{sec:fixed_size_sampling} we give an efficient sampling algorithm for the canonical ensemble, improving the density up to which the latter can be efficiently sampled by a factor of nearly $2$; the previous best result for sampling fixed-size sphere packings was by \cite{KMM03}. We state the result here for the hard spheres model. Our proof works more generally for Gibbs point processes with repulsive pair potentials, though the bound on $k$ may be different.

\begin{restatable}[Canonical Model]{theorem}{thmcanonical} \label{thm:canonical_hard_sphere_density}
    For any $\epsilon\in(0,1)$ and $\delta \in (0, 1)$, there exists $d_0\in\mathbb{N}$ such that for $d\ge d_0$, there exists an algorithm that takes as input a bounded measurable $\Lambda\subseteq\R^d$ (given by its volume $|\Lambda|$ and an oracle for sampling uniformly from $\Lambda$) and $k\le (1-\epsilon)\frac{1}{2^d}|\Lambda|$ and outputs a sphere packing of size $k$ with centers in $\Lambda$ from a distribution within TV distance $\delta$ of the uniform distribution over sphere packings of size $k$ of $\Lambda$. The runtime of the algorithm is bounded by a polynomial in $|\Lambda|$ and $\log(1/\delta)$.
\end{restatable}

In summary, the main contributions of this paper are four-fold:
\begin{enumerate}
     \item We extend the notions of Spectral Independence and Negative Fields Localization to the continuous setting and develop their basic properties.
    \item We prove that Continuuum Glauber mixes rapidly for Gibbs point processes with repulsive pair potentials with activity all the way up to Spectral Independence threshold.
    \item We widen the regime of activity for efficient sampling from Gibbs point processes with arbitrary-range repulsive pair potentials to match the known regime for finite-range repulsive pair potentials.
    \item We improve the density up to which the canonical ensemble can be efficiently sampled.
\end{enumerate}

We discuss these contributions and related aspects in more detail in the technical overview.

\paragraph{Subsequent work.} \cite{gobel2026uniqueness} substantially improves the range of activity in which the spectral gap of Continuum Glauber dynamics on the hard-sphere model is bounded from below. 
However, their improved range is specific to the hard-sphere potential, and for general repulsive pair potentials their result is incomparable with ours.

\subsection{Preliminaries}\label{subsec: prelim}

Here we introduce basic definitions and notation needed for the technical overview in the next section.

We fix $\mathbb{X}$ to be a complete separable metric space, equipped with a locally-finite reference measure on the Borel sets of $\mathbb{X}$ with respect to the metric. We assume $\mathbb{X}$ satisfies Assumption 1 of \cite{MP22Analyticity}, that is, for each $x\in\mathbb{X}$, the pushforward of the reference measure by the map $y\mapsto d(x,y)$ is absolutely continuous wrt the Lebesgue measure on $[0,\infty)$. For example, we can let $\mathbb{X}=\R^d$ be equipped with the Euclidean metric and Lebesgue measure.

We will let $\Lambda\subseteq\mathbb{X}$ have finite measure.
We use $|\Lambda|$ to denote the volume of $\Lambda$.
We use the notational convention that $\lambda\in\R_{\ge 0}$ while $\bm\lambda:\Lambda\to\R_{\ge 0}$. For a measurable region $B\subseteq \Lambda$, we define $\bm\lambda(B) \defeq \int_B \bm\lambda(x) dx$.
We will let $\Omega\subseteq 2^{\Lambda}$ be the collection of all finite subsets of $\Lambda$. For hard spheres this is the collection of sets of sphere centers. 
For $\eta\in\Omega$ we use $|\eta|:=\eta(\Lambda)$ to denote the cardinality of $\eta$. Any $\eta\in\Omega$ can also be considered a $\Z_{\ge 0}$-valued measure on $\Lambda$.
For measurable $B\subseteq\Lambda$, we write $\eta(B)=|\eta\cap B|$. For a function $f:\Lambda\to\R$, we use the notation $\eta(f)$ to denote $\int fd\eta=\sum_{x\in\eta}f(x)$.

We equip $\Omega$ with the $\sigma$-algebra generated by sets of the form
    \begin{align*}
        \{\eta\in\Omega:\eta(B)=k\}&&\text{ for measurable }B\in\Lambda\text{ and integer }k\ge 0
    \end{align*}
    In other words, it is the smallest $\sigma$-algebra such that for any measurable $B\subseteq\Lambda$, the $\Omega\to\R$ map
        $\eta\mapsto \eta(B)$
    is measurable.

\begin{definition}[Poisson point process]
The Poisson point process of intensity $\bm\lambda:\Lambda\to\R_{\ge 0}$ is a probability distribution $\rho_{\bm\lambda}$ on $\Omega$ such that $\eta\sim\rho_{\bm\lambda}$ satisfies the following: for any measurable $B\subseteq\Lambda$, $\eta(B)\sim\Pois(\bm\lambda(B))$, and for disjoint $B_1,\ldots,B_k$, $\eta(B_1),\ldots,\eta(B_k)$ are mutually independent.
\end{definition}

\begin{definition}[Hard Sphere Model]\label{def: hard sphere}
    The hard sphere model on a bounded measurable set $\Lambda\subseteq\R^d$ at activity $\lambda\ge 0$ is defined by conditioning a Poisson point process on $\Lambda$ of intensity $\lambda$ on the event that the distance between any two points is at least $2r$, i.e. that the chosen points are centers of a valid sphere packing with radius $r$.
    We can also think of it in terms of its partition function
    \begin{align*}
    Z_{\Lambda}(\lambda)&\defeq \sum_{k\ge 0}\dfrac{\lambda^k}{k!}\int_{\Lambda^k}\prod_{1\le i<j\le k}1_{\|x_i-x_j\|\ge 2r}dx_1\ldots dx_k.
\end{align*}\end{definition}

If a pair potential $\phi$ is repulsive and has finite range $r$, then $0\le C_{\phi}\le \vol(B_d(0, r))$, where $\vol(B_d(0, r))$ is the volume of a $d$-dimensional ball of radius $r$.  For hard spheres with radius $r$ (and thus range of $2r$), $C_{\phi} = \vol(B_d(0, 2r))$.

It was shown in \cite{BCC02} that one can define a Markov semigroup from the following generator for Continuum Glauber.

\begin{definition}[Continuum Glauber for Gibbs point processes]
The generator $\mcL$ for Continuum Glauber is defined as follows: 
for $f:\Omega\to\R$ in a dense subset of $L^2(\mu)$ known as the domain of the generator,

    \begin{align*}
    \mcL f(\eta)\defeq \sum_{x\in \eta}(f(\eta\setminus\{x\})-f(\eta))+\int_{\Lambda}e^{-\nabla_x^+H(\eta)}(f(\eta\cup\{x\})-f(\eta))\bm\lambda(x) dx 
\end{align*}
  where
    \begin{align*}    \nabla_x^+H(\eta)&\defeq H(\eta\cup\{x\})-H(\eta)\, .
    \end{align*}
\end{definition}

\begin{definition}[Continuum Glauber for hard spheres]\label{def:cont_glauber}
    For the hard sphere model with activity $\lambda$ on $\Lambda\subseteq\R^d$, the infinitesimal generator of this process is defined by
    \begin{align*}
    \mathcal{L}f(\eta)\defeq\sum_{x\in\eta}\nabla_x^-f(\eta)+\lambda \int_{\Lambda\setminus \bigcup_{x\in\eta}B_{2r}(x)}\nabla_x^+f(\eta)dx 
    \end{align*}
    where 
    \begin{align*}
    \nabla_x^-f(\eta)&\defeq f(\eta\setminus\{x\})-f(\eta) \, .
    \end{align*}
\end{definition}

We find it more convenient to instead define continuum Glauber as a Markov jump-type process (\Cref{sec:mixing_time}). This avoids the ambiguity on measure zero sets that arises from using $L^2(\mu)$, which will be important for \Cref{subsec: burn-in}.

We will need to condition Gibbs point processes on some finite subset of points being included. While this is relatively straightforward in the discrete setting, in the continuous setting, we would be conditioning on a probability zero event. To handle this, we use a notion of {\em pinnings} from \cite{MP22Strong} (defined there in the context of boundary conditions). The next few definitions help us modify point process distributions (including pinnings and tilts), and are needed for the localization scheme definition later.

\begin{definition}[Addition pushforward]
Denote $\mu_{\bm\lambda}^{+A}$ to be the pushforward of $\mu_{\bm\lambda}$ by $\eta\mapsto\eta\cup A$. In other words, to draw a sample from $\mu_{\bm\lambda}^{+A}$, we can draw $\eta\sim\mu_{\bm\lambda}$ and output $\eta\cup A$. We define $\mcP(\Lambda)$ to be the set of $\mu_{\bm\lambda}^{+A}$ for all $A\in\Omega$ and $\bm\lambda\ge 0$.
\end{definition}

\begin{definition}[Addition operation]
    For $A\in\Omega$, we define the \textit{addition operation} $\mcI_A:\mcP(\Lambda)\to\mcP(\Lambda)$ by
    \begin{align*} \mcI_{A}\mu_{\bm\lambda}^{+S}&:=\mu_{\bm\lambda}^{+S\cup A}&\text{for any }\mu_{\bm\lambda}^{+S}\in \mcP(\Lambda)
    \end{align*}
\end{definition}

\begin{definition}[Pinning]
For $A\in\Omega$, we define the \textit{pinning operation} $\mcR_A:\mcP(\Lambda)\to\mcP(\Lambda)$ by
\begin{align*}              \mcR_{A}\mu_{\bm\lambda}^{+S}&:=\mu_{\bm\lambda\cdot \exp\left(-\sum_{x\in A\setminus S}\phi(x,\cdot)\right)}^{+S\cup A}&\text{for any }\mu_{\bm\lambda}^{+S}\in \mcP(\Lambda)
\end{align*}
\end{definition}

\begin{definition}[Negative Tilt]\label{def:tilt}
For $t \ge 0$, 
we define the \textit{negative tilt} $\mcT_{-t}:\mcP(\Lambda)\to\mcP(\Lambda)$ by
\begin{align*}
\mcT_{-t}\mu&\defeq g_{t,\mu} \mu&\text{ where }g_{t,\mu}(\eta) \defeq \frac{e^{-t|\eta|}}{\int e^{-t|\eta|} d\mu(\eta) }\text{ for }\eta\in\Omega
\end{align*}
\end{definition}

In \Cref{lemma: another def of tilts}, we will see that negative tilts for Gibbs point processes satisfy $\mcT_{-t}\mu_{\bm\lambda} = \mu_{e^-t\bm\lambda}$. 

Finally, we define spectral gap. 

\begin{definition}[Spectral Gap]
The spectral gap for continuum Glauber on $\mu$ is the largest $\delta$ for which
\begin{align*}
    \mcE_{\mu}(f,f)\ge \delta\Var_{\mu}(f)
\end{align*}
for all bounded measurable $f$ with $\Var_\mu(f)<\infty$, 
where
\begin{align*}
    \mcE_{\mu}(f,g):=\int \sum_{x\in\eta}(f(\eta\setminus\{x\})-f(\eta))(g(\eta\setminus\{x\})-g(\eta))d\mu(\eta)
\end{align*}
is the Dirichlet form for continuum Glauber and
\begin{align*}
    \Var_{\mu}(f):=\int f^2d\mu-\left(\int fd\mu\right)^2
\end{align*}
is the variance.
\end{definition}

\subsection{Technical overview}
The main goal of this paper is to prove a bound on the spectral gap of Continuum Glauber dynamics and thereby derive rapid mixing. To do this, we start with a known bound on the spectral gap for a lower value of activity $\lambda$ \cite{KL03} and run a {\em localization scheme}, a general approach introduced by Chen and Eldan~\cite{CE22}. This transforms the distribution at high activity to a random distribution at low activity, and can be used to convert a spectral gap at low activity to one at high activity. Our goal is to show that this transformation approximately preserves {\em variance.} To do this, we use Spectral Independence (SI), a tool that has proven to be very effective in discrete settings. We show that Bounded Complex Density implies SI and also that Spatial Mixing implies SI (the latter for finite-range potentials), and then use SI to show approximate variance conservation.  We extend the definitions of both SI and the localization scheme we use to continuous domains and establish basic properties. We explain the structure of the proof in more detail below.

\paragraph{The discrete setting.}
The discrete analog of the hard sphere model is the {\em hard core} model, in which every independent set on a given graph has weight proportional to $\lambda$ to the power of its size. 
A key tool used to analyze hard core and other discrete models in recent work is {\em Spectral Independence} (SI)~\cite{anari2021spectral, CLV23,CE22, alev2020improved, 10.1145/3531008}.
Roughly speaking, SI is a property of a probability distribution, which quantifies influence of one particle on the others ``on average''. 
We recall the definition of the influence matrix for the hard core model on a graph.

\begin{definition}[Influence Matrix]
    Let $\mu$ be the distribution of hard core model over a graph with $n$ vertices, and $S$ be an independent set randomly sampled from $\mu$. Then, the influence matrix of $\mu$ is 
    \[ M(i, j) = \begin{cases}
    \P[j \in S~|~i \in S] - \P[j \in S~|~ i \not \in S] \quad \mbox{ for } i\neq j\\
    0 \quad \mbox{ for } i=j.
    \end{cases}
    \]
\end{definition}

\begin{definition}[Spectral independence (discrete)]
    A distribution $\mu$ is said to be $\kappa$-spectrally independent if its influence matrix $M$ (under any pinning) satisfies $\rho(M)\le \kappa$, where $\rho(\cdot)$ denotes the spectral radius.
    \footnote{There are several slightly different definitions of spectral independence in the literature. We find this one to be the most convenient for our purposes.}
\end{definition}

SI has been used to prove a variety of results, many of which are optimal or best known. 
\cite{anari2021spectral} gave polynomial-time approximate sampling and counting algorithms for the hard core model on graphs with maximum degree $\Delta$ with activity up to the tree-uniqueness threshold $\lambda_c(\Delta)$ (which is best possible, since there is no sampling algorithm above the threshold unless $NP= RP$ as established by a series of works \cite{SlyUniqueness2010, Sly_Nike, galanis2011}). This was improved to $O_\delta(n\log{n})$ mixing by ~\cite{CLV20}. 
The work of  \cite{CLV23} improves on the mixing time by simplifying the definition of spectral independence and generalizing its application to all 2-spin models. 
Spectral independence is also one of the main ingredients in results for sampling colorings \cite{CGSV2021}, spin glasses \cite{AJKPV2024}, 
sampling hard core model at the uniqueness threshold \cite{CCYZ2025}, sampling antiferromagnetic two-spin systems with no maximum degree assumption \cite{chen2024rapid} and more. It is also used in ``canonical'' sampling, for example, sampling independent sets with fixed size \cite{JMPV23}, and sampling from ferromagnetic Ising model with fixed magnetization \cite{KPPY2024}. 

We remark that while it is tempting to pass from the hard core model to $\R^d$ via a limit argument, this approach runs into technical difficulties. Recent work by \cite{FGKKP22,friedrich2022using} gives more careful randomized reductions
to obtain polynomial-time sampling algorithms up to $\lambda<e/C_{\phi}$. Note that for a target sampling error of $\varepsilon$ (say in TV distance), the complexity of the algorithm of \cite{friedrich2022using} grows as $\poly(1/\epsilon)$ (Thm 1.3) while, as we will see, Continuum Glauber has complexity growing with $\log(1/\varepsilon)$.
We avoid passing to the discrete setting and work directly with the natural continuous counterparts of the above and other relevant notions.

\paragraph{Spectral Independence in Continuous Space.}
To the best of our knowledge, SI has so far only been defined and used in the setting of discrete spaces.
In this paper, we provide rigorous extensions of the definitions and analysis techniques of spectral independence to continuous space. We hope that these provide a useful framework for proving rapid mixing results in continuous settings. 
To define spectral independence, we first need a generalization of the influence matrix.

\begin{definition}[Influence Operator]
    For a point process $\mu$, define its influence operator to be
    \begin{align*}
        \Psi f(x)&:=\E_{\eta\sim\mcR_{+x}\mu}[\eta(f)]-\E_{\eta\sim\mu}[\eta(f)]&f:\Lambda\to\R, x\in\Lambda
    \end{align*}
    where $\eta(f) \defeq \int fd\eta=\sum_{x\in\eta}f(x)$.
\end{definition}

\begin{definition}[Spectral independence (continuous)]
    A Gibbs distribution $\mu_{\bm\lambda}$ is said to be $\kappa$-spectrally independent if its influence operator satisfies $\rho(\Psi_{\bm\lambda'})\le \kappa$ for any $0\le \bm\lambda'\le \bm\lambda$, where $\rho(\cdot)$ denotes the spectral radius.
\end{definition}

\paragraph{Bounded Complex Density.}

We will be using a bound on the one-point density $\zeta_{\bm\lambda}(x)$ (\Cref{def:one_point_density}) for complex activities $\bm\lambda:\Lambda\to\C$.
It implies that the partition function of the Gibbs point process has no zeros in the region $\mathbb{D}([0,\lambda_{0}],\epsilon)$ in the complex plane~\cite[Lemma 33]{MP22Analyticity}.

\begin{definition}[Bounded Complex Density]\label{def:bcd} We say that the family of Gibbs point processes on $\mathbb{X}$ with pair potential $\phi:\mathbb{X}\times\mathbb{X}\to[0,\infty]$ has bounded complex density up to activity $\lambda_0>0$ if there exists $\epsilon>0$ and $C>0$ such that for $\Lambda\subseteq\mathbb{X}$ with $|\Lambda|<\infty$ and $\bm\lambda:\Lambda\to\mathbb{D}([0,\lambda_{0}],\epsilon)$,
\begin{align*}
|\zeta_{\bm\lambda}(x)|&\le C|\bm\lambda(x)|&\forall x\in\mathbb{X}
\end{align*}
\end{definition}

This property implies Spectral Independence. The proof is in \Cref{subsec:si_from_bcd}.

\begin{lemma}\label{thm:zero_freeness_spectral_independence}
Suppose the family of Gibbs point processes on $\mathbb{X}$ with pair potential $\phi:\mathbb{X}\times\mathbb{X}\to[0,\infty]$ has bounded complex density up to activity $\lambda_0>0$.
Then, for all $\Lambda\subseteq\mathbb{X}$ with $|\Lambda|<\infty$ and $\bm\lambda:\Lambda\to[0,\lambda_0]$, the spectral radius of the operator $\Psi_{\bm\lambda}-\mathrm{Id}$ is bounded:
\begin{align*}
\rho(\Psi_{\bm\lambda}-\mathrm{Id})\le \frac{C}{\epsilon}e^{\lambda C_{\phi}}+1.
\end{align*}
Thus, $\mu_{\lambda_0}$ is $\frac{C}{\epsilon}e^{\lambda C_{\phi}}+2$-spectrally independent.
\end{lemma}

\paragraph{Strong Spatial Mixing.} SSM is a property of a probability distribution, which quantifies the correlation decay.  
Intuitively, it says that the effect of the presence of a point on the distribution of distant points decays exponentially in the distance.

For $d\ge 2$ and $0<\lambda<2^{1-d}$, Helmuth, Perkins and Petti \cite{HPP21} showed that the single-center dynamics for the hard sphere model on $\R^d$ exhibits optimal temporal mixing
which implies the hard sphere model exhibits Strong Spatial Mixing. 
In that direction, Michelen and Perkins proved Strong Spatial Mixing for $\lambda<e/\Delta_{\phi}$, which implies that block dynamics on $\Lambda_n=[-n,n]^d$ mixes in $O(|\Lambda_n|\log(|\Lambda_n|/\epsilon))$ steps where $|\Lambda_n|$ is the volume of box domain \cite{MP22Strong}.
This also gives a $\tilde{O}(|\Lambda_n|^3)$ time $(1\pm\epsilon)$-approximation for partition functions of the repulsive Gibbs point processes, assuming unit cost for a single radius-$L$ block dynamics update.

\begin{definition}[Strong Spatial Mixing, \cite{MP22Strong}]\label{def:ssm}

The family of point processes on $\R^d$ defined by a repulsive pair potential $\phi$ exhibits strong spatial mixing with activities bounded by $\lambda>0$ if there exist constants $\alpha,\beta$ so that the following holds: for any bounded, measurable region $B\subseteq \R^d$ and any two activity functions $\bm{\lambda},\bm{\lambda'}:B\to[0,\lambda]$,
\begin{align*}
\|\mu_{\bm{\lambda}}-\mu_{\bm{\lambda}'}\|_{B}\le \alpha|B|e^{-\beta\cdot\dist(B,\supp(\bm{\lambda}-\bm{\lambda'}))},
\end{align*}
here $\text{supp}(f)$ denotes the support of function $f$.
\end{definition}

Let us define $\lambda_{SSM}$ to be the maximum value such that for all $\lambda<\lambda_{SSM}$, Strong Spatial Mixing holds. It is shown in \cite{MP22Strong} that $\lambda_{SSM}\ge e/\Delta_{\phi}$, where $\Delta_{\phi}$ is called the potential-weighted connective constant. It is at most the temperedness constant, i.e., $\Delta_\phi \le C_\phi$. There is no known upper bound on $\lambda_{SSM}$; indeed, we do not even know if it is finite.
We show that SSM implies Spectral Independence. 

\begin{restatable}
[SSM implies SI]{lemma}{lemmaSSMSI}\label{lemma: SSM implies SI}

Suppose the family of Gibbs point process on $\R^d$ with repulsive pair potential $\phi$ (with finite range $r$) exhibits strong spatial mixing with activities bounded by $\lambda$ (with constants $\alpha,\beta$ in \Cref{def:ssm}). Then, if $\Psi_{\bm\lambda}$ is the influence operator of a repulsive point process at activity $\bm\lambda:\Lambda\to[0,\lambda]$,
\begin{align*}             
    \sup_{\|f\|_{\infty}=1}\|(\Psi_{\bm\lambda}-\mathrm{Id})f\|_{\infty}\le \alpha \int_{\R^d} e^{-\beta\cdot\max(0,\dist(0,y)-r)}dy
\end{align*}
Thus, $\mu_{\bm\lambda}$ is $\kappa$-spectrally independent for $\kappa=1+\alpha \int_{\R^d} e^{-\beta\cdot\max(0,\dist(0,y)-r)}dy$.

\end{restatable}

The constants $\alpha$ and $\beta$ derived in \cite{MP22Strong} depend on the activity $\lambda$ and range of particle interactions. See Section \ref{sec: SSM implies Influence Operator bounds} for details of the proof of \Cref{lemma: SSM implies SI}

\begin{figure}[!htbp]
\fbox{\parbox{\textwidth}{
\begin{center}
\begin{tikzpicture}[scale=0.5]
\node (BCD) at (-8,12) {Bounded Complex Density (Thm.~\ref{thm:zero_freeness})};
\node (SSM) at (8,12) {Strong Spatial Mixing (Thm.~\ref{thm: SSM})};
\node (SI) at (0,8) {Spectral Independence (Lem.~\ref{thm:zero_freeness_spectral_independence}, Lem.~\ref{lemma: SSM implies SI})};
\node (ACV) at (0,4) {\begin{tabular}{c}Approximate conservation of variance of\\ negative-fields localization process (Thm. \ref{thm:variance_conservation}) \end{tabular}};
\node (B) at (0,0) {};
\node (GL) at (-10,0) {Spectral gap for $\lambda<\frac{1}{C_{\phi}}$ (Thm. \ref{thm:low_activity_spectral_gap})};
\node (GH) at (10,0) {Spectral gap for $\lambda<\frac{e}{\Delta_{\phi}}$ (Thm. \ref{thm:spectral_gap})};
\draw[->,ultra thick] (BCD) -> (SI); 
\draw[->,ultra thick] (SSM) -> (SI);
\draw[->,ultra thick] (SI) -> (ACV);
\draw[->,ultra thick] (ACV) -> (B);
\draw[->,ultra thick] (GL) -> (GH);
\end{tikzpicture}
\end{center}
}
}
\caption{Proof outline}
\label{fig: Proof overview}
\end{figure}
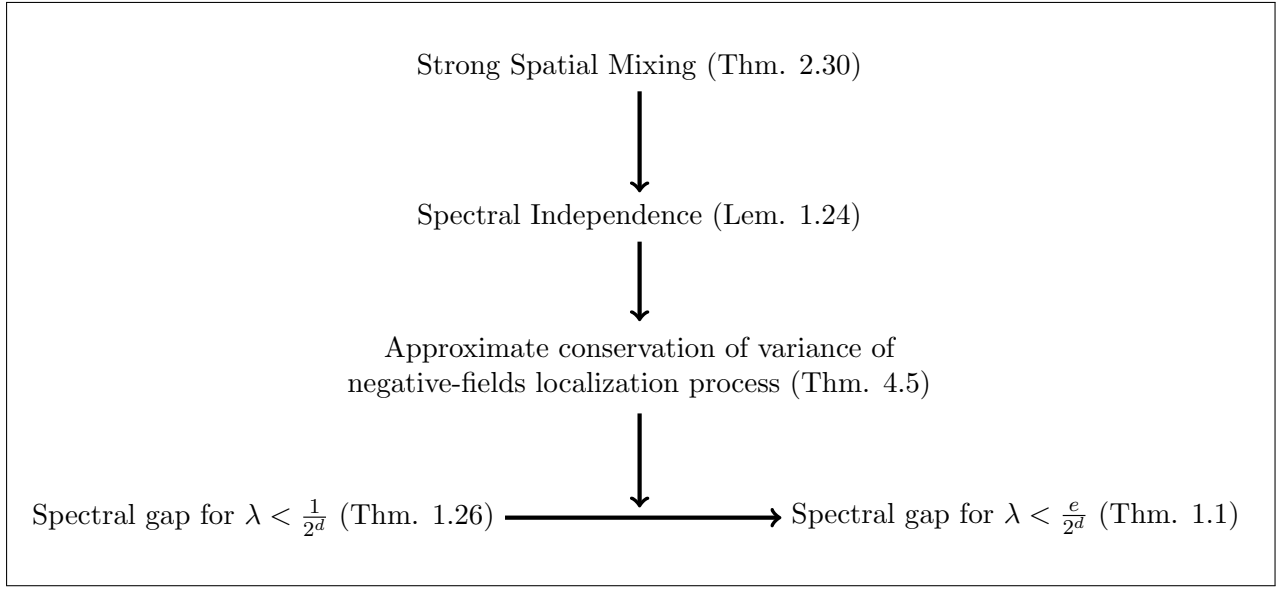

\paragraph{Localization Schemes.}
This technique was introduced by Chen and Eldan \cite{CE22} as a general approach to analyzing Markov chains. 
A localization scheme uses a martingale process (typically {\em stochastic localization}) to transform a
distribution into a family of distributions
that are easier to reason about while still implying interesting consequences for the original distribution. We extend the {\em negative-fields} localization scheme previously used to prove results for Glauber dynamics in the hard core model \cite{CE22} to the continuous setting. Its construction and properties are developed in Section \ref{sec: loc scheme}. 

Specifically, using this localization scheme we will show approximate variance conservation. The proof uses SI.

\begin{restatable}{theorem}{measuredecomposition}\label{thm:measure_decomposition}
Let $\lambda_0>\lambda_1>0$. Let $\nu_0=\mu_{\lambda_0}$ be $\kappa$-spectrally independent.
Then, there exists a random $A\in\Omega$ so that letting $\nu_{A}=\mcR_{A}\mu_{\lambda_1}$, we have for all measurable $\varphi:\Omega\to[-1,1]$,
\begin{align*}
    \E[\nu_A]=\nu_0 &\quad \mbox{ and } \quad 
    \E[\Var_{\nu_{A}}[\varphi]]
    \ge \left(\dfrac{\lambda_1}{\lambda_0}\right)^{\kappa}\Var_{\nu_0}[\varphi]
\end{align*}
\end{restatable}

At lower activity, we can use a spectral gap result of \cite{KL03} (c.f. \cite{DP13}). 
While they state their result for $\R^d$, their proof works essentially unchanged for a finite measure space $\Lambda\subseteq\mathbb{X}$. 
See \Cref{sec: appendix low activity} for details.

\begin{restatable}{theorem}{lowactivityspectralgap}\label{thm:low_activity_spectral_gap}
Let $\phi:\mathbb{X}\times\mathbb{X}\to[0,\infty]$ be a repulsive pair potential, $\Lambda\subseteq\mathbb{X}$ have finite measure, and $\lambda=(1-\delta)\frac{1}{C_{\phi}}$. Then, for $\bm\lambda:\Lambda\to[0,\lambda]$, the spectral gap of continuum Glauber for $\mu=\mu_{\bm\lambda}$ is at least $\delta$.
\end{restatable}

By combining the above lemmas, we obtain a spectral gap lower bound for repulsive point processes (including the hard sphere model) with activity in the regime of Bounded Complex Density. With the additional assumption of finite-range pair potentials, we also obtain this in the regime of Strong Spatial Mixing. 
This implies a lower bound on the spectral gap for continuum Glauber when $\lambda< \max(\lambda_{\bcd}, \lambda_{SSM})$ (Section~\ref{sec: spec gap at high}).
The overall proof structure is summarized in Figure~\ref{fig: Proof overview}. 
With the spectral gap lower bound in hand, we obtain \Cref{thm: main mixing}, which implies that continuum Glauber mixes in $O_{d,\lambda}\left(|\Lambda|+\log\left(\frac{1}{\eps}\right)\right)$ starting from the empty set for $\lambda< \max(\lambda_{\bcd}, \lambda_{SSM})$.

\paragraph{Mixing time.}
The mixing time is usually bounded using the spectral gap by showing that the variance decays exponentially. This works if the initial variance (of the relative density of the initial distribution with respect to the stationary distribution) is finite. This condition holds for the empty configuration, but not for other fixed configurations, as the point mass on a nonempty configuration is not absolutely continuous with respect to the stationary distribution. 
For the purpose of a sampling algorithm, one could start with the empty configuration, but bounding the mixing time from an arbitrary starting configuration is of independent interest.

To analyze mixing from an arbitrary initial configuration $S$, we give a ``burn-in'' argument that shows that after $\Theta(\log|S|)$ time, the distribution will be close to one whose relative density with respect to the stationary distribution is bounded.
The intuitive idea is: once the initial particles die out, no configuration will be too likely compared to the stationary distribution. Each initial particle will live for $\Exp(1)$ time, so after waiting $\Theta(\log|S|)$ time, the initial points will likely have all died.
However, during this time some particles may be born, and thus to show that the relative density is bounded we have to more carefully track how the distribution evolves over time. Note that we cannot hope for a mixing time bound that is independent of $S$, as we cannot be close to the stationary distribution (in TV distance) until it is likely that all the initial points have died, which takes roughly $\log|S|$ time.

We can couple a run of the algorithm with a sample path of continuum Glauber to translate mixing of continuum Glauber after $T$ time to a sampling guarantee of the algorithm for the same $T$.
To bound the runtime of the algorithm, we observe that the number of attempted births is distributed according to $\Pois(\lambda|\Lambda|T)$, and the number of deaths is bounded by the number of births.

\paragraph{Sampling from the canonical model.} 
To sample from the canonical model, i.e. the Gibbs point process conditioned on having exactly $k$ points, we use rejection sampling, similar to the discrete setting~\cite{davies2023approximately}.
For $k$ smaller than the expected number of points in a sample from the (grand canonical) Gibbs point process of activity $\bm\lambda$, we can show there exists a scaling $t_k\in[0,1]$ for which the Gibbs point process with activity $t_k\bm\lambda$ has probability $\Omega(\frac{1}{\bm\lambda(\Lambda)^2})$ of getting exactly $k$ points. (The proof of this, inspired by \cite{davies2023approximately}, uses the spectral gap of continuum Glauber). We can then sample from this distribution repeatedly until we get a sample with exactly $k$ points. Since we do not know a simple way to compute $t_k$, we instead sample from the Gibbs point process with activity $t\bm\lambda$ for different values of $t$ until we get a sample of exactly $k$ points. This yields a polynomial-time (approximate) sampling algorithm for small enough $k$. The range of $k$ for which this works depends on the expected size of a sample from the grand canonical model. In the case of the hard sphere model in $\R^d$, this approach provably works for $k\le (1-o_d(1))|\Lambda|/2^d$. 
This improves on the previous best bound by nearly a factor of $2$. (The previous bound was $k\le |\Lambda|/2^{d+1}$~\cite{KMM03}).

\paragraph{Technical challenges.} 
To conclude this overview, we remark that substantial parts of our analysis are included to ensure that our proofs are rigorous and complete. In many cases, 
we try to provide general properties and lemmas in the hope that they will be useful in the future. A number of technical issues arise when going to the continuous setting. For example, the definition of spectral independence involves pinnings (conditioning that a set of points is included); however, including a specific point in space is a zero-probability event and this needs care.
The influence matrix is now an operator, and in order to bound its Rayleigh quotient by its $\ell_{\infty}$-operator norm, we need to show that the influence operator is self-adjoint (which is straightforward in the discrete case). We also have to extend the negative-fields localization scheme introduced by Chen and Eldan \cite{CE22} to the continuous setting, which again needs extra care for events of probability zero.  

\paragraph{Acknowledgments. }
We are deeply grateful to Will Perkins for many valuable remarks and pointers. We also thank Mark Jerrum, Marcus Michelen, Zongchen Chen and June Vuong for helpful discussions and anonymous reviewers for useful comments. This work is supported in part by NSF Awards CCF-2106444 and CCF-2504994, and a Simons Investigator Award. 

\section{Influence Operator}
In this section, we introduce the influence operator, a continuous analog of the influence matrix (Sec. \ref{subsec: properties of Influence}) and develop its basic properties.  Before we introduce the operator, we need a few definitions.

\begin{definition}[Intensity]\label{def: intensity (no addition)}
    The \textit{intensity} of $\mu_{\bm\lambda}$ is defined as 
    \begin{equation*}
        \iota_{\bm\lambda}(B)\defeq\E_{\eta\sim\mu_{\bm\lambda}}[\eta(B)] \qquad \text{for all measurable } B\subseteq \Lambda
    \end{equation*}
\end{definition}

The intensity turns out to be absolutely continuous with respect to the reference measure on $\Lambda$, and its Radon-Nikodym derivative is given by the one-point density.

\begin{definition}[One-point density]\label{def:one_point_density}
      The one-point density $\zeta_{\bm\lambda}$ of $\mu_{\bm\lambda}$ is given by
        \begin{align*}
            \zeta_{\bm\lambda}(x)=\bm\lambda(x)\dfrac{Z_{\Lambda}(\bm\lambda e^{-\phi(x,\cdot)})}{Z_{\Lambda}(\bm\lambda)}
        \end{align*}
        where $\bm\lambda e^{-\phi(x,\cdot)}$ denotes the function $y\mapsto \bm\lambda(y)e^{-\phi(x,y)}$.
\end{definition}

\begin{lemma}[Intensity vs One-point density]\label{thm: intensity vs one point}
    \[ \iota_{\bm \lambda}(B) = \int_{B} \zeta_{\bm\lambda}(x) dx\]
\end{lemma}
This is restated equivalently as \Cref{lem:intensity_vs_one_point_density} and proved there.
Intuitively, the intensity of $B$ measures how many points we expect to appear in $B$. In Section~\ref{sec: loc scheme} we will extend this to sets and also allow for pinnings.

\subsection{Influence Operator and its properties}\label{subsec: properties of Influence}

We will use $\Psi_{\bm\lambda}$ to denote the influence operator of $\mu_{\bm\lambda}$.

\begin{definition}[Inner product]
    \begin{align*}
        \langle f,g\rangle_{\iota}\defeq\int_{\Lambda} f(x)g(x)d\iota(x)
    \end{align*}
\end{definition}

\begin{definition}
    Let $L^2(\iota)$ denote the space of functions $f:\Lambda\to\R$ such that
    \begin{align*}
        \langle f,f\rangle_{\iota}<\infty
    \end{align*}
    where functions equal a.e. $\iota$ are identified.
\end{definition}

\begin{fact}[\cite{Rudin}, p. 77]
    $L^2(\iota)$ is a Hilbert space.
\end{fact}

\begin{definition}
The $\iota$ norm is 
\begin{align*}
    \| f\|_\iota &\defeq \sqrt{\int f^2(x)d\iota (x)}
\end{align*}
\end{definition}

\begin{definition}
    The $\iota$ operator norm is
    \begin{align*}
        \|A\|_{\iota\to\iota} &\defeq \sup_{\|f\|_\iota = 1} \| Af\|_\iota
    \end{align*}
\end{definition}

\begin{observation}
    The influence operator is symmetric with respect to the inner product, i.e.
    \begin{align*}
        \langle f,\Psi g\rangle_{\iota}&=\langle \Psi f,g\rangle_{\iota}.
    \end{align*}
\end{observation}

\begin{lemma}[Influence and Covariance]\label{lem:influence_covariance}
    \begin{align*}
        \langle f,\Psi g\rangle_{\iota}&=\E_{\eta\sim\mu}[\eta(f)\eta(g)]-\iota(f)\iota(g)\\
        &=\E_{\eta\sim\mu}[(\eta(f)-\iota(f))(\eta(g)-\iota(g))]
    \end{align*}
\end{lemma}
\begin{proof}
\begin{align*}
    \langle f,\Psi g\rangle_{\iota}&=\int_{\Lambda}f(x)\E_{\eta\sim\mcR_{+x}\mu}[\eta(g)]d\iota(x)-\int f(x)\iota(g)d\iota(x)~~~\text{(using \Cref{lem:integrate_pinning})}\\
    &=\E_{\eta\sim\mu}[\eta(f)\eta(g)]-\iota(f)\iota(g)&
    \\
    &=\E_{\eta\sim\mu}[(\eta(f)-\iota(f))(\eta(g)-\iota(g))].
\end{align*}
\end{proof}

\subsubsection{Influence Operator is an integral operator}

In this subsection, we will show that $\Psi-\mathrm{Id}:L^2(\iota)\to L^2(\iota)$ is an integral operator. 

\begin{theorem}\label{thm:influence_integral}
    There exists $k:\Lambda\times\Lambda\to\R$ 
    satisfying $|k(x,y)|\le e^{\bm\lambda(\Lambda)}$ such that
\begin{align*}
    (\Psi-\mathrm{Id})f(x)=\int_{\Lambda}k(x,y)f(y)d\iota(y).
\end{align*}
\end{theorem}

\begin{proof}
First, let us rewrite the left hand side in terms of one-point densities:
\begin{align*}
    (\Psi-\mathrm{Id})f(x)&=\E_{\eta\sim\mcR_{+x}\mu}[\eta(f)]-f(x)-\E_{\eta\sim\mu}[\eta(f)]\\
    &=\iota_{e^{-\phi(x,\cdot)}\bm\lambda}(f)-\iota_{\bm\lambda}(f)\\
    &=\int_{\Lambda}(\zeta_{e^{-\phi(x,\cdot)}\bm\lambda}(y)-\zeta_{\bm\lambda}(y))f(y)dy.
\end{align*}

The one-point densities satisfy, for all $y\in\Lambda$,
\begin{align*}
    \zeta_{\bm\lambda}(y)&=\dfrac{\bm\lambda(y)Z(e^{-\phi(y,\cdot)}\bm\lambda)}{Z(\bm\lambda)}&
    \zeta_{e^{-\phi(x,\cdot)}\bm\lambda}(y)&=\dfrac{e^{-\phi(x,y)}\bm\lambda(y)Z(e^{-\phi(x,\cdot)}e^{-\phi(y,\cdot)}\bm\lambda)}{Z(e^{-\phi(x,\cdot)}\bm\lambda)}
\end{align*}

Thus, when $\bm\lambda(y)>0$,
\begin{align*}
    \dfrac{\zeta_{e^{-\phi(x,\cdot)}\bm\lambda}(y)}{\zeta_{\bm\lambda}(y)}=e^{-\phi(x,y)}\cdot\dfrac{Z(e^{-\phi(x,\cdot)}e^{-\phi(y,\cdot)}\bm\lambda)Z(\bm\lambda)}{Z(e^{-\phi(x,\cdot)}\bm\lambda)Z(e^{-\phi(y,\cdot)}\bm\lambda)}.
\end{align*}
Hence,
\begin{align*}
    (\Psi-\mathrm{Id})f(x)&=\int_{\{y\in\Lambda:\bm\lambda(y)>0\}}\left(\frac{\zeta_{e^{-\phi(x,\cdot)}\bm\lambda}(y)}{\zeta_{\bm\lambda}(y)}-1\right)\zeta_{\bm\lambda}(y)f(y)dy\\
    &=\int_{\Lambda}k(x,y)f(y)d\iota_{\bm\lambda}(y)
\end{align*}
where
\begin{align*}
    k(x,y)&:=e^{-\phi(x,y)}\cdot\dfrac{Z(e^{-\phi(x,\cdot)}e^{-\phi(y,\cdot)}\bm\lambda)Z(\bm\lambda)}{Z(e^{-\phi(x,\cdot)}\bm\lambda)Z(e^{-\phi(y,\cdot)}\bm\lambda)}-1.
\end{align*}

Now, since $1\le Z(\bm\lambda)\le e^{\bm\lambda(\Lambda)}$ and $Z(\bm\lambda)\le Z(\bm\lambda')$ whenever $\bm\lambda\le \bm\lambda'$,
\begin{align*}
    -1\le k(x,y)\le e^{\bm\lambda(\Lambda)}
\end{align*}

\end{proof}

\subsubsection{Influence operator is compact and self-adjoint}

\begin{definition}[Compact Operator, \cite{Conway} (verbatim)]
    Let $B_H$ denote the closed unit ball in $H$. A linear transformation $T: H \to H$ is compact if $T(B_H)$ has compact closure in $H$. 
\end{definition}

\begin{theorem}[II.4.7 \cite{Conway}]\label{thm:integral_operators_are_compact}
If $(X,\Omega,\mu)$ is a measure space and $k\in L^2(X\times X,\Omega\times\Omega,\mu\times\mu)$, then
\begin{align*}
(Kf)(x)=\int k(x,y)f(y)d\mu(y)
\end{align*}
is a compact operator and $\|K\|_{\mu}\le \|k\|_2$.
\end{theorem}

\begin{corollary}\label{lemm: compact and bounded}
    $\Psi - \mathrm{Id}$ is a compact operator on $L^2(\iota)$, and
    \begin{align*}
        \|\Psi-\mathrm{Id}\|_{\iota\to\iota}\le e^{\bm\lambda(\Lambda)}\iota(\Lambda) .
    \end{align*}
\end{corollary}

\begin{proof}
    By \Cref{thm:influence_integral} and \Cref{thm:integral_operators_are_compact}, we see that $\Psi-\mathrm{Id}$ is a compact operator on $L^2(\iota)$ with
    \begin{align*}
        \|\Psi-\mathrm{Id}\|_{\iota\to\iota}\le e^{\bm\lambda(\Lambda)}\iota(\Lambda)
    \end{align*}

\end{proof}

\begin{definition}[Self-adjoint operator]
    Let $A$ be a bounded linear operator in the Hilbert space $H$. 
    We call $A^*$ \textit{adjoint} to $A$ if for all $ x \in H$
    \[ \langle Ax, y \rangle = \langle x, A^* y \rangle . \]
    An operator is called \textit{self-adjoint} if $A = A^*$. 
\end{definition}

\begin{lemma}\label{lemma: influence_self-adjoint}
     $\Psi - \mathrm{Id}$ is a self-adjoint operator on $L^2(\iota)$.
\end{lemma}

\begin{proof}
By \Cref{lemm: compact and bounded}, $\Psi-\mathrm{Id}$ is a bounded linear operator. 
    By \Cref{lem:influence_covariance}, $\langle \Psi f, g\rangle=\langle f,\Psi g \rangle$.

Thus,
\begin{align*}
    \langle (\Psi-\mathrm{Id})f,g\rangle&=\langle \Psi f, g\rangle-\langle f,g\rangle
    =\langle f,\Psi g \rangle-\langle f,g\rangle
    =\langle f,(\Psi-\mathrm{Id})g\rangle .
\end{align*}
\end{proof}

\subsection{Spectral radius of influence operator}

In this section, we show that the Rayleigh quotients are bounded by the spectral radius, which are in turn bounded by operator norms.

\begin{remark}
    Since the operator norm equals the spectral radius for a self-adjoint operator on a Hilbert space (\cite{reed1981functional} Theorem VI.6), we have $\|\Psi-\mathrm{Id}\|_{\iota\to\iota}=\rho(\Psi-\mathrm{Id})$ where $\rho(\cdot)$ denotes the spectral radius. Likewise, $\|\Psi\|_{\iota\to\iota}=\rho(\Psi)$.
\end{remark}

\begin{theorem}[II.5.1, 5.3 of \cite{Conway}]\label{thm:spectral_theorem_compact_self_adjoint}
If $T$ is a compact self-adjoint operator on a Hilbert space $\mcH$, then $T$ has a countable number of distinct eigenvalues $\{\lambda_1,\lambda_2,\ldots,\}$, each $\lambda_n$ is real, $\|T\|_{\iota \to \iota}=\sup\{|\lambda_n|:n\ge1\}$, and $\lambda_n\to 0$ as $n\to\infty$.
\end{theorem}

\begin{lemma}\label{lem:influence_rayleigh_quotient_le_iota_operator_norm}
\begin{align*}
\sup_{\langle f,f\rangle_{\iota}>0}\left|\dfrac{\langle f,\Psi f\rangle_{\iota}}{\langle f,f\rangle_{\iota}}\right|\le \|\Psi\|_{\iota\to\iota}=\rho(\Psi)
\end{align*}
\end{lemma}

\begin{proof}
\begin{align*}
\sup_{\langle f,f\rangle_{\iota}>0}\left|\dfrac{\langle f,\Psi f\rangle_{\iota}}{\langle f,f\rangle_{\iota}}
\right|&=\sup_{\langle f,f\rangle_{\iota}=1}|\langle f,\Psi f\rangle_{\iota}|\\
&\le \sup_{\langle f,f\rangle_{\iota}=1}\|f\|_{\iota}\|\Psi f\|_{\iota}&\text{Cauchy-Schwarz}\\
&\le \sup_{\langle f,f\rangle_{\iota}=1}\|f\|_{\iota}\|\Psi\|_{\iota\to\iota}\|f\|_{\iota}&\text{operator norm}\\
&=\|\Psi\|_{\iota\to\iota}
\end{align*}

\end{proof}

When we later show spectral independence, we will need to bound the spectral radius of the influence operator. We show here that this can be done by bounding an operator norm.

\begin{definition}
    The infinity operator norm is
    \begin{align*}
    \|A\|_{\infty\to\infty}:=\sup_{\|f\|_{\infty}=1}\|A f\|_{\infty}
    \end{align*}
\end{definition}
We also define a weighted infinity norm and its corresponding operator norm.
\begin{definition}
The weighted infinity norm is
    \begin{align*}
        \|f\|_{\bm\lambda\infty}:=\|\bm\lambda f\|_{\infty}
    \end{align*}
\end{definition}
\begin{definition}        
The weighted infinity operator norm is
    \begin{align*}
    \|A\|_{\bm\lambda\infty\to\bm\lambda\infty}:=\sup_{\|f\|_{\bm\lambda\infty}=1}\|A f\|_{\bm\lambda\infty}
    \end{align*}
\end{definition}

Weighted infinity operator norms have previously been used to show spectral independence in the discrete setting \cite{chen2021rapidmixingglauberdynamics}.

\begin{lemma}
    If $f$ is an eigenvector of $\Psi-\mathrm{Id}$ with eigenvalue $\lambda\ne 0$, then $\|f\|_{\infty}<\infty$.
\end{lemma}
\begin{proof}
    Suppose $(\Psi-\mathrm{Id})f=\lambda f$ for $f\in L^2(\iota)$ where $\lambda\ne 0$. Then, for almost all $x$ wrt $\iota$,
    \begin{align*}
        \lambda f(x)&=
        \int_{\Lambda}k(x,y)f(y)d\iota(y)&\text{by \Cref{thm:influence_integral}}\\
        |\lambda||f(x)|&\le \left|\int_{\Lambda}k(x,y)f(y)d\iota(y)\right|\\
        &\le \sqrt{\int_{\Lambda}k(x,y)^2d\iota(y)\int f(y)^2d\iota(y)}&\text{by Cauchy-Schwarz}\\
        &\le \sqrt{\iota(\Lambda)}e^{\bm\lambda(\Lambda)}\|f\|_{\iota}&\text{by \Cref{thm:influence_integral}}\\
        |f(x)|&\le \dfrac{1}{|\lambda|}\sqrt{\iota(\Lambda)}e^{\bm\lambda(\Lambda)}\|f\|_{\iota}
    \end{align*}

    In particular, if $\lambda\ne 0$, then $\|f\|_{\infty}<\infty$.
\end{proof}

\begin{lemma}\label{lem:influence_iota_operator_norm_le_infinity_operator_norm}
    $\rho(\Psi)\le 1+\|\Psi-\mathrm{Id}\|_{\infty\to\infty}$, and $\rho(\Psi)\le 1+\|\Psi-\mathrm{Id}\|_{\bm\lambda\infty\to\bm\lambda\infty}$ for bounded $\bm\lambda:\Lambda\to[0,\infty)$.
\end{lemma}

\begin{proof}
    Suppose $(\Psi-\mathrm{Id})f=\lambda f$ for nonzero $f\in L^2(\iota)$ and $\lambda\ne0$. Then, $\|f\|_{\infty}<\infty$, and
    \begin{align*}
        \|(\Psi-\mathrm{Id})f\|_{\infty}=\|\lambda f\|_{\infty}=|\lambda|\|f\|_{\infty}
    \end{align*}
    
    We also have $\|f\|_{\infty}>0$, so
    \begin{align*}
        |\lambda|&=\dfrac{\|(\Psi-\mathrm{Id})f\|_{\infty}}{\|f\|_{\infty}}\le \|\Psi-\mathrm{Id}\|_{\infty\to\infty}
    \end{align*}

Similarly, $\|f\|_{\bm\lambda\infty}=\|\bm\lambda f\|_{\infty}\le \|\bm\lambda\|_{\infty}\|f\|_{\infty}<\infty$. Then,
    \begin{align*}
        \|(\Psi-\mathrm{Id})f\|_{\bm\lambda\infty}=\|\lambda f\|_{\bm\lambda\infty}=|\lambda|\|f\|_{\bm\lambda\infty}
    \end{align*}

    We also have $\|\bm\lambda f\|_{\infty}>0$, as otherwise $\bm\lambda f=0$ a.e $\iota$, which together with $\bm\lambda\ne 0$  a.e. $\iota$ shows that $f=0$ a.e. $\iota$. Thus,
    \begin{align*}
        |\lambda|&=\dfrac{\|(\Psi-\mathrm{Id})f\|_{\bm\lambda\infty}}{\|f\|_{\bm\lambda\infty}}\le \|\Psi-\mathrm{Id}\|_{\bm\lambda\infty\to\bm\lambda\infty}
    \end{align*}
    
    Using $\|\Psi-\mathrm{Id}\|_{\iota\to\iota}=\sup\{|\lambda_n|:n\ge 1\}$, we get
    \begin{align*}
        \|\Psi-\mathrm{Id}\|_{\iota\to\iota}&\le \|\Psi-\mathrm{Id}\|_{\infty\to\infty}\\
        \|\Psi-\mathrm{Id}\|_{\iota\to\iota}&\le \|\Psi-\mathrm{Id}\|_{\bm\lambda\infty\to\bm\lambda\infty}
    \end{align*}
    Using $\rho(\Psi)=\|\Psi\|_{\iota\to\iota}\le 1+\|\Psi-\mathrm{Id}\|_{\iota\to\iota}$ completes the proof.
\end{proof}

\section{Spectral Independence}\label{sec: SSM implies Influence Operator bounds}
We can show spectral independence using either Bounded Complex Density or Strong Spatial Mixing.
They are currently known to hold up to the same threshold on the activity parameter of a Gibbs point process: $\lambda<\frac{e}{\Delta_{\phi}}$. However, SSM requires the potential to be finite-range on $\R^d$, while BCD does not have this restriction. For $\lambda<\frac{e}{C_{\phi}}$, we also compute an explicit bound on the spectral radius of the influence operator (\Cref{subsec: explicit SI}).

\subsection{Spectral Independence from Bounded Complex Density}
\label{subsec:si_from_bcd}

In this section, we show that when $\bm\lambda:\Lambda\to[0,\lambda_0]$ for $\lambda_0<\frac{e}{\Delta_{\phi}}$, spectral independence holds.

We require a bound on the one-point density for complex activities shown in \cite{MP21Potential}. Here, $\mathbb{X}$ is any complete separable metric space satisfying Assumption 1 of \cite{MP21Potential}.

\begin{theorem}[\cite{MP21Potential} proof of Theorem 4, c.f. \cite{MP22Analyticity} Theorem 19]\label{thm:mp22_zero_freeness}
For $\lambda_{0}<\frac{e}{\Delta_{\phi}}$, there exists $\epsilon>0$ and $C>0$ such that for $\Lambda\subseteq\mathbb{X}$ with $|\Lambda|<\infty$ and $\bm\lambda:\Lambda\to\mathbb{D}([0,\lambda_{0}],\epsilon)$,
\begin{align*}
    |\zeta_{\bm\lambda}(x)|&\le C&\forall x\in\mathbb{X}
\end{align*}
In particular,
\begin{align*}
    |\log Z(\bm\lambda)|\le C|\Lambda|
\end{align*}
\end{theorem}

We would like to use a slight strengthening of this statement which we call Bounded Complex Density. Rather than redoing their
proofs, we will prove it using \Cref{thm:mp22_zero_freeness} as a black box.

\begin{corollary}\label{thm:zero_freeness}
For $\lambda_{0}<\frac{e}{\Delta_{\phi}}$, there exists $\epsilon>0$ and $C>0$ such that for $\Lambda\subseteq\mathbb{X}$ with $|\Lambda|<\infty$ and $\bm\lambda:\Lambda\to\mathbb{D}([0,\lambda_{0}],\epsilon)$,
\begin{align*}
|\zeta_{\bm\lambda}(x)|&\le C|\bm\lambda(x)|&\forall x\in\mathbb{X}
\end{align*}
\end{corollary}

\begin{proof}
By \cite{MP21Potential} Proposition 17,
\begin{align*}
\zeta_{\bm\lambda}(x)=\bm\lambda(x)\exp\left(-\int_{\mathbb{X}}\zeta_{\bm\lambda_{x\to y}}(y)(1-e^{-\phi(x,y)})dy\right)
\end{align*}

where $\bm\lambda_{x\to y}(w)=\bm\lambda(w)\exp(-\phi(w,x)\bm1_{d(w,x)<d(y,x)})$. In particular, $\bm\lambda_{x\to y}:\Lambda\to\mathbb{D}([0,\lambda_0],\epsilon)$, and hence $|\zeta_{\bm\lambda_{x\to y}}(y)|\le C$ by \Cref{thm:mp22_zero_freeness}. Thus,
\begin{align*}
    |\zeta_{\bm\lambda}(x)|\le |\bm\lambda(x)|\exp\left(\int_{\mathbb{X}}C(1-e^{-\phi(x,y)})dx\right)\le |\bm\lambda(x)|e^{CC_{\phi}}
\end{align*}

\end{proof}

We can now prove the main theorem of this section, Thm.~\ref{thm:zero_freeness_spectral_independence}.
We use the Schwarz-Pick theorem to show spectral independence, as previously done in the the discrete setting \cite{chen2024spectral}\cite{alimohammadi2021fractionally}.

The following norm will be useful in the proof.
\begin{align*}
    \|f\|_{\bm\lambda\infty}:=\|\bm\lambda f\|_{\infty}
\end{align*}

\begin{proof}[Proof of \Cref{thm:zero_freeness_spectral_independence}]
Fix $x\in \Lambda$, $\bm\lambda:\Lambda\to[0,\lambda_0]$, and measurable $f:\Lambda\to \C$ such that $\|f\|_{\bm\lambda\infty}\le 1$. We will show that
\begin{align*}
    |\bm\lambda(x)(\Psi_{\bm\lambda}-\mathrm{Id})f(x)|\le \frac{C}{\epsilon}e^{\lambda C_{\phi}}+1.
\end{align*}
Define
\begin{align*}
    h(z)&\defeq\zeta_{(1+zf)\bm\lambda}(x)&\forall~z, \text{ s.t. }|z|<\epsilon.
\end{align*}
\begin{claim}
    $h(z)$ is analytic on $\mathbb{D}(0,\epsilon)$.
\end{claim}

\begin{proof}

    Note that
    \begin{align*}
        h(z)=\frac{(1+zf(x))\bm\lambda(x)Z((1+zf)\bm\lambda e^{-\phi(x,\cdot)})}{Z((1+zf)\bm\lambda)}
    \end{align*}
    is a ratio of analytic functions in $z$.
    Also, when $|z|<\epsilon$, we have $(1+zf)\bm\lambda(y)\in \mathbb{D}([0,\lambda_0],\epsilon)$ for all $y\in\mathbb{X}$, so the denominator is nonzero.
\end{proof}

\begin{claim}
The image of $h$ is contained in $\overline{\mathbb{D}}(0,C\bm\lambda(x))$.
\end{claim}

\begin{proof}
Immediate from \Cref{thm:zero_freeness}.
\end{proof}

\begin{claim}
\begin{align*}
    \dfrac{d}{dz}\Big|_{z=0}Z((1+zf)\bm\lambda)=Z(\bm\lambda)\iota_{\bm\lambda}(f)
\end{align*}
\end{claim}
\begin{proof}
    Recall that
    \begin{align*}
        Z(\bm\lambda)=1+\sum_{k\ge 1}\frac{1}{k!}\int e^{-H(x_1,\ldots,x_k)}\bm\lambda(x_1)\cdots\bm\lambda(x_k)dx_1\cdots dx_k
    \end{align*}

    Let $\bm\lambda_z=\bm\lambda(1+zf)$.
    Then,
    \begin{align*}
        &\dfrac{d}{dz}Z(\bm\lambda(1+zf))=\sum_{k\ge 1}\frac{1}{k!}\frac{d}{dz}\int e^{-H(x_1,\ldots,x_k)}\bm\lambda_z(x_1)\cdots\bm\lambda_z(x_k)dx_1\cdots dx_k=\\
        &=\sum_{k\ge 1}\frac{1}{k!}\int e^{-H(x_1,\ldots,x_k)}\bm\lambda_z(x_1)\cdots\bm\lambda_z(x_k)(\frac{f(x_1)}{1+zf(x_1)}+\cdots+\frac{f(x_k)}{1+zf(x_k)})dx_1\cdots dx_k
    \end{align*}
    Thus,
    \begin{align*}
        \dfrac{d}{dz}\Big|_{z=0}Z(\bm\lambda(1+zf))&=\sum_{k\ge 1}\frac{1}{k!}\int e^{-H(x_1,\ldots,x_k)}\bm\lambda(x_1)\cdots\bm\lambda(x_k)(f(x_1)+\cdots+f(x_k))dx_1\cdots dx_k\\
        &=Z(\bm\lambda)\iota_{\bm\lambda}(f)
    \end{align*}
\end{proof}

\begin{claim}
\begin{align*}
    h'(0)=\zeta_{\bm\lambda}(x)(\Psi_{\bm\lambda}-\mathrm{Id})f(x).
\end{align*}
\end{claim}

\begin{proof}
Recall that
    \begin{align*}
        h(z)=(1+zf(x))\frac{\bm\lambda(x)Z((1+zf)\bm\lambda e^{-\phi(x,\cdot)})}{Z((1+zf)\bm\lambda)}
    \end{align*}
Thus,
\begin{align*}
    h'(0)&=f(x)\frac{\bm\lambda(x)Z((1+zf)\bm\lambda e^{-\phi(x,\cdot)})}{Z((1+zf)\bm\lambda)}+\frac{\bm\lambda(x)Z({\bm\lambda e^{-\phi(x,\cdot)}})\iota_{\bm\lambda e^{-\phi(x,\cdot)}}(f)}{Z({\bm\lambda})}-\frac{\bm\lambda(x)Z({\bm\lambda e^{-\phi(x,\cdot)}})}{Z({\bm\lambda})^2}Z({\bm\lambda})\iota_{\bm\lambda}(f)\\
    &=f(x)\zeta_{\bm\lambda}(x)+\zeta_{\bm\lambda}(x)(\iota_{\bm\lambda e^{-\phi(x,\cdot)}}(f)-\iota_{\bm\lambda}(f))\\
    &=\zeta_{\bm\lambda}(x)\Psi_{\bm\lambda}f(x)
\end{align*}
\end{proof}

We need the following consequence of the Schwarz-Pick Theorem.
\begin{fact}\label{lem:schwarz}
Let $f:\mathbb{D}(0,1)\to\overline{\mathbb{D}}(0,1)$ be analytic. Then, $
|f'(0)|\le 1$.
\end{fact}

\begin{claim}
\begin{align*}
    |h'(0)|\le \frac{C}{\epsilon}
\end{align*}

\end{claim}
\begin{proof}
    Apply \Cref{lem:schwarz} to $\frac{1}{C}h(\frac{1}{\epsilon}z)$.
\end{proof}

\begin{claim}
\begin{align*}
|\bm\lambda(x)(\Psi_{\bm\lambda}-\mathrm{Id})f(x)|\le \frac{C}{\epsilon}e^{\lambda C_{\phi}}+1.
\end{align*}
\end{claim}

\begin{proof}
\begin{align*}
|\zeta_{\bm\lambda}(x)\Psi_{\bm\lambda}f(x)|=|h'(0)|\le \frac{C}{\epsilon}.
\end{align*}
By \Cref{lem:gpp:intensity_bounds}, $\zeta_{\bm\lambda}(x)\ge e^{-\lambda C_{\phi}}\bm\lambda(x)$, so
\begin{align*}
    |\bm\lambda(x)\Psi_{\bm\lambda}f(x)|\le \frac{C}{\epsilon}e^{\lambda C_{\phi}}.
\end{align*}
Since $|\bm\lambda(x)f(x)|\le 1$, we have
\begin{align*}
    |\bm\lambda(x)(\Psi_{\bm\lambda}-\mathrm{Id})f(x)|\le \frac{C}{\epsilon}e^{\lambda C_{\phi}}+1.
\end{align*}
\end{proof}

We complete the proof using \Cref{lem:influence_iota_operator_norm_le_infinity_operator_norm}.
\end{proof}

\subsection{Spectral Independence from Strong Spatial Mixing}\label{subsec:si_from_ssm}

In this section we will define Strong Spatial Mixing (SSM) and show how it implies spectral independence.
Before we define Strong Spatial Mixing, let us introduce some necessary definitions and notation.

\begin{definition}[Stochastic domination]
Given distributions $\mu$ and $\nu$ whose samples lie in a partially ordered set $(S,\le)$, we say that $\mu$ is \textbf{stochastically dominated} by $\nu$ if for any nondecreasing measurable function $f:S\to\R$, $\int fd\mu\le \int f\nu$.
\end{definition}

We can consider $\Omega$ partially ordered by set inclusion, or $\R$ ordered by $\le$.

\begin{lemma}[\cite{GK97}]\label{lemma: stochastic dom Pois}
    For a Gibbs point process $\mu_{\bm\lambda}$ with \textbf{repulsive} pair potentials with activity $\bm\lambda$, we have $\mu_{\bm\lambda}\preceq \rho_{\bm\lambda}$. In other words, a repulsive Gibbs point process is stochastically dominated by a Poisson point process of the same $\bm\lambda$.
\end{lemma}

\begin{definition}
    By $\|\mu_{\bm{\lambda}}-\mu_{\bm{\lambda'}}\|_{B}$ we will denote the TV distance between the laws of $\mu_{\bm{\lambda}}$ and $\mu_{\bm{\lambda'}}$ projected to $B$.
\end{definition}

We are ready to start working with SSM. 

\begin{definition}[Strong Spatial Mixing, \cite{MP22Strong}]
    The family of point processes on $\R^d$ defined by a repulsive pair potential $\phi$ exhibits SSM with activities bounded by $\lambda>0$ if there exist constants $\alpha,\beta$ so that the following holds. For any bounded, measurable region $B\subset \R^d$ and any two activity functions $\bm{\lambda},\bm{\lambda'}:B\to[0,\lambda]$, we have that:
\begin{align*}
\|\mu_{\bm{\lambda}}-\mu_{\bm{\lambda}'}\|_{B}\le \alpha|B|e^{-\beta\cdot\dist(B,\supp(\bm{\lambda}-\bm{\lambda'}))}.
\end{align*}
\end{definition}

Before we state the SSM result by \cite{MP22Strong}, let us define the potential-weighted connective constant. We give the definition of the constant for completeness, but only a few properties will be used.

\begin{definition}[$V_k$]\label{def: vk}
    For a repulsive potential $\phi$ and $k \in \mathbb{N}$, 
    \[ V_k \defeq \sup_{v_0 \in \mathbb{X}} \int_{\mathbb{X}^k} \prod_{j=1}^k \left( \exp\left(-\sum_{i=0}^{j-2} 1_{d(v_j, v_{i})<d(v_i, v_{i+1}) }\phi(v_j, v_i)\right)  \cdot \left(1-e^{-\phi(v_j, v_{j-1})}\right)\right) dx^k(\textbf{v}) \]
\end{definition}

\begin{definition}[Potential-weighted connective constant $\Delta_{\phi}$]\label{def: pot-weighted con constant}
    \begin{align*}
    \Delta_{\phi}\defeq\lim_{k\to\infty}V_k^{1/k}=\inf_{k\ge 1}V_k^{1/k}
\end{align*}
\end{definition}

\begin{theorem}[SSM for Gibbs point process, \cite{MP22Strong}]\label{thm: SSM}
    Let $\phi:\R^d\times\R^d\to[0,\infty]$ be a tempered finite-range repulsive potential. Then for any $\lambda \in [0, e/\Delta_{\phi})$, where $\Delta_\phi$ is the potential-weighted connective constant, the family of point processes defined by $\phi$ exhibits SSM with activities bounded by $\lambda$. 
\end{theorem}

In the following lemma, we show that Strong Spatial Mixing gives a bound on the infinity norm of the $\Psi-\mathrm{Id}$ operator, implying spectral independence.

\lemmaSSMSI*

\begin{proof}
    Let $\bm\lambda:\Lambda\to[0,\lambda]$. Fix $x\in\Lambda$ and let $\bm\lambda'=\bm\lambda\cdot\exp(-\phi(x,\cdot))$.
    Then, $\supp(\bm\lambda-\bm\lambda')\subseteq \overline{B(x,r)}$.
    
    Let $A\subseteq\Lambda$.
Let $\epsilon>0$. Let $\mcP$ be a partition of $A$ into a finite number of pieces with diameter at most $\epsilon$ and $\lambda|B|<\epsilon$ for all $B\in \mcP$.

From the definition of Strong Spatial Mixing and supremum characterization of total variation distance, we have
\begin{align*}
|\Pr_{\eta\sim\mu_{\bm\lambda}}(\eta(B)\ge 1)-\Pr_{\eta\sim\mu_{\bm\lambda'}}(\eta(B)\ge 1)|\le \|\mu_{\bm{\lambda}}-\mu_{\bm{\lambda'}}\|_{B}\le \alpha|B|e^{-\beta\cdot\dist(B,\supp(\bm{\lambda}-\bm{\lambda'}))}
\end{align*}

We next show that we can approximate $\Pr_{\eta\sim\mu_{\bm\lambda}}(\eta(B)\ge 1)$ with $\E_{\eta\sim\mu_{\bm\lambda}}[\eta(B)]$.

     Since $\mu_{\bm\lambda}$ is stochastically dominated by a Poisson point process of intensity $\lambda$,
    \begin{align}
        \E_{\eta\sim\mu_{\bm\lambda}}[\eta(B)-1_{\eta(B)\ge 1}]&\le \E_{X\sim\Pois(\lambda|B|)}[X-1_{X\ge 1}] \nonumber\\
        &\le \lambda|B|-(1-e^{-\lambda|B|}) \nonumber\\
        &=e^{-\lambda|B|}-(1-\lambda|B|)\le \dfrac{1}{2}(\lambda|B|)^2. \label{eq: by stochastic dom}
    \end{align}
    
    Similarly,
    \begin{align*}
        \E_{\eta\sim\mu_{\bm\lambda'}}[\eta(B)-1_{\eta(B)\ge 1}]\le \frac{1}{2}(\lambda|B|)^2
    \end{align*}

    Now,
    \begin{align}
        |\E_{\eta\sim\mu_{\bm\lambda'}}[\eta(A)]-\E_{\eta\sim\mu_{\bm\lambda}}[\eta(A)]|&\le \sum_{B\in\mcP}|\E_{\eta\sim\mu_{\bm\lambda'}}[\eta(B)]-\E_{\eta\sim\mu_{\bm\lambda}}[\eta(B)]|\nonumber \\
        &\le\sum_{B\in\mcP}\Big(|\E_{\eta\sim\mu_{\bm\lambda'}}[1_{\eta(B)\ge 1}]-\E_{\eta\sim\mu_{\bm\lambda}}[1_{\eta(B)\ge 1}]|\nonumber \\
        &~~+\E_{\eta\sim\mu_{\bm\lambda'}}[\eta(B)-1_{\eta(B)\ge 1}]+\E_{\eta\sim\mu_{\bm\lambda}}[\eta(B)-1_{\eta(B)\ge 1}]\Big)\nonumber \\
        &\le \sum_{B\in\mcP}\alpha|B|e^{-\beta\cdot\dist(B,B(x,r))}+\sum_{B\in\mcP}(\lambda|B|)^2
        \label{eq: two term upper}
    \end{align}
    where the last inequality follows by (\ref{eq: by stochastic dom}).

    Let us upper bound each of the two terms separately:
    
    \begin{align}
    \sum_{B\in\mcP}\alpha|B|e^{-\beta\cdot\dist(B,B(x,r))}
    &=\alpha \sum_{B\in\mcP}\int_B e^{-\beta\cdot\dist(B,B(x,r))}dy \nonumber\\
    &\le \alpha \sum_{B\in\mcP}\int_B e^{-\beta\cdot\max(0,\dist(x,y)-r-\epsilon)}dy \nonumber\\
    &=\alpha \int_A e^{-\beta\cdot\max(0,\dist(x,y)-r-\epsilon)}dy \nonumber\\
    &\le \alpha e^{\beta\epsilon}\int_{A} e^{-\beta\cdot\max(0,\dist(x,y)-r)}dy
    \label{eq: int term}
    \end{align}

    and 
    \begin{align}
        \sum_{B\in\mcP}(\lambda|B|)^2&\le \sum_{B\in\mcP}\lambda|B|\epsilon=\epsilon\lambda|A|. \label{eq: lambda term}
    \end{align}

    Combining \ref{eq: two term upper}, \ref{eq: int term}, \ref{eq: lambda term} we get:
    
    \begin{align*}
        |\E_{\eta\sim\mu_{\bm\lambda'}}[\eta(A)]-\E_{\eta\sim\mu_{\bm\lambda}}[\eta(A)]|&\le \alpha e^{\beta\epsilon}\int_{A} e^{-\beta\cdot\max(0,\dist(x,y)-r)}dy+\epsilon\lambda|A|.
    \end{align*}
    
    Taking $\epsilon\to 0$, this yields
    \begin{align*}
        |\E_{\eta\sim\mu_{\bm\lambda'}}[\eta(A)]-\E_{\eta\sim\mu_{\bm\lambda}}[\eta(A)]|&\le \alpha \int_{A} e^{-\beta\cdot\max(0,\dist(x,y)-r)}dy
    \end{align*}

    By first proving for $f=\bm1_{A}$, and then using linearity and taking limits,
    we get that for all bounded measurable $f:\Lambda\to\R_{\ge 0}$, we have
    \begin{align*}
    |\E_{\eta\sim\mu_{\bm\lambda'}}[\eta(f)]-\E_{\eta\sim\mu_{\bm\lambda}}[\eta(f)]| \le \alpha \int_{\Lambda}|f(y)|e^{-\beta\cdot\max(0,\dist(x,y)-r)}dy.
    \end{align*}
    Splitting into positive and negative parts then shows that this in fact holds for all bounded measurable $f:\Lambda\to\R$.
    
    Recall that
    \begin{align*}
        \Psi_{\bm\lambda} f(x)&:=\E_{\eta\sim\mcR_{+x}\mu_{\bm\lambda}}[\eta(f)]-\E_{\eta\sim\mu_{\bm\lambda}}[\eta(f)]\\
        &=f(x)+\E_{\eta\sim\mu_{\bm\lambda'}}[\eta(f)]-\E_{\eta\sim\mu_{\bm\lambda}}[\eta(f)]
    \end{align*}
    Thus, for all bounded measurable $f:\Lambda\to\R$,
    \begin{align*}
    |(\Psi_{\bm\lambda}-\mathrm{Id})f(x)|\le \alpha \int_{\Lambda}|f(y)|e^{-\beta\cdot\max(0,\dist(x,y)-r)}dy.
    \end{align*}
    
    In particular,
    \begin{align*}
        \|(\Psi_{\bm\lambda}-\mathrm{Id})f\|_{\infty}\le \alpha \int_{\Lambda}e^{-\beta\cdot\max(0,\dist(x,y)-r)}dy\|f\|_{\infty},
    \end{align*}
    so
    \begin{align*}
        \|\Psi_{\bm\lambda}-\mathrm{Id}\|_{\infty\to\infty}&\le \alpha \int_{\R^d}e^{-\beta\cdot\max(0,\dist(x,y)-r)}dy.
    \end{align*}

    Finally, \Cref{lem:influence_iota_operator_norm_le_infinity_operator_norm} shows that
    \begin{align*}
        \rho(\Psi_{\bm\lambda})\le 1+\alpha \int_{\R^d}e^{-\beta\cdot\max(0,\dist(x,y)-r)}dy.
    \end{align*}
\end{proof}

\begin{remark}
    The upper bound here depends on the dimension, as the constants $\alpha$ and $\beta$ from \cite{MP22Strong} do.
\end{remark}

\begin{remark}
\cite{gobel2026uniqueness} later introduced a variant of Strong Spatial Mixing called single-site Strong Spatial Mixing which has the advantage of not requiring the pair potential to be finite-range. This also implies spectral independence by an almost identical proof.
\end{remark}

\subsection{Explicit Spectral Independence Constant}\label{subsec: explicit SI}

For a slightly smaller activity range $\lambda<\frac{e}{C_{\phi}}$, it is possible to directly bound the spectral radius of the influence operator using an ingredient used to prove SSM in \cite{MP22Strong}. The benefit of this approach is it lets us compute an explicit upper bound on the spectral radius, which we will show in this subsection. 

Recall that the one-point density is defined as $\zeta_{\bm\lambda}(v) \defeq \bm \lambda(v)\frac{Z(\bm\lambda e^{-\phi(v, \cdot)})}{Z(\bm\lambda )}$, where $\phi(v, \cdot)$ is the function $w \to \phi(v, w)$.

\begin{proposition}[Proposition 10 of \cite{MP22Strong}]\label{prop: 10}
Suppose $\phi$ has range at most $r$ and let $V_k$ be defined as in \Cref{def: vk}.
Then for every $\lambda\ge 0$, $x\in\R^d$ and activity functions $\bm\lambda,\bm\lambda':\Lambda\to[0,\lambda]$ we have
\begin{align*}
|\zeta_{\bm\lambda}(x)-\zeta_{\bm\lambda'}(x)|\le 2\lambda(\lambda/e)^{k/2}\sqrt{V_k}
\end{align*}
where $k=\lfloor \dist(v,\supp(\bm\lambda-\bm\lambda'))/r\rfloor$.
\end{proposition}

We also abbreviate $B_d$ to be a unit ball in $\R^d$. 
\begin{lemma}
    For $\lambda<e^{-\delta}\dfrac{e}{C_{\phi}}$ and $\bm\lambda:\Lambda\to[0,\lambda]$,
\begin{align*}
    \|\Psi_{\bm\lambda}-\mathrm{Id}\|_{\infty\to\infty}&\le 2\dfrac{e}{C_{\phi}}\vol_{d}(B_d)\left(\dfrac{2r}{\delta}\right)^dd!
\end{align*}
\end{lemma}

\begin{proof}

    Recall that  
    \begin{align*}
        (\Psi_{\bm\lambda}-\mathrm{Id})f(x)&=\E_{\eta\sim\tilde{\mcR}_{+x}\mu_{\bm\lambda}}[\eta(f)]-\E_{\eta\sim\mu_{\bm\lambda}}[\eta(f)]\\
        &=\iota_{e^{-\phi(x,\cdot)}\bm\lambda}(f)-\iota_{\bm\lambda}(f)\\
        &=\int_{\Lambda}(\zeta_{e^{-\phi(x,\cdot)}\bm\lambda}(y)-\zeta_{\bm\lambda}(y))f(y)dy.
    \end{align*}
    
    Thus,
    \begin{align*}
        |(\Psi_{\bm\lambda}-\mathrm{Id})f(x)|&\le \int_{\Lambda}|\zeta_{e^{-\phi(x,\cdot)}\bm\lambda}(y)-\zeta_{\bm\lambda}(y)|dy\, \, \|f\|_{\infty}.
    \end{align*}
    
    Therefore,
    \begin{align*}
        \|\Psi_{\bm\lambda}-\mathrm{Id}\|_{\infty\to\infty}&\le \sup_{x\in\Lambda}\int_{\Lambda}|\zeta_{e^{-\phi(x,\cdot)}\bm\lambda}(y)-\zeta_{\bm\lambda}(y)|dy.
    \end{align*}

    The RHS can be shown to be bounded by a constant (depending only on $\lambda$ and $(V_k)_{k\ge 1}$) using \Cref{prop: 10}:

    Note that $\lfloor\dist(y,B(x,r))/r\rfloor=\lfloor\max(0,\|\frac{y-x}{r}\|_2-1)\rfloor$.
    
    Letting $k(w):=\lfloor\max(0,\|\frac{w}{r}\|_2-1)\rfloor$, we have
    \begin{align*}
        \|\Psi_{\bm\lambda}-\mathrm{Id}\|_{\infty\to\infty}&\le \sup_{x\in\Lambda}\int_{\Lambda}|\zeta_{e^{-\phi(x,\cdot)}\bm\lambda}(y)-\zeta_{\bm\lambda}(y)|dy\\
        &\le \sup_{x\in\Lambda}\int_{\R^d}2\lambda(\lambda/e)^{k(y-x)/2}\sqrt{V_{k(y-x)}}dy\\
        &=\sup_{x\in\Lambda}\int_{\R^d}2\lambda(\lambda/e)^{k(w)/2}\sqrt{V_{k(w)}}dw\\
        &=\int_{\R^d}2\lambda(\lambda/e)^{k(w)/2}\sqrt{V_{k(w)}}dw.
    \end{align*}

    We have $V_k\le C_{\phi}^k$ for $k\ge 1$ (see \cite{MP21Potential}), so
    \begin{align*}
        \|\Psi_{\bm\lambda}-\mathrm{Id}\|_{\infty\to\infty}&\le \int_{\R^d}2\lambda(\lambda C_{\phi}/e)^{k(w)/2}dw<\infty.
    \end{align*}
    
    Thus, for $\lambda\le e^{-\delta}\dfrac{e}{C_{\phi}}$,
    \begin{align*}
        \|\Psi_{\bm\lambda}-\mathrm{Id}\|_{\infty\to\infty}&\le 2\lambda  \int_{\R^d}e^{-\delta k(w)/2}dw&\text{note }k(w)\ge \max(0,\|\frac{w}{r}\|_2-2)\\
        &\le 2\lambda\int_{\R^d}\exp(-\delta(\|\frac{w}{2r}\|_2-1))dw\\
        &=2\lambda e^{\delta}\int_0^\infty \exp(-\frac{\delta}{2r}w)\vol_{d-1}(\partial B_d)w^{d-1}dw\\
        &=2\lambda e^{\delta}\vol_{d-1}(\partial B_d)\int_0^{\infty}\exp(-\frac{\delta}{2r}w)w^{d-1}dw\\
        &=2\lambda e^{\delta}\vol_{d-1}(\partial B_d)\left(\dfrac{2r}{\delta}\right)^d\int_0^{\infty}\exp(-w)w^{d-1}dw\\
        &=2\lambda e^{\delta}\vol_{d-1}(\partial B_d)\left(\dfrac{2r}{\delta}\right)^d(d-1)!\\
        &=2\lambda e^{\delta}\vol_{d}(B_d)\left(\dfrac{2r}{\delta}\right)^dd!\\
        &\le 2\dfrac{e}{C_{\phi}}\vol_{d}(B_d)\left(\dfrac{2r}{\delta}\right)^dd!
    \end{align*}
    where we used the fact that
    \begin{align*}
        \vol_d(B_d)=\int_0^1r^{d-1}\vol_{d-1}(\partial B_d)dr=\dfrac{1}{d}\vol_{d-1}(\partial B_d) 
    \end{align*}
    and 
    \begin{align*}
        \int_0^{\infty}e^{-w}w^{d-1}dw=(d-1)! \,\,.
    \end{align*}
\end{proof}

\begin{corollary}\label{cor:infinity_op_norm_for_hard_spheres}
    For hard spheres, $C_\phi=(2r)^d\vol_d(B_d)$, where $r$ is the radius of the sphere. 
    
    Hence, for $\lambda\le e^{-\delta}\dfrac{e}{C_{\phi}}$ and $\bm\lambda:\Lambda\to[0,\lambda]$,
    \begin{align*}
        \|\Psi_{\bm\lambda}-\mathrm{Id}\|_{\infty\to\infty}&\le \dfrac{e2^{d+1}d!}{\delta^d}
    \end{align*}
    Thus, $\mu_{\bm\lambda}$ is $\kappa$-spectrally independent for $\kappa=1+\dfrac{e2^{d+1}d!}{\delta^d}=\exp(\Theta(d \log \frac{d}{\delta}))$
\end{corollary}

\section{Localization Scheme}\label{sec: loc scheme}

In this section, we describe the method to transform a spectral gap for a Gibbs point process process at a lower activity to one at a higher activity. The method we will develop is a continuous version of the {\em negative fields} localization scheme introduced by Chen and Eldan \cite{CE22} in the setting of the hard-core model on graphs. 

We will generalize the definitions of {\em tilts} and {\em pinnings} (Subsection \ref{subsec: tilts and pinnings}), analyze linear approximation of tilts (Subsection \ref{subsec: lin approx of tilt}), construct a pinning process (Subsection \ref{subsec: construction of loc scheme}), and verify the martingale property for the localization scheme (Subsection \ref{subsec: martingale}).

\subsection{Properties of Tilts and Pinnings}\label{subsec: tilts and pinnings}

In this section, we develop properties of pinnings and tilts, the main concepts necessary to define the negative-fields localization scheme. 

\subsubsection{Basic properties}

To avoid technical difficulties, we define the following operation of ``adding'' a set, which is slightly different from a pinning.

\begin{definition}[Addition]
    Let $\nu$ be a measure on $\Omega$. For fixed $A\in \Omega$, define $f_A:\Omega \to \Omega$ by
    \begin{align*}
        f_A(\eta)=\eta\cup A
    \end{align*}
    Then, we define the \textit{measure with addition} $\nu^{+A}$ to be the pushforward of $\nu$ by $f_A$. In other words, for measurable $\mcB\subseteq\Omega$,
    \begin{align*}
        \nu^{+A}(\mcB)\defeq \nu(f_A^{-1}(\mcB)).
    \end{align*}
\end{definition}
Observe that sampling from $\nu^{+A}$ is the same as sampling from $\nu$ and adding $A$ to the sample.

\begin{example}
    For a simple example, consider the discrete space $\Lambda = \{1, 2, 3\}$, $\Omega = 2^{\Lambda}$. Let $A=\{1\}$, and let $\mcB = \{ \{1, 2, 3\}, \{3\} \}$. Let $\mu$ be the uniform distribution on $\Omega$. Then, since only $\{1,2,3\}$ is in the image of $f_A$, $f_A^{-1}(\mcB) = f_A^{-1}(\{1,2,3\}) = \{ \{1, 2, 3\}, \{2, 3\} \}$. Hence, $\mu^{+A}(\mcB)=\mu(f^{-1}(\mcB)) = 1/4$.
\end{example}

The following change-of-variable formula will be useful later in the section.

\begin{lemma}\label{lemma: change of variable for addition} For any measurable $\varphi:\Omega\to\R$,
\begin{align*}
    \E_{\eta\sim\nu^{+A}}[\varphi(\eta)]=\int_{\Omega} \varphi(\eta) d\nu^{+A}(\eta)=\int_{\Omega} \varphi(\eta\cup A)d\nu(\eta)=\E_{\eta\sim\nu}[\varphi(\eta\cup A)]
\end{align*}
\end{lemma}
\begin{proof}
    First and third equalities are true by definition. The second equality is immediate from \cite{folland1999real} (Proposition 10.1).
\end{proof}

\begin{lemma}\label{lemma: disjoint as}
 Note that for $\bm\lambda:\Omega\to\R_{\ge 0}$, any finite $A\subseteq\Lambda$ is disjoint from $\eta\sim\mu_{\bm\lambda}$ almost surely.
\end{lemma}

\begin{proof}
Let $x\in\Lambda$. Let $\rho_{\bm\lambda}$ be the Poisson point process of intensity $\bm\lambda$. Then, since $\mu_{\bm\lambda}\preceq\rho_{\bm\lambda}$,
\begin{align*}
    \Pr_{\eta\sim\mu_{\bm\lambda}}(x\in\eta)\le \Pr_{\eta\sim\rho_{\bm\lambda}}(x\in\eta)=0
\end{align*}
The union bound then shows the claim.
\end{proof}

\begin{definition}[Intensity of measure with addition]

Let $\iota_{\bm\lambda}^{+A}$ denote the \textit{intensity} of $\mu_{\bm\lambda}^{+A}$, 
    \begin{align*}
        \iota_{\bm\lambda}^{+A}(B)&:=\E_{\eta\sim\mu_{\bm\lambda}^{+A}}[\eta(B)]&\forall B\subseteq \Lambda, B \mbox{ is measurable}
    \end{align*}
\end{definition}

\begin{lemma}\label{lemma: addition in intensity}
    \begin{align*}
    \iota_{\bm\lambda}^{+A}(\Lambda)=|A|+\iota_{\bm\lambda}(\Lambda)
    \end{align*}
\end{lemma}
\begin{proof}
\begin{align*}
    \iota_{\bm\lambda}^{+A}(\Lambda)&=\E_{\eta\sim\mu_{\bm\lambda}^{+A}}[|\eta|]=\E_{\eta\sim\mu_{\bm\lambda}}[|\eta\cup A|]=\E_{\eta\sim\mu_{\bm\lambda}}[|\eta|+|A|]=\E_{\eta\sim\mu_{\bm\lambda}}[|\eta|]+|A|
\end{align*}

where we used that $\eta\cap A=\emptyset$ almost surely for $\eta\sim\mu_{\bm\lambda}$
\end{proof}

Based on the definition of measure with addition, we define a \textit{set of probability measures with additions}
     \[ \mcP(\Lambda)=\{\mu_{\bm\lambda}^{+S}|\bm\lambda:\Lambda\to\R_{\ge 0},S\in\Omega\}\]

\begin{restatable}[]{lemma}{lemanotherdeftilt}\label{lemma: another def of tilts}
    For $t\ge 0$, tilt $\mcT_{-t}:\mcP(\Lambda)\to\mcP(\Lambda)$ can be defined alternatively by
\begin{align*}
    \mcT_{-t}\mu_{\bm\lambda}^{+S}&=\mu_{e^{-t}\bm\lambda}^{+S}&\text{ for any }\mu_{\bm\lambda}^{+S}\in\mcP(\Lambda)
\end{align*}
\end{restatable}

This is closer to the definition of tilts in \cite{CE22}, we show the equivalence in appendix \ref{sec: append for sec 2+3}.

\begin{observation}
    $\mcT_{-t}$ and $\mcI_A$ commute for any $t\ge 0$ and finite $A\subseteq\Lambda$.
\end{observation}

Now we are ready to define pinnings. 

\begin{definition}[Pinning]
For $A\in\Omega$, we define $\mcR_A:\mcP(\Lambda)\to\mcP(\Lambda)$ by
\begin{align*}
    \mcR_{A}\mu_{\bm\lambda}^{+S}:=\mu_{\bm\lambda\cdot \exp\left(-\sum_{x\in A\setminus S}\phi(x,\cdot)\right)}^{+S\cup A}
\end{align*}

For the \textit{hard sphere model}, we can simplify this to be
\begin{align*}
    \mcR_{A}\mu_{\bm\lambda}^{+S}:=\mu_{\bm\lambda\cdot \mathbf{1}_{{B(A\setminus S,r)}^c}}^{+S\cup A},
\end{align*}
    where $B(A\setminus S,r)^c$ is the complement of the union of balls of radius $r$ around points in $A\setminus S$. 
\end{definition}
Pinnings not only require points in $A$ to be present, but also modify the activity in a way that prevents points within radius $r$ of $A$ to be added (excluding points that are already in $S$ to ensure idempotence). 

Note that $\mcR_{B}$ and $\mcI_A$ do not necessarily commute.

We also need a definition of pinning without addition for some arguments. 

\begin{definition}[Pinning without addition]
For $A\in\Omega$, we define $\tilde{\mcR}_A:\mcP(\Lambda)\to\mcP(\Lambda)$ by
\begin{align*}
    \tilde{\mcR}_{A}\mu_{\bm\lambda}^{+S}:=\mu_{\bm\lambda\cdot \exp\left(-\sum_{x\in A\setminus S}\phi(x,\cdot)\right)}^{+S}
\end{align*}
\end{definition}

\begin{lemma}[Commutativity]
For $A,B\in\Omega$ and $t,u\ge 0$, note that 
    \begin{align}
        \mcT_{-t}\mcT_{-u}\mu_{\bm\lambda}^{+S}&=\mcT_{-t-u}\mu_{\bm\lambda}^{+S} \label{eq: t+t}\\
        \mcR_A\mcR_B\mu_{\bm\lambda}^{+S}&=\mcR_{A\cup B}\mu_{\bm\lambda}^{+S} \label{eq: p+p}\\
        \mcR_A\mcT_{-t}\mu_{\bm\lambda}^{+S}&=
        \mcT_{-t}\mcR_A\mu_{\bm\lambda}^{+S} \label{eq: t+p},
    \end{align}
    In particular, tilts and pinnings commute. 
\end{lemma}

\begin{proof}
For (\ref{eq: t+t}),
\begin{align*}
    \mcT_{-t}\mcT_{-u}\mu_{\bm\lambda}^{+S}=\mcT_{-t}\mu_{e^{-u}\bm\lambda}^{+S}=\mu_{e^{-t-u}\bm\lambda}^{+S}=\mcT_{-t-u}\mu_{\bm\lambda}^{+S},
\end{align*}

For (\ref{eq: p+p}),
\begin{align*}
    \mcR_{A}\mcR_{B}\mu_{\bm\lambda}^{+S} = \mcR_{A} \mu_{\bm\lambda\cdot \exp\left(-\sum_{x\in B\setminus S}\phi(x,\cdot)\right)}^{+S\cup B} &= \mu_{\bm\lambda\cdot \exp\left(-\sum_{x\in B\setminus S}\phi(x,\cdot)\right) \exp\left(-\sum_{x\in A\setminus (S\cup B)}\phi(x,\cdot)\right)}^{+S\cup B \cup A}\\
    = \mu_{\bm\lambda\cdot \exp\left(-\sum_{x\in (A\cup B)\setminus S}\phi(x,\cdot)\right) }^{+S\cup (A \cup B)} &= \mcR_{A\cup B}\mu_{\bm\lambda}^{+S} ,
\end{align*}

For (\ref{eq: t+p}),
\begin{align*}
        \mcR_A\mcT_{-t}\mu_{\bm\lambda}^{+S}&=\mcR_A\mu_{e^{-t}\bm\lambda}^{+S}=\mu_{\exp(-\sum_{x\in A\setminus S}\phi(x,\cdot))e^{-t}\bm\lambda}^{+S\cup A}\\
        &=\mcT_{-t}\mu_{\exp(-\sum_{x\in A\setminus S}\phi(x,\cdot))\bm\lambda}^{+S\cup A}=
        \mcT_{-t}\mcR_A\mu_{\bm\lambda}^{+S} .
    \end{align*}
\end{proof}

Lastly, we show a useful lemma, which will appear in the proofs of the next subsection. 

\begin{lemma}\label{lem:integrate_pinning}
For bounded $\varphi:\Omega\to\R$, $f:\Lambda \to \R$,
\begin{align*}
    \E_{\eta\sim\mu_{\bm\lambda}}[\eta(f)\varphi(\eta)]   &=  \int_{\Lambda} f(x) \E_{\eta\sim\mcR_{+x}\mu_{\bm\lambda}}[\varphi(\eta)]d\iota_{\bm\lambda}(x)
\end{align*}
\end{lemma}

\begin{proof}
Let
\begin{align*}
\bm\lambda^{+x}(y)=\bm\lambda(y)\cdot e^{-\phi(x,y)}
\end{align*}

Recall that the one-point density $\zeta_{\bm\lambda}$ of $\mu$ is given by
\begin{align*}
    \zeta(x)\defeq\bm\lambda(x)\dfrac{Z_{\Lambda}(\bm\lambda^{+x})}{Z_{\Lambda}(\bm\lambda)}
\end{align*}

We treat $\varphi$ as a symmetric function on $\Lambda^k$.
\begin{align*}
    &\E_{\eta\sim\mu_{\bm\lambda}}[\eta(f)\varphi(\eta)]
    \\
    &=\dfrac{1}{Z_{\Lambda}(\bm\lambda)}\sum_{k\ge 0}\dfrac{1}{k!}\int_{\Lambda^k}\bm\lambda(x_1)\cdots\bm\lambda(x_k)(f(x_1)+\cdots+f(x_k))\varphi(x_1,\ldots,x_k)e^{-H(x_1,\ldots,x_k)}dx_1\cdots dx_k\\
    &=\dfrac{1}{Z_{\Lambda}(\bm\lambda)}\sum_{k\ge 1}\dfrac{1}{(k-1)!}\int_{\Lambda^k}\bm\lambda(x_1)\cdots\bm\lambda(x_k)f(x_k)\varphi(x_1,\ldots,x_k)e^{-H(x_1,\ldots,x_k)}dx_1\cdots dx_k\\
    &=\dfrac{1}{Z_{\Lambda}(\bm\lambda)}\sum_{k\ge 0}\dfrac{1}{k!}\int_{\Lambda^{k+1}}\bm\lambda(x_1)\cdots\bm\lambda(x_{k+1})f(x_{k+1})\varphi(x_1,\ldots,x_{k+1})e^{-H(x_1,\ldots,x_{k+1})}dx_1\cdots dx_{k+1}\\
    &=\dfrac{1}{Z_{\Lambda}(\bm\lambda)}\sum_{k\ge 0}\dfrac{1}{k!}\int_{\Lambda^{k+1}}\bm\lambda^{+x_{k+1}}(x_1)\cdots\bm\lambda^{+x_{k+1}}(x_k)\bm\lambda(x_{k+1})f(x_{k+1})\varphi(x_1,\ldots,x_{k+1})e^{-H(x_1,\ldots,x_k)}dx_1\cdots dx_{k+1}\\
    &=\int_{\Lambda}\left(\dfrac{1}{Z_{\Lambda}(\bm\lambda^{+x_0})}\sum_{k\ge 0}\dfrac{1}{k!}\int_{\Lambda^{k+1}}\bm\lambda^{+x_k}(x_1)\cdots\bm\lambda^{+x_k}(x_k)\varphi(x_0,x_1,\ldots,x_k)e^{-H(x_1,\ldots,x_k)}dx_1\cdots dx_k\right)\\
    &~~~~~~~~~\bm\lambda(x_{k+1})\dfrac{Z_{\Lambda}(\bm\lambda^{+x_{k+1}})}{Z_{\Lambda}(\bm\lambda)}f(x_{k+1})dx_{k+1}\\
    &=\int_{\Lambda}f(x)\E_{\eta\sim\mcR_{+x}\mu_{\bm\lambda}}[\varphi(\eta)]d\iota(x)
\end{align*}

\end{proof}

\subsubsection{Linear Approximation of Negative Tilt}\label{subsec: lin approx of tilt}

In this subsection, we show a linear approximation for negative tilts (\Cref{lem:tilt_first_order}).

In this subsection, $o(h)$ means a quantity bounded in magnitude by some function $r(h)$ depending only on $h$ and $\rho$ such that $\lim_{h\to 0}\frac{r(h)}{h}=0$, that may differ between instances.

\begin{lemma}\label{lem:tilt_first_order}
Let $\mu\preceq \rho$ be a repulsive point process, and $\iota$ be its intensity measure. Then, for measurable $\varphi:\Omega\to[-1,1]$,
\begin{align*}
    \E_{\eta\sim\mcT_{-h}\mu}[\varphi(\eta)]&=\E_{\eta\sim\mu}[\varphi(\eta)]-h\E_{\eta\sim\mu}[(|\eta|-\iota(\Lambda))\varphi(\eta)]+o(h)
\end{align*}

Furthermore, if $A\in\Omega$ and $\iota^{+A}$ is the intensity measure of $\mu^{+A}$,
\begin{align*}
    \E_{\eta\sim\mcT_{-h}\mu^{+A}}[\varphi(\eta)]&=\E_{\eta\sim\mu^{+A}}[\varphi(\eta)]-h\E_{\eta\sim\mu^{+A}}[(|\eta|-\iota^{+A}(\Lambda))\varphi(\eta)]+o(h)
\end{align*}
\end{lemma}

Note that we can apply this to any bounded $\varphi$ by scaling.

\begin{proof}
    Recall that
    \begin{align*}
        d\mcT_{-h}\mu&=g_hd\mu&g_h(\eta):=\frac{e^{-h|\eta|}}{\int e^{-h|\xi|}d\mu(\xi)}
    \end{align*}

    We first give a linear approximation for the denominator of $g_h$.
    \begin{align*}
        \left|\int e^{-h|\xi|}-(1-h|\xi|)d\mu(\xi)\right|&\le \int (h|\xi|)^2d\mu(\xi)&\text{ since $|e^{-x}-(1-x)|\le x^2$ for $x\ge 0$}\\
        &\le h^2\int |\xi|^2d\rho(\xi)&\text{by stochastic domination}\\
    \end{align*}
    Hence,
    \begin{align*}
        \int e^{-h|\xi|}d\mu(\xi)&=1-h\iota(\Lambda)+o(h)
    \end{align*}
and so
\begin{align*}
    g_h(\eta)&=\dfrac{e^{-h|\eta|}}{\int e^{-h|\xi|}d\mu(\xi)}=\dfrac{e^{-h|\eta|}}{1-h\iota(\Lambda)+o(h)}=(1+h\iota(\Lambda)+o(h))e^{-h|\eta|}.
\end{align*}
We next give a linear approximation for $g_h$.
\begin{align*}
\int_{\Omega}|g_h(\eta)-(1-h(|\eta|-\iota(\Lambda)))|d\mu(\eta)&=
\int_{\Omega}|(1+h\iota(\Lambda)+o(h))e^{-h(|\eta|)}-(1-h(|\eta|-\iota(\Lambda)))|d\mu(\eta)\\
&\le o(h)\cdot \underbrace{\int_{\Omega}e^{-h|\eta|}d\mu(\eta)}_{\le 1}+\int\left|(1+h\iota(\Lambda))(e^{-h|\eta|}-1)+h|\eta|\right|d\mu(\eta)\\
&\le \int|(1+h\iota(\Lambda))(e^{-h|\eta|}-1+h|\eta|)|d\mu(\eta)+o(h)+h^2\iota(\Lambda)^2\\
&\le (1+h\iota(\Lambda))\int h^2|\eta|^2d\mu(\eta)+o(h)\\
&=o(h)
\end{align*}

Finally,
\begin{align*}
    \left|\E_{\eta\sim\mcT_{-h}\mu}[\varphi(\eta)]-\E_{\eta\sim\mu}[\varphi(\eta)]+h\E_{\eta\sim\mu}[(|\eta|-\iota(\Lambda))\varphi(\eta)]\right|&=
    \left|\int (g(\eta)-1+h(|\eta|-\iota(\Lambda))\varphi(\eta))d\mu(\eta)\right|\\
    &\le \int |g(\eta)-1+h(|\eta|-\iota(\Lambda))|d\mu(\eta)\\
    &\le \int o(h)d\mu(\eta)=o(h)
\end{align*}

For $A\in\Omega$,
\begin{align*}
    \E_{\eta\sim\mcT_{-h}\mu^{+A}}[\varphi(\eta)]&=\E_{\eta\sim\mcT_{-h}\mu}[\varphi(\eta\cup A)]\\
    &=
    \E_{\eta\sim\mu}[\varphi(\eta\cup A)]-h\E_{\eta\sim\mu}[(|\eta|-\iota(\Lambda))\varphi(\eta\cup A)]+o(h)\\
    &=
    \E_{\eta\sim\mu}[\varphi(\eta\cup A)]-h\E_{\eta\sim\mu}[(|\eta\cup A|-\iota^{+A}(\Lambda))\varphi(\eta\cup A)]+o(h)\\
    &=\E_{\eta\sim\mu^{+A}}[\varphi(\eta)]-h\E_{\eta\sim\mu^{+A}}[(|\eta|-\iota^{+A}(\Lambda))\varphi(\eta)]+o(h)
\end{align*}

\end{proof}

\begin{corollary}\label{lem:tilt_first_order_squared}
    In the setting of \Cref{lem:tilt_first_order},
\begin{align*}
    \E_{\mcT_{-h}\mu^{+A}}[\varphi]^2=\E_{\mu^{+A}}[\varphi]^2-2h\E_{\mu^{+A}}[\varphi]\E_{\mu^{+A}}[(|\eta|-\iota^{+A}(\Lambda))\varphi]+o(h)
\end{align*}
\end{corollary}

\subsection{Construction of the Localization Scheme}\label{subsec: construction of loc scheme}
First we will construct the localization scheme and then prove necessary facts about it. 
\\
Let $\nu$ be a point process (i.e. random set of points) on $\Lambda\subseteq\mathbb{X}$. For us, $\nu$ will be a Gibbs point process on $\Lambda$ at activity $\lambda$. Our setup of the negative fields localization scheme is slightly different from the definition in \cite{CE22}, as they use $\{-1,1\}$ coordinates while we are using $\{0,1\}$ coordinates.

\begin{definition}[Negative-fields localization (continuous space)]
    The negative-fields localization process is a martingale of probability measures $(\nu_t)_t$ that can be written as
    \begin{align*}
    \nu_t&\defeq\mcT_{-t}\mcR_{A(t)}\nu
    \end{align*}
    where $A(t)$ is an almost-surely increasing process of finite subsets of $\Lambda$.
    
    Essentially, points will be randomly added to $A(t)$ to balance out the negative tilt in expectation. We will give a more rigorous construction of $A(t)$ below.
\end{definition}

    $A(t)$ is a time-inhomogeneous spatial birth process. Such processes were constructed in \cite{BP22}, but their construction is more complicated as they allow for deaths. We will explicitly construct the process here in a simpler way---we will be using details of the construction in the proof.

The following remark gives intuition for the choice of $A(t)$.
\begin{remark}
    In the finite (i.e. hard core) case, both pinning and applying an infinitesimal external field can be viewed as linear tilts:
\begin{align*}
    \mcR_{+i}\nu(x)&=\nu(x)\dfrac{x(i)}{\iota(i)}=\nu(x)(1+\dfrac{x(i)-\iota(i)}{\iota(i)})&x\in\{0,1\}^{\Lambda}\\
    \mcT_{-h}\nu(x)&=\nu(x)(1-h(|x|_1-\iota(\Lambda)))+o(h)
\end{align*}
where $\iota(i)=\E_{x\sim \nu}[x(i)]$, and $x(i)=1[i\in x]$.

Heuristically, to ensure that $\E[\nu_{t+h}|\nu_t]=\nu_t$, we want to balance out the negative tilts in expectation  by pinning each coordinate $i$ with probability $p_i$ such that
\begin{align*}
    \sum_{i\in \Lambda}p_i \left(\dfrac{x(i)-\iota(i)}{\iota(i)}\right)-h(|x|_1-\iota(\Lambda))&=0 
\end{align*}
This solves to $p_i=h\iota(i)$.
\end{remark}

For us $A(t+h)\setminus A(t)$ is approximately a Poisson point process with intensity equal to $h$ times the intensity of $\nu_t$ (excluding pinned points).
To avoid technical issues, we will construct $A(t)$ as a dependent thinning of a Poisson point process. 

We now formally construct $A(t)$.

\lemanotherdeftilt*

\begin{restatable}[$A(t)$]{definition}{defnA}\label{def: A(t)}
    Let $X$ be drawn from a Poisson point process on $\Lambda\times[0,\infty)\times[0,1]$ with intensity $\lambda$ (this will be our probability space). Intuitively, the first coordinate is location, the second coordinate is time, and the third coordinate is a random number in $[0,1]$ used to perform a random choice without using external sources of randomness.
    
    Almost surely, $X_{\le t}:=X\cap \{ \Lambda\times[0,t]\times[0,1]\}$ is finite for all $t\in[0,\infty)$. Let $\mcF_t$ be the $\sigma$-algebra generated by $X_{\le t}$,  $(\mcF_t)_{t\ge 0}$ is a filtration. Let 
    $(x_1,t_1,l_1),(x_2,t_2,l_2),\ldots,(x_k,t_k,l_k),\ldots,$ be the points of $X$ in the increasing order of time $t_i$.
    Let $t_0:=0$. We have $0=t_0\le t_1\le t_2\le\ldots$, and
    $A(t)$ will be constant on $[t_i,t_{i+1})$ for $i\in \N\cup \{ 0\}$.
    We will set $A(t_0)=\emptyset$, and for $i\in \N$,
    \begin{align*}
        A(t_i)&\defeq\begin{cases}
            A(t_{i-1})\cup\{x_i\}&\text{if $l_i\le \dfrac{\tilde{\zeta}_t(x)}{\lambda}$}\\
            A(t_{i-1})&\text{otherwise}
        \end{cases}
    \end{align*}
    
    where $\tilde{\zeta}_t$ is the one-point density of $\mcT_{-t_i}\tilde{\mcR}_{A(t_{i-1})}\nu$.
\end{restatable}
    
    (The purpose of the $l_i\in[0,1]$ is to make a ``random decision'' to include $x_i$ with probability $\tilde{\iota}_t(x_i)/\lambda$, without needing external randomness that will need to be incorporated into the state space.).  
    Note that $A(t)$ is a (deterministic) function of $X_{\le t}$. 
    \begin{lemma}
        $A(t)$ is a measurable function of $X_{\le t}$.
    \end{lemma}
The proof is deferred to \Cref{sec:measurability}.

\begin{remark}
    Note that in hard spheres model no overlapping spheres can be added since the intensity changes to prevent that after the addition of each sphere. 
\end{remark}

We conclude with an important property of the localization scheme, proved in the next subsection.
\begin{lemma}
For any measurable $\mcB\subseteq\Omega$, $\nu_t(\mcB)$ is a martingale with respect to the filtration $(\mcF_t)_{t\ge 0}$
\end{lemma}

\subsubsection{Proof of the Martingale Property}\label{subsec: martingale}    

As before, $o(h)$ denotes an anonymous function $f(h)$ such that $\lim_{h\to 0}\dfrac{f(h)}{h}=0$, and note that $f(h)$ may depend on $\lambda$ or $|\Lambda|$, which are treated as constant here. 
Let $\tilde{\iota}_t$ be the intensity measure of $\tilde{\nu}_t=\mcT_{-t}\tilde{\mcR}_{A(t)}\nu$, and $\tilde{\zeta}_t$ be the corresponding one-point density.

Let $Y$ denote the points in $X$ with $t_i\in(t,t+h)$, therefore $Y$ is a Poisson point process on $\Lambda\times(t,t+h)\times[0,1]$ with intensity $\lambda$, independent of $X_{\le t}$, and distributed with rate $\lambda h|\Lambda|$:
\begin{align*}
    |Y|\begin{cases}
        =0&\text{with probability $e^{-\lambda h|\Lambda|}$}\\
        =1&\text{with probability $\lambda h|\Lambda|e^{-\lambda h|\Lambda|}$}\\
        >1&\text{with probability $o(h)$}
    \end{cases}
\end{align*}

Conditioned on $|Y|=1$, the unique point $(x_i,t_i,l_i)$ in $Y$ is draw uniformly from $\Lambda\times(t,t+h)\times[0,1]$. Thus, $x_i\sim \Lambda$ uniformly, and is then kept with probability $\tilde{\zeta}_t(x_i)/\lambda$.

Let $\psi:\Omega\to[0,1]$ be a measurable function.
Then,
\begin{align}
    \E[\psi(A(t+h))|X_{\le t}]&=e^{-\lambda h|\Lambda|}\psi(A(t))+\lambda h|\Lambda|e^{-\lambda h|\Lambda|}\cdot \dfrac{1}{|\Lambda|}\int_{\Lambda}\dfrac{\tilde{\zeta}_t(x)}{\lambda}\psi(A(t)\cup\{x\})+(1-\dfrac{\tilde{\zeta}_t(x)}{\lambda})\psi(A(t))dx+o(h)\nonumber\\
    &=(1-\lambda h|\Lambda|)\psi(A(t))+\lambda h|\Lambda|\cdot \dfrac{1}{|\Lambda|}\int_{\Lambda}\dfrac{\tilde{\zeta}_t(x)}{\lambda}\psi(A(t)\cup\{x\})+(1-\dfrac{\tilde{\zeta}_t(x)}{\lambda})\psi(A(t))dx+o(h)\nonumber\\
    &=\psi(A(t))+h\int_{\Lambda}\psi(A(t)\cup\{x\})-\psi(A(t))d\tilde{\iota}_t(x)+o(h)\nonumber\\
    &=(1-h\tilde{\iota}_t(\Lambda))\psi(A(t))+h\int_{\Lambda}\psi(A(t)\cup\{x\})d\tilde{\iota}_t(x)+o(h) \label{eq: martingale_lemma}
\end{align}

Intuitively, the distribution of $A(t+h)$ has TV distance $o(h)$ from
\begin{align*}
    \begin{cases}
    A(t)\cup\{x\}&\text{with probability density $h\tilde{\zeta}_t(x)$}\\
    A(t)&\text{with probability $1-h\tilde{\iota}_t(\Lambda)$}
    \end{cases}
\end{align*}

\begin{claim}
Let $\nu_t$ and $\mcF_t$ be defined as above. Then, for any measurable $\mcB\subseteq\Omega$, $\nu_t(\mcB)$ is a martingale with respect to the filtration $\mcF_t$.
\end{claim}

\begin{proof}
Let $t\in[0,T]$ and $\varphi(\eta)=1_{\mcB}(\eta)$.

Let $\psi(A)=\E_{\eta\sim\mcR_{A}\mcT_{-t-h}\nu}[\varphi(\eta)]$. 
Let $\widehat{\varphi}(\eta):=\varphi(\eta\cup A(t))$. Then,
\begin{align*}
    \psi(A(t+h))&=\E_{\eta\sim\nu_{t+h}}[\varphi(\eta)]\\
    \psi(A(t))]
    &=\E_{\eta\sim\mcT_{-h}\tilde{\nu}_t}[\widehat{\varphi}(\eta)]\\
    \psi(A(t)\cup\{x\})
    &=\E_{\eta\sim\mcR_{+x}\mcT_{-h}\tilde{\nu}_t}[\widehat{\varphi}(\eta)]
\end{align*}

Hence, by \Cref{eq: martingale_lemma},
\begin{align*}
    \E[\E_{\eta\sim\nu_{t+h}}[\varphi(\eta)]|\mcF_t]
    &=(1-h\tilde{\iota}_t(\Lambda))\E_{\eta\sim\mcT_{-h}\tilde{\nu}_t}[\widehat{\varphi}(\eta)]+h\int_{\Lambda}\E_{\eta\sim\mcR_{+x}\mcT_{-h}\tilde{\nu}_t}[\widehat{\varphi}(\eta)]d\tilde{\iota}_t(x)+o(h)\\
    &\overset{(i)}=(1-h\tilde{\iota}_t(\Lambda))\E_{\eta\sim\mcT_{-h}\tilde{\nu}_t}[\widehat{\varphi}(\eta)]+h\int_{\Lambda}\E_{\eta\sim\mcR_{+x}\tilde{\nu}_t}[\widehat{\varphi}(\eta)]d\tilde{\iota}_t(x)+o(h)\\
    &\overset{(ii)}=
    (1-h\tilde{\iota}_t(\Lambda))(\E_{\eta\sim\tilde{\nu}_t}[\widehat{\varphi}(\eta)]-h\E_{\eta\sim\tilde{\nu}_t}[(|\eta|-\tilde{\iota}_t(\Lambda))(\widehat{\varphi}(\eta))])+h\E_{\tilde{\nu}_t}[|\eta|\widehat{\varphi}(\eta)]+o(h)\\
    &\overset{(iii)}=\E_{\eta\sim\widehat{\nu_t}}[\widehat{\varphi}(\eta)]+o(h)\\
    &=\E_{\eta\sim\nu_t}[\varphi(\eta)]+o(h)
\end{align*}
where the $o(h)$ are uniform in $t$ and $\varphi$.

In (i), we can drop the $\mcT_{-h}$ because that introduces a lower order error (\Cref{lem:tilt_first_order}) which is uniformly bounded over all $x$.

In (ii), we use the linear approximation of negative tilts (\Cref{lem:tilt_first_order}) and \Cref{lem:integrate_pinning}.

In (iii), we used that $|\E_{\eta\sim\tilde{\nu}_t}[(|\eta|-\tilde{\iota}_t(\Lambda))\widehat{\varphi}(\eta)]|\le 2\lambda|\Lambda|$ uniformly in $t$ and $\varphi$.

Letting $\varphi(\eta)=1_{\mcB}(\eta)$, we get that
\begin{align*}
    \E[\nu_{t+h}(\mcB)|\mcF_t]&=\nu_t(\mcB)+o(h)
\end{align*}

Finally, for $t,s\ge 0$ with $s+t\le T$, since the $o(h)$ was uniform in $t$,
\begin{align*}
    \E[\nu_{t+s}(\mcB)|\mcF_t]&=\E[\nu_t(\mcB)|\mcF_t]+\sum_{j=1}^k\E[\nu_{t+sj/k}(\mcB)-\nu_{t+s(j-1)/k}(\mcB)]|\mcF_t]\\
    &=\nu_t(\mcB)+\sum_{j=1}^k\E[\E[\nu_{t+sj/k}(\mcB)-\nu_{t+s(j-1)/k}(\mcB)|\mcF_{t+s(j-1)/k}]|\mcF_t]\\
    &=\nu_t(\mcB)+\sum_{j=1}^ko(s/k)\\
    &=\nu_t(\mcB)+sko(1/k)
\end{align*}
so taking $k\to\infty$ shows that $\E[\nu_{t+s}(\mcB)|\mcF_t]=\nu_t(\mcB)$.
\end{proof}

\section{Spectral Gap at Higher Activity}\label{sec: spec gap at high}

In this section we obtain a spectral gap of the Continuum Glauber for higher activity. We do this by showing the approximate variance conservation (Subsection \ref{subsec: variance conservation}), which implies that a
repulsive point process can be decomposed into a mixture of repulsive point processes of lower activity. We combine the knowledge of the spectral gap at lower activity together with the measure decomposition via a boosting theorem (Subsection \ref{sec: boosting}).

\subsection{Approximate Variance Conservation}\label{subsec: variance conservation}

Recall that for $\nu_0\in\mcP(\Lambda)$, the negative fields localization process is defined by $\nu_t\defeq\mcR_{A(t)}\mcT_{-t}\nu_0$. Let $\iota_t$ be the intensity measure of $\nu_t$. 
Before we prove the main theorem, let us introduce a useful lemma first.

\begin{lemma}\label{lem:pinning_covariance_bound}
Let $\nu\in\mcP(\Lambda)$ have intensity measure $\iota$ and $\Psi$ be the influence operator for $\nu$. Suppose $\varphi:\Omega\to \R$ is bounded and $\E_{\nu}[\varphi]=0$. 
Then,
\begin{align*}
    \int \E_{\mcR_{+x}\nu}[\varphi]^2d\iota(x)&\le \rho(\Psi)\cdot  \E_{\nu}[\varphi^2]
\end{align*}
\end{lemma}

\begin{proof}
    Wlog by scaling, we may assume that $\varphi:\Omega\to[-1,1]$. 

Let $w:\Lambda\to\R$ be defined by $w(x)\defeq\E_{\eta\sim\mcR_{+x}\nu}[\varphi(\eta)]$.
Then,
\begin{align*}
    (\int_{\Lambda} \E_{\mcR_{+x}\nu}[\varphi]^2d\iota(x))^2&=(\int_{\Lambda} w(x)\E_{\eta\sim\mcR_{+x}\nu}[\varphi(\eta)]d\iota(x))^2\\
    &=\E_{\eta\sim\nu}[\eta(w)\varphi(\eta)]^2&&\text{by \Cref{lem:integrate_pinning}}\\
    &=\E_{\eta\sim\nu}[(\eta(w)-\iota(w))\varphi(\eta)]^2&&\text{since $\E_{\eta\sim\nu}[\varphi(\eta)]=0$}\\
    &\le \E_{\eta\sim\nu}[(\eta(w)-\iota(w))^2]\E_{\eta\sim\nu}[\varphi(\eta)^2]&&\text{by Cauchy-Schwarz}\\
    &=\langle w,\Psi w\rangle_{\iota}\cdot\E_{\nu}[\varphi^2]&&\text{by \Cref{lem:influence_covariance}}
\end{align*}
Thus,
\begin{align*}
    \int_{\Lambda} \E_{\mcR_{+x}\nu}[\varphi]^2d\iota(x)=\langle w,w\rangle_{\iota}&\le \dfrac{\langle w,\Psi w\rangle_{\iota}}{\langle w,w\rangle_{\iota}}\cdot\E_{\nu}[\varphi^2]\le \rho(\Psi)\cdot\E_{\nu}[\varphi^2]
\end{align*}
where the last inequality is by
\Cref{lem:influence_rayleigh_quotient_le_iota_operator_norm}.
\end{proof}

The following theorem bounds the decrease of the  variance of a test function along the negative-fields localization process.

\begin{theorem}[Approximate variance conservation]\label{thm:variance_conservation}

Let $\nu_t=\mcR_{A(t)}\mcT_{-t}\nu_0$ be the negative fields localization processes started from $\mu_{\bm\lambda}$, with associated filtration $(\mcF_t)_t$. Suppose $\mu_{\bm\lambda}$ is $\kappa$-spectrally independent.
Then, for any measurable $\varphi:\Omega\to[-1,1]$, for $t\ge 0$,
    \begin{align*}
    \E[\Var_{\nu_{t+h}}[\varphi]|\mcF_t]-\Var_{\nu_t}[\varphi]&\ge -\kappa h\Var_{\nu_t}[\varphi]+o(h).
\end{align*}
\end{theorem}

\begin{proof}

We may assume wlog that $\E_{\nu_t}[\varphi]=0$ (otherwise apply the argument to $\varphi-\E_{\nu_t}[\varphi]$).

Let $\psi(A)=\E_{\eta\sim\mcR_{A}\mcT_{-t-h}\nu}[\varphi(\eta)]^2$.
Let $\widehat{\varphi}(\eta):=\varphi(\eta\cup A(t))$. Then,
\begin{align*}
    \psi(A(t+h))&=\E_{\eta\sim\nu_{t+h}}[\varphi(\eta)]^2\\
    \psi(A(t))]
    &=\E_{\eta\sim\mcT_{-h}\tilde{\nu}_t}[\widehat{\varphi}(\eta)]^2\\
    \psi(A(t)\cup\{x\})
    &=\E_{\eta\sim\mcR_{+x}\mcT_{-h}\tilde{\nu}_t}[\widehat{\varphi}(\eta)]^2
\end{align*}
Note that $\E_{\mcT_{-t}\tilde{\nu}}[\widehat{\varphi}]^2=o(h)$ by \Cref{lem:tilt_first_order_squared}.

By (\ref{eq: martingale_lemma}),
\begin{align*}
    \E[\E_{\nu_{t+h}}[\varphi]^2|\mcF_t]&=(1-h\tilde{\iota}_t(\Lambda))\E_{\mcT_{-h}\tilde{\nu}_t}[\widehat{\varphi}]^2+h\int_{\Lambda}\E_{\mcR_{+x}\mcT_{-h}\tilde{\nu}_t}[\widehat{\varphi}]^2d\tilde{\iota}_t(x)+o(h)\\
    &=h\int_{\Lambda}\E_{\mcR_{+x}\mcT_{-t}\tilde{\nu}_t}[\widehat{\varphi}]^2d\tilde{\iota}_t(x)+o(h)&\text{as $\E_{\mcT_{-t}\tilde{\nu}}[\widehat{\varphi}]^2=o(h)$}\\
    &=h\int_{\Lambda}\E_{\mcR_{+x}\tilde{\nu}_t}[\widehat{\varphi}]^2d\tilde{\iota}_t(x)+o(h)&\text{by \Cref{lem:tilt_first_order_squared}}\\
    &\le hC\E_{\tilde{\nu}_t}[\widehat{\varphi}^2]+o(h)&\text{by \Cref{lem:pinning_covariance_bound}}\\
    &=hC\E_{\nu_t}[\varphi^2]+o(h)
\end{align*}

Then, by the martingale property,
\begin{align*}
    \E[\Var_{\nu_{t+h}}[\varphi]|\mcF_t]-\Var_{\nu_t}[\varphi]&=\underbrace{(\E[\E_{\nu_{t+h}}[\varphi^2]|\mcF_t]-\E_{\nu_t}[\varphi^2])}_{=0}-(\E[\E_{\nu_{t+h}}[\varphi]^2|\mcF_t]-\underbrace{\E_{\nu_t}[\varphi]^2}_{=0})\\
    &\ge -Ch \E_{\nu_t}[\varphi^2]+o(h)\\
    &=-Ch\Var_{\nu_t}[\varphi]+o(h)
\end{align*}
\end{proof}

\begin{corollary}
Let $\mu_{\bm\lambda}$ for some bounded $\bm\lambda:\Lambda\to[0,\infty)$ be $\kappa$-spectrally independent, and $c\in(0,1)$.

Then, there is a random $A\in\Omega$ such that for all measurable $\varphi:\Omega\to[0,1]$,
\begin{align*}
    \E[\mcR_A\mu_{c\bm\lambda}]&=\mu_{\bm\lambda}\quad \mbox{ and } \quad \E[\Var_{\mcR_A\mu_{c\bm\lambda}}[\varphi]]
    \ge c^{\kappa}\Var_{\nu_{t_0}}[\varphi]
\end{align*}
\end{corollary}

\begin{proof}
Let $\nu_t$ be the negative fields localization scheme starting from $\nu_0=\mu_{\bm\lambda}$. Let $\tau=\log(1/c)$ and $\nu_{\alpha}=\nu_\tau$ (this is a random measure).

Then, for any measurable $\mcB\subseteq\Omega$ by the martingale property
\begin{align*}
    \E[\nu_{\alpha}(\mcB)]=\nu(\mcB)
\end{align*}

Also, by \Cref{thm:variance_conservation} and a standard integration argument (\Cref{lem:gronwall_variant}), we have
\begin{align*}
    \E[\Var_{\nu_{\alpha}}[\varphi]]&\ge e^{-\kappa\tau}\Var_{\nu_{t_0}}[\varphi]\\
    &=c^\kappa\Var_{\nu_{t_0}}[\varphi]
\end{align*}
\end{proof}

As a special case, we get \Cref{thm:measure_decomposition}.

\measuredecomposition*

\subsection{Spectral Gap}\label{sec: boosting}
In this subsection we will prove a lower bound on the spectral gap for higher activity using the approximate variance conservation established in the previous subsection and the boosting theorem.

\begin{theorem}[Boosting - Variance]\label{thm:boosting}

Let $\nu_0$ be a Gibbs point process with repulsive pair potentials, and $(\nu_{\alpha})_{\alpha}$ be a distribution of Gibbs point processes with repulsive pair potentials such that $\E[\nu_{\alpha}]=\nu_0$.

Let $\mcE_{\nu_0}$ be the Dirichlet form of Continuum Glauber on $\nu_0$, and $\mcE_{\nu_{\alpha}}$ be the Dirichlet form of Continuum Glauber on $\nu_{\alpha}$.

Suppose that for all bounded measurable $f:\Omega\to\R$, we have
\begin{itemize}
    \item Approximate conservation of variance:
    \begin{align*}
        \E[\Var_{\nu_{\alpha}}[f]]&\ge \epsilon \Var_{\nu_0}[f]
    \end{align*}
    \item Lower bound on spectral gap for $\nu_{\alpha}$:
    Almost surely over $\nu_{\alpha}$,
    \begin{align*}
        \mcE_{\nu_{\alpha}}(f,f)&\ge \delta \Var_{\nu_{\alpha}}[f]
    \end{align*}
\end{itemize}

Then, we have a lower bound on the spectral gap for $\nu$:
    \begin{align*}
        \mcE_{\nu_{0}}(f,f)&\ge \epsilon \delta \Var_{\nu_0}[f]
    \end{align*}
\end{theorem}

\begin{proof}
According to \cite{KL03} equation (3.2), the Dirichlet form for Continuum Glauber is
\begin{align*}
\mcE_{\nu_{\alpha}}(f,g)&=\int_{\Omega}\sum_{x\in\eta}(f(\eta\setminus\{x\})-f(\eta))(g(\eta\setminus\{x\})-g(\eta))d\nu_{\alpha}(\eta)
\end{align*}

Therefore, when $\E[\nu_{\alpha}]=\nu_0$,
\begin{align*}
    \E[\mcE_{\nu_{\alpha}}(f,g)]=\mcE_{\nu_{0}}(f,g)
\end{align*}

Hence,
\begin{align*}
\mathcal{E}^{(0)}(f,f)&=\E[\mcE_{\nu_{\alpha}}(f,f)]\\
    &\ge \E[\delta \Var_{\nu_{\alpha}}[f]]\\
    &\ge \delta \epsilon \Var_{\nu_0}[f]
\end{align*}
\end{proof}

We note that adding points to the measure (which are ignored by continuum Glauber) does not affect the spectral gap.
\begin{lemma}\label{lem:spectral_gap_after_addition}
    Suppose continuum Glauber for a repulsive Gibbs point process $\mu_{\bm\lambda}$ has spectral gap at least $\gamma>0$. Then, for any $A\in\Omega$, continuum Glauber on $\mu_{\bm\lambda}^{+A}$ also has spectral gap at least $\gamma$.
\end{lemma}
\begin{proof}
    The spectral gap for $\mu_{\bm\lambda}$ says that for any bounded measurable $\hat{f}:\Omega\to\R$,
    \begin{align*}
        \int_{\Omega}\sum_{x\in\eta}(\hat{f}(\eta)-\hat{f}(\eta\setminus\{x\}))^2d\mu_{\bm\lambda}(\eta)\ge \gamma\Var_{\eta\sim\mu}(\hat{f}(\eta))
    \end{align*}
    Fix $A\in\Omega$. For any bounded measurable $f:\Omega\to\R$, if we set $\hat{f}(\eta):=f(\eta\cup A)$ above, we get
    \begin{align*}
        \int_{\Omega}\sum_{x\in\eta}(f(\eta)-f(\eta\setminus\{x\}))^2d\mu_{\bm\lambda}^{+A}(\eta)\ge \gamma\Var_{\eta\sim\mu^{+A}}(f(\eta))
    \end{align*}
    
\end{proof}

We can now prove \Cref{thm:spectral_gap}, which we restate for convenience.

\thmspectralgap*

\begin{proof}
To apply the boosting theorem to get the spectral gap at higher activity, we simply need to combine the known bounds for the spectral gap and use approximate variance conservation. 

Recall that the spectral gap for activities bounded by $\lambda_1=(1-\delta)\frac{1}{C_{\phi}}$ is at least $\delta$ (\Cref{thm:low_activity_spectral_gap}). By \Cref{lem:spectral_gap_after_addition}, the spectral gap of $\mcR_A\mu_{\lambda_1}$ is also at least $\delta$, for all $A\in\Omega$.
Hence by 
\Cref{thm:measure_decomposition} and \Cref{thm:boosting}, the spectral gap for $\lambda_0$ is at least
\[ 
\left(\frac{\lambda_1}{\lambda_0}\right)^{\kappa}\delta \ge \left(\frac{(1-\delta)1/C_{\phi}}{e/\Delta_{\phi}}\right)^{\kappa}\delta = \left(\frac{\Delta_{\phi}}{eC_{\phi}}\right)^{\kappa} (1-\delta)^{\kappa}\delta
\]
where we used $\lambda_1 = (1-\delta)\frac{1}{C_{\phi}}$ and $\lambda_0 < e/\Delta_{\phi}$. Taking any $\delta\in(0,1)$ to be any constant completes the proof.

\end{proof}
We show the mixing time that follows from it in the next section.

For the hard sphere model with activity $\lambda\le e^{-\delta}\frac{e}{C_{\phi}}$, the spectral gap
can be computed more explicitly. This is \Cref{thm:spectral_gap_hard_spheres}, which we restate for the reader's convenience.

\thmspectralgaphardspheres*

\begin{proof}

Recall that the spectral gap of continuum Glauber for $\lambda_1=(1-\alpha)\frac{1}{C_{\phi}}$ is at least $\alpha$ (\Cref{thm:low_activity_spectral_gap}). 
Hence by 
\Cref{thm:measure_decomposition} and \Cref{thm:boosting}, the spectral gap of continuum Glauber for $\lambda_0\le \frac{e}{C_{\phi}}$ is at least
\[ 
\left(\frac{\lambda_1}{\lambda_0}\right)^{\kappa}\alpha \ge \left(\frac{(1-\alpha)/C_{\phi}}{e/C_{\phi}}\right)^{\kappa}\alpha = e^{-\kappa} (1-\alpha)^{\kappa}\alpha
\]

Choosing $\alpha=1/2$, we see that the spectral gap for $\lambda_0$ is at least $\frac{1}{2}e^{-2\kappa}$.

For $\lambda_0\le e^{-\delta}\frac{e}{C_{\phi}}$, we have $\kappa\le 1+\dfrac{e2^{d+1}d!}{\delta^d}$ (\Cref{cor:infinity_op_norm_for_hard_spheres}), so the spectral gap is at least
\begin{align*}
    \frac{1}{2}\exp(-2(1+\dfrac{e2^{d+1}d!}{\delta^d}))
\end{align*}

\end{proof}

\section{Mixing and Run Time}\label{sec:mixing_time}
The main theorem of this section will be the following:
\begin{theorem}\label{thm:mixing:sampling_algorithm_guarantees}
For $T>0$, the number of iterations of \Cref{alg:birth_death} is (up to constant factors) stochastically dominated by $\Pois(\bm\lambda(\Lambda)T)$.

If there exists some $\gamma>0$ such that
\begin{align*}
    \mcE(f,f)&\ge \gamma \Var_{\mu}(f)&\forall\text{ bounded measurable }f:\Omega\to\R
\end{align*}
then the output of the algorithm will have a distribution $\mu_T$ such that
\begin{align*}
d_{TV}(\mu_T,\mu)\le \dfrac{1}{2}e^{(\log B)/2-\gamma T}
\end{align*}
where $\frac{d\mu_0}{d\mu}\le B$. If starting with an empty state, $B$ can be taken to be $e^{\bm\lambda(\Lambda)}$.
\end{theorem}

We do not assume finite-range or repulsive-pair potentials for this section, only that $\nabla_x^+H(\eta)$ can be computed efficiently.

The proof of the main theorem (which can be found at \ref{proof: main mix theorem}) requires several ingredients that we introduce in the subsequent sections. 
We first define a Markov chain on (uncountable state space) $\Omega$ that models the steps of the algorithm. We note some of its properties and in Subsection \ref{subsec: bounding numbr of iterations} bound the number of iterations needed for it. In Subsection \ref{subsec:cont_time_MC} we look at the corresponding continuous time Markov process and show its main properties.

We can first rewrite \Cref{alg:birth_death} as a Markov chain (\Cref{alg:continuum_glauber_chain}). 

\begin{algorithm}[H]
\caption{Simulate continuum Glauber for $T>0$ time units}
\label{alg:continuum_glauber_chain}
\begin{algorithmic}[1]
\State $Y_0 \gets \emptyset$
\State $t_0\gets 0$
\For{$k=1,2,3,\ldots$}
    \State Sample $\gamma_k\sim \Exp(1)$ independently
    \State $t_k \gets t_{k-1}+\frac{\gamma_k}{|Y_{k-1}| + \bm\lambda(\Lambda)}$
    \If{$t_k>T$}
        \State \Return $Y_{k-1}$
    \EndIf
    \WithProb{$\frac{|Y_{k-1}|}{|Y_{k-1}| + \bm\lambda(\Lambda)}$}
        \State Choose $x \in Y_{k-1}$ uniformly at random
        \State $Y_k \gets Y_{k-1} \setminus \{x\}$\label{lst:alg2:line:death}
    \Otherwise
        \State Choose $y \in \Lambda$ uniformly at random \label{lst:alg2:line:attempt-birth}
        \WithProb{$e^{-\nabla_x^+H(\eta)}$}
            \State $Y_k \gets Y_{k-1}\cup \{y\}$\label{lst:alg2:line:birth}
        \Otherwise
            \State $Y_k \gets Y_{k-1}$
        \EndWithProb
    \EndWithProb
\EndFor
\end{algorithmic}
\end{algorithm}

For convenience of analysis, let us define $Y_k$ and $\gamma_k$ for $t_k>T$ (after the algorithm stops). Note that $(Y_k)_{k\ge 0}$ is a Markov chain with uncountable state space $\Omega$ starting at $\emptyset$ with transition kernel
\begin{align*}
    \rho(\eta,\mcB)&=\frac{1}{|\eta|+\bm\lambda(\Lambda)}\left(\sum_{x\in\eta}1_{\mcB}(\eta\setminus\{x\})+\int_{\Lambda}1_{\mcB}(\eta)+e^{-\nabla_x^+H(\eta)}(1_{\mcB}(\eta\cup\{x\})-1_{\mcB}(\eta))\bm\lambda(x)dx\right)
\end{align*}
for all $\text{measurable }\mcB\subseteq\Omega$, 
and
\begin{align*}
\gamma_1,\gamma_2,\ldots&\overset{i.i.d.}{\sim}\Exp(1)&\text{independently of $Y$.}
\end{align*}
It is easy to see that
\begin{align*}
    t_k&=\sum_{j=1}^k\frac{\gamma_j}{|Y_{j-1}|+\bm\lambda(\Lambda)}
\end{align*}

The algorithm will return $Y_{k^*}$ where $k^*=\sup\{k\ge 0:t_k<T\}$.

It remains to bound the number of iterations $k^*$ of the algorithm and show the distribution of $Y_{k^*}$ is close in TV distance to the Gibbs point process distribution.
For technical reasons, we will first bound the number of iterations. This will show that $t_k\to\infty$ almost surely, which is used for the other part.

\subsection{Bounding number of iterations}\label{subsec: bounding numbr of iterations}

Let $I=\{k\in\Z_{>0}:|Y_k|\ge |Y_{k-1}|\}$ be the iterations for which we have an attempted birth (i.e. we reach line \ref{lst:alg2:line:attempt-birth}).

\begin{restatable}{lemma}{timediffind}\label{lem:mixing:attempt_birth_memoryless}
    Let $i\ge 0$. Then, if we let $j=\min\{k>i:|Y_k|\ge |Y_{k-1}|\}$,
\begin{align*}
    t_j-t_i\sim\Exp(\bm\lambda(\Lambda))
\end{align*}
independently of $\gamma_1,\ldots,\gamma_i$ and $Y_0,\ldots,Y_i$.
\end{restatable}

The following is the main idea of the proof by induction. When $|Y_i|=0$, $j=i+1$, and $t_{i+1}-t_i\sim\Exp(\bm\lambda(\Lambda))$. When $|Y_i|\ge 1$, we have $t_{i+1}-t_i\sim\Exp(|Y_i|+\bm\lambda(\Lambda))$, and with probability $\frac{\bm\lambda(\Lambda)}{\bm\lambda(\Lambda)+|Y_i|}$, $j>i+1$ and $t_j-t_{i+1}\sim\Exp(\bm\lambda(\Lambda))$ by the inductive hypothesis. A calculation then shows that $t_j-t_i\sim\Exp(\bm\lambda(\Lambda))$. We formalize the argument in

\begin{lemma}\label{lem:mixing:attempted_births_poisson_dominated}
    $B:=\{t_j:j\in I\}$ is distributed according to a Poisson point process of intensity $\bm\lambda(\Lambda)$ on $\R_{>0}$.
\end{lemma}
\begin{proof}

    Set $i_0=0$ and let $i_1<i_2<\cdots$ be the elements of $I$, so that
    \begin{align*}
        I=\{i_1,i_2,\ldots\}
    \end{align*}

    By a well-known characterization of a Poison point process of intensity $\bm\lambda(\Lambda)$ (see e.g. Theorem 7.2 from \cite{Last_Penrose_2017}), it suffices to show that $t_{i_1}-t_{i_0},t_{i_2}-t_{i_1},\ldots\overset{iid}{\sim}\Exp(\bm\lambda(\Lambda))$.
    
    Let $k\ge 0$.
    The idea is that conditioned on $i_1,\ldots,i_k$ and $\gamma_1,\ldots,\gamma_{i_k}$, we have $t_{i_{k+1}}-t_{i_k}\sim\Exp(\bm\lambda(\Lambda))$ by \Cref{lem:mixing:attempt_birth_memoryless}, but $t_{i_1}-t_{i_0},\ldots,t_{i_k}-t_{i_{k-1}}$ will be deterministic, meaning they are independent.
    We now make this argument rigorous.

Let $A\subseteq\R^k$ and $C\subseteq\R$ be Borel measurable, and let $i'\in\Z_{\ge 0}$.
Since $\{(t_{i_1}-t_{i_0},\ldots,t_{i_k}-t_{i_{k-1}})\in A\wedge i_k=i'\}$ belongs to the $\sigma$-algebra generated by $\gamma_1,\ldots,\gamma_{i'}$ and $Y_1,\ldots,Y_{i'}$,
 we can use \Cref{lem:mixing:attempt_birth_memoryless} to show that
\begin{align*}
    &\Pr((t_{i_1}-t_{i_0},\ldots,t_{i_k}-t_{i_{k-1}})\in A\wedge (i_k=i')\wedge (t_{i_{k+1}}-t_{i_k}\in C))\\
    &=\Pr((t_{i_1}-t_{i_0},\ldots,t_{i_k}-t_{i_{k-1}})\in A\wedge (i_k=i')\wedge (t_{\min\{i''>i':|Y_{i''}|\ge |Y_{i''-1}|\}}-t_{i'}\in C))\\
    &=\Pr((t_{i_1}-t_{i_0},\ldots,t_{i_k}-t_{i_{k-1}})\in A\wedge (i_k=i'))\Pr(\Exp(\bm\lambda(\Lambda))\in C)
\end{align*}
Summing over all possible $i'$,
\begin{align*}
    &\Pr((t_{i_1}-t_{i_0},\ldots,t_{i_k}-t_{i_{k-1}})\in A\wedge (t_{i_{k+1}}-t_{i_k}\in C))\\
    &=\Pr((t_{i_1}-t_{i_0},\ldots,t_{i_k}-t_{i_{k-1}})\in A)\Pr(\Exp(\bm\lambda(\Lambda))\in C)
\end{align*}
Thus, $t_{i_{k+1}}-t_{i_k}\sim\Exp(\bm\lambda(\Lambda))$ and is independent of $(t_{i_1}-t_{i_0},\ldots,t_{i_k}-t_{i_{k-1}})$.
    As this holds for any $k\ge 0$, we have $t_{i_1}-t_{i_0},t_{i_2}-t_{i_1},\ldots\overset{iid}{\sim}\Exp(\bm\lambda(\Lambda))$.
    Thus, $B$ is distributed according to a Poisson point process of intensity $\bm\lambda(\Lambda)$ on $\R_{>0}$.
\end{proof}

\begin{lemma}\label{lem:mixing:iteration_bound}
$\frac{k^*}{2}$ is stochastically dominated by $\Pois(\bm\lambda(\Lambda)T)$.
\end{lemma}
\begin{proof}
The number of deaths is at most the number of attempted births. Formally, since $|Y_0|=0$ and $|Y_{k^*}|\ge 0$,
\begin{align*}
    0&\le \sum_{j=1}^{k^*}(|Y_k|-|Y_{k-1}|)\\
    &\le |\{1\le k\le k^*:|Y_k|\ge |Y_{k-1}|\}|-|\{1\le k\le k^*:|Y_k|<|Y_{k-1}|\}|\\
    &=2|\{1\le k\le k^*:|Y_k|\ge |Y_{k-1}|\}|-k^*
\end{align*}
Hence, $k^*\le 2|B\cap(0,T)|$, where $B:=\{t_j:j\in I\}$ satisfies $|B\cap(0,T)|\sim\Pois(\bm\lambda(\Lambda)T)$ by \Cref{lem:mixing:attempted_births_poisson_dominated}.
\end{proof}

\begin{corollary}
    $t_k\to\infty$ almost surely.
\end{corollary}
\begin{proof}
    If $t_k\not\to\infty$, then there exists a positive integer $T>0$ such that $\sup\{k\in\Z_{>0}:t_k\le T\}=\infty$. This happens with probability $0$ for any positive integer $T>0$, so the probability of any of these countably many events happening is zero by the union bound.
\end{proof}

\subsection{Continuous-time Markov process}\label{subsec:cont_time_MC}

To show that the distribution of $Y_{k^*}$ is close to the Gibbs point process, we will analyze an associated continuous-time Markov process.

\begin{lemma}\label{lem:implementation:chain_for_process}
Let $(Y_n)_{n\in\Z_{\ge 0}}$ be a Markov chain on state space $\Omega$ with transition kernels $\frac{\beta(x,B)}{\beta(x,\Omega)}$. 
Let $\gamma_1,\gamma_2,\ldots\overset{i.i.d.}\sim\Exp(1)$, and define
\begin{align*}
    \sigma_n\defeq \sum_{k=1}^n\frac{\gamma_k}{\beta(Y_{k-1},\Omega)}~~~~~~~~~~~~~~~~
X_t&\defeq Y_n,~~~t\in[\sigma_n,\sigma_{n+1}),~n\in\Z_+
\end{align*}

Suppose $\lim_{n\to\infty}\sigma_n=\infty$ almost surely.
Then, $(X_t)_{t\ge 0}$ is a pure jump-type Markov process with rate kernel $\alpha(x,B)=\beta(x,B\setminus \{x\})$.
\end{lemma}

This is essentially Theorem 13.4 from \cite{kallenberg2021foundations}, except we are allowing a state in $(Y_n)_{n\in\Z_{\ge 0}}$ to transition to itself. See Proposition 13.7 from \cite{kallenberg2021foundations} for how to handle that.

We will choose $\beta$ such that for $\eta\in\Omega$ and all bounded measurable $f:\Omega\to\R$, 
\begin{align*}
    \int f(\xi)\beta(\eta,d\xi)&=\sum_{x\in\eta}f(\eta\setminus\{x\})+\int_{\Lambda}f(\eta)+e^{-\nabla_x^+H(\eta)}(f(\eta\cup\{x\})-f(\eta))\bm\lambda(x)dx
\end{align*}
Note that $\beta(\eta,\Omega)=|\eta|+\bm\lambda(\Lambda)$.

The following weak bound will be used later in the burn-in argument.
\begin{lemma}\label{lem:mixing:cardinality_bounded_by_births}
$|X_t|-|X_0|$ is stochastically dominated by $\Pois(\bm\lambda(\Lambda)t)$.
\end{lemma}
\begin{proof}
Define the set of attempted birth times $B$ as in \Cref{lem:mixing:attempted_births_poisson_dominated}.
    We always have $|X_t|\le |X_0|+|B\cap(0,t]|$. By \Cref{lem:mixing:attempted_births_poisson_dominated}, $|B\cap(0,t]|$ is stochastically dominated by $\Pois(\bm\lambda(\Lambda)t)$.
\end{proof}

\begin{definition}
Let $\rho_t(\eta,\cdot)$ be the distribution of $X_t$ if we start with $X_0=\eta$. Define
\begin{align*}
T_tf(\eta)&=\int f(\xi)\rho_t(\eta,d\xi)&\forall x\in \Omega,t\ge 0
\end{align*}
\end{definition}

\begin{theorem}[\cite{kallenberg2021foundations} Thm. 13.9 verbatim (backward equation, Kolmogorov)]
Let $\alpha$ be the rate kernel of a pure jump-type Markov process in $S$, and fix a bounded, measurable function $f:S\to\R$. Then $T_tf(x)$ is continuously differentiable in $t$ for fixed $x$, and
\begin{align*}
\frac{\partial}{\partial t}T_tf(x)&=\int\alpha(x,dy)\{T_tf(y)-T_tf(x)\},&t\ge 0,x\in S
\end{align*}
\end{theorem}

We will iterate this theorem to get a more refined estimate.

\begin{theorem}\label{thm:mixing:jump_process_first_order}
Let $\alpha$ be the rate kernal of a pure jump-type Markov process with $T_t$ as defined above. Let $f:\Omega\to[0,1]$ be measurable. Then,
\begin{align*}
    |T_hf(x)-f(x)|&\le h\alpha(x,\Omega)
\end{align*}
\begin{align*}
    \left|T_hf(x)-f(x)-h\int(f(y)-f(x))\alpha(x,dy)\right|&\le \frac{1}{2}h^2\alpha(x,\Omega)^2+\frac{1}{2}h^2\int\alpha(y,\Omega)\alpha(x,dy)
\end{align*}
\end{theorem}
\begin{proof}
By Theorem 13.9 of \cite{kallenberg2021foundations},
\begin{align*}
    T_tf(x)&=f(x)e^{-t\alpha(x,\Omega)}+\int_0^te^{-s\alpha(x,\Omega)}\int T_{t-s}f(y)\alpha(x,dy)ds \Rightarrow
    \\
    T_tf(x)-f(x)&=f(x)(e^{-t\alpha(x,\Omega)}-1)+\int_0^te^{-s\alpha(x,\Omega)}\int T_{t-s}f(y)\alpha(x,dy)ds\\
    &=\int_0^te^{-s\alpha(x,\Omega)}\int T_{t-s}f(y)-f(x)\alpha(x,dy)ds \Rightarrow
    \\
    |T_tf(x)-f(x)|&\le\int_0^te^{-s\alpha(x,\Omega)}\alpha(x,dy)ds=1-e^{-t\alpha(x,\Omega)}\le t\alpha(x,\Omega)
\end{align*}

Thus, 
\begin{align*}
    T_hf(x)-f(x)-h\int (f(y)-f(x))\alpha(x,dy)=
    \\
    f(x)(e^{-h\alpha(x,\Omega)}-1)+\int_0^h(e^{-t\alpha(x,\Omega)}-1)\int T_{h-t}f(y)\alpha(x,dy)dt
    \\
    +\int_0^h\int T_{h-t}f(y)\alpha(x,dy)ds-h\int f(y)-f(x)\alpha(x,dy)\\
    =f(x)(e^{-h\alpha(x,\Omega)}-(1-h\alpha(x,\Omega)))
    +\int_0^h(e^{-t\alpha(x,\Omega)}-1)\int T_{h-t}f(y)\alpha(x,dy)dt
    \\
    +\int_0^h\int (T_{h-t}f(y)-f(y))\alpha(x,dy)dt
\end{align*}
We now bound each term:
\begin{align*}
    0&\le f(x)(e^{-h\alpha(x,\Omega)}-(1-h\alpha(x,\Omega)))\le \frac{1}{2}(h\alpha(x,\Omega))^2\\
    0&\le \int_0^h(1-e^{-t\alpha(x,\Omega)})\int T_{h-t}f(y)\alpha(x,dy)dt\le \int_0^h t\alpha(x,\Omega)\alpha(x,\Omega)dt=\frac{1}{2}h^2\alpha(x,\Omega)^2
\end{align*}
\begin{align*}
    \left|\int_0^h\int (T_{h-t}f(y)-f(y))\alpha(x,dy)dt\right|&\le \int_0^h(h-t)\int\alpha(y,\Omega)\alpha(x,dy)dt=\frac{1}{2}h^2\int\alpha(y,\Omega)\alpha(x,dy)
\end{align*}

Putting it together,
\begin{align*}
    \left|T_hf(x)-f(x)-h\int(f(y)-f(x))\alpha(x,dy)\right|&\le \frac{1}{2}h^2\alpha(x,\Omega)^2+\frac{1}{2}h^2\int\alpha(y,\Omega)\alpha(x,dy)
\end{align*}

\end{proof}

\begin{lemma}
For $f\in L^2(\mu)$ and $t\ge 0$,
\begin{align*}
    \int (T_tf)^2d\mu\le \int f^2\mu
\end{align*}
\end{lemma}
\begin{proof}
By Jensens's inequality, c.f. \cite{BGL14} section 1.2.1 equation 1.2.1 . 
\end{proof}

\begin{lemma} \label{lem:mixing:strongly_continuous_semigroup}
For any $f\in L^2(\mu)$, we have $\|T_tf-f\|_{\mu}\to 0$, i.e.
\begin{align*}
    \lim_{t\to0+}\int(T_tf-f)^2d\mu=0
\end{align*}

This shows that $(T_t)_{t\ge 0}$ is a strongly continuous semigroup for the Banach space $L^2(\mu)$. (c.f. \cite{Engel2000OneParameter}) 
\end{lemma}
\begin{proof}
We first show that lemma for bounded measurable $f:\Omega\to\R$. By scaling, we may assume that $f:\Omega\to[-1,1]$.

By \Cref{thm:mixing:jump_process_first_order},
\begin{align*}
    |T_tf(\eta)-f(\eta)|&\le 2h(|\eta|+\bm\lambda(\Lambda))
\end{align*}

Thus, since $\mu$ is stochastically dominated by a Poisson point process of intensity $\lambda(x)$,
\begin{align*}
    \int (T_tf-f)^2d\mu&\le 4h^2\int (|\eta|+\bm\lambda(\Lambda))^2d\mu\\
    &\le 4h^2\E_{X\sim\Pois(\bm\lambda(\Lambda))}[(X+\bm\lambda(\Lambda))^2]
\end{align*}

Since $\Pois(\bm\lambda(\Lambda))$ has finite first and second moments, the expectation is finite, and so
\begin{align*}
    \lim_{t\to0+}\int(T_tf-f)^2d\mu=0
\end{align*}

Finally, for any $f\in L^2(\mu)$ and $\epsilon>0$, we can choose bounded measurable $g:\Omega\to\R$ such that $\|f-g\|_\mu\le \epsilon$. Then, for any $t\ge 0$, $\|T_tf-T_tg\|_\mu\le \|f-g\|_{\mu}\le \epsilon$. For sufficiently small $t>0$, we will have $\|T_tg-g\|_{\mu}\le \epsilon$, and so by triangle inequality, $\|T_tf-f\|_{\mu}\le 3\epsilon$.

\end{proof}

We have shown that $(T_t)_{t\ge 0}$ forms a strongly continuous semigroup. We recall some of its basic properties and prove auxiliary lemmas in \Cref{sec: appendix cont MC}. 

\begin{lemma}[\cite{BGL14} Section 1.7.1]
\label{lem:mixing:derivative_of_variance}
For $f\in D(\mcL)$,
\begin{align*}
    \dfrac{d}{dt}\Var_{\mu}(T_tf)&=-2\mathcal{E}(T_tf,T_tf)&t\ge 0
\end{align*}
\end{lemma}
\begin{proof}
We want to show that
\begin{align*}
\frac{d}{dt}\int T_tfd\mu&=\int \mcL T_tfd\mu=0
\end{align*}

since that would directly imply the lemma:
\begin{align*}
    \dfrac{d}{dt}\Var_{\mu}(T_tf)&=\frac{d}{dt}\left(\int(T_tf)^2d\mu-\left(\int T_tfd\mu\right)^2\right)\\
    &=\int 2(T_tf)(\frac{d}{dt}T_tf)d\mu\\
    &=\int 2T_tf\mcL T_tfd\mu\\
    &=-2\mcE(T_tf,T_tf).
\end{align*}

To do that (restricting the limits below to $h\to 0+$ if $t=0$),
\begin{align*}
    \frac{d}{dt}\int T_tfd\mu&=\lim_{h\to 0}\int \frac{T_{t+h}f-T_tf}{h}d\mu=\int \mcL T_tfd\mu=0
\end{align*}
by \Cref{lem:mixing:inner_product_is_continuous_in_l2}
since $\frac{T_{t+h}f-T_tf}{h}\to \mcL T_tf$ in $L^2(\mu)$. Also,
\begin{align*}
    \dfrac{d}{dt}\int (T_tf)^2d\mu&=\lim_{h\to0}\int \frac{(T_{t+h}f)^2-T_tf^2}{h}d\mu\\
    &=\lim_{h\to0}\int (T_{t+h}f+T_tf)\frac{T_{t+h}f-T_tf}{h}d\mu\\
    &=\int 2T_tf\mcL T_tfd\mu=-2\mathcal{E}(T_tf,T_tf)
\end{align*}
where we used \Cref{lem:mixing:inner_product_is_continuous_in_l2} with $T_{t+h}f\to T_tf$ and $\frac{T_{t+h}f-T_tf}{h}\to\mcL f$ in $L^2(\mu)$ as $h\to 0$.
Hence,
\begin{align*}
\dfrac{d}{dt}\Var_{\mu}(T_tf)&=
\dfrac{d}{dt}\left(\int (T_tf)^2d\mu-\left(\int T_tfd\mu\right)^2\right)\\
&=-2\mathcal{E}(T_tf,T_tf)
\end{align*}
\end{proof}

\subsection{Mixing of $(X_t)_{t\ge 0}$}
Before we conclude with the main mixing time result for Continuum Glauber dynamics, we will prove the following theorem:
\begin{theorem}\label{thm:mixing:continuous_time_mixing}
Suppose that for some repulsive Gibbs point process $\mu$ with Dirichlet form $\mcE$, we have that
\begin{align*}
\mcE(f,f)&\ge \gamma\Var_{\mu}(f)&\forall\text{ bounded measurable }f:\Omega\to\R
\end{align*}

Let $(X_t)_{t\ge 0}$ be as constructed above and let $\mu_t$ be the law of $X_t$. Suppose we choose $\mu_0$ so that $\dfrac{d\mu_0}{d\mu}\le B$. Then,
\begin{align*}
    d_{TV}(\mu_t,\mu)&\le \dfrac{1}{2}e^{(\log B)/2-\gamma t}
\end{align*}

\end{theorem}

We will divide the proof into several lemmas. 
\begin{lemma}\label{lem:mixing:evolution_of_relative_density}
Let $f_0:=\frac{d\mu_0}{d\mu}$ (chosen so that $0\le f_0\le B$), and $f_t:=T_tf_0$ for $t>0$.
Then, for $t\ge 0$, we have $0\le f_t\le B$ and
\begin{align*}
    f_t&=\frac{d\mu_t}{d\mu}
\end{align*}
In other words, for bounded measurable $g:\Omega\to\R$,
\begin{align*}
    \int g(x)f_t(x)d\mu(x)&=\int g(x) d\mu_t(x)
\end{align*}
\end{lemma}
\begin{proof}

Since $f_0(y)\in [0,B]$ for all $y\in\Omega$,
\begin{align*}
    f_t(x)&=T_tf_0(x)=\E_{y\sim\rho_t(x,\cdot)}[f_0(y)]\in[0,B]&\forall x\in\Omega
\end{align*}

Let $g:\Omega\to\R$ be a bounded measurable function. Then,
\begin{align*}
    \E[g(X_t)]&=\E_{x\sim \mu_t}[g(x)]=\int g(x)d\mu_t(x)
\end{align*}
\begin{align*}
    \E[g(X_t)]&=\int \int f(y)\rho_t(x,dy) d\mu_0(x)\\
    &=\int T_tg(x)d\mu_0(x)\\
    &=\int f_0(x)T_tg(x)d\mu(x)\\
    &=\int g(x)T_tf_0(x)d\mu(x)&\text{by \Cref{lem:mixing:reversible}}\\
    &=\int g(x)f_t(x)d\mu(x)
\end{align*}
\end{proof}

\begin{lemma}\label{lem:mixing:variance_decay}
\begin{align*}
    \Var_{\mu}(f_t)\le e^{-2\gamma t}\Var_{\mu}(f_0)
\end{align*}
\end{lemma}
\begin{proof}
For all $t\ge 0$, $f_t$ is bounded, and so by the assumption of the theorem,
\begin{align*}
    \mcE(f_t,f_t)\ge \gamma \Var_{\mu}(f_t)  
\end{align*}

By \Cref{lem:mixing:derivative_of_variance}, we have for all $t\ge 0$,
\begin{align*}
    \dfrac{d}{dt}\Var_{\mu}(f_t)&=-2\mcE(f_t,f_t)\le -2\gamma\Var_{\mu}(f_t)
\end{align*}

Hence, integrating yields 
\begin{align*}
    \Var(f_t)\le e^{-2\gamma t}\Var(f_0)
\end{align*}
\end{proof}

\begin{lemma}\label{lem:mixing:mixing_from_renyi_infty}
    \begin{align*}
        d_{TV}(\mu_t,\mu)&\le \dfrac{1}{2}e^{(\log B)/2-\gamma t}
    \end{align*}
\end{lemma}

\begin{proof}
We will show that
\begin{align*}
    4d_{TV}(\mu_t,\mu)^2\le \Var_{\mu}(f_t)\le \Var_{\mu}(f_0) e^{-2\gamma t}\le e^{\log B-2\gamma t}
\end{align*}

The first inequality follows from the well-known fact that for $\nu\ll\mu$,
\begin{align*}
4d_{TV}(\nu,\mu)^2&=\left(\int|\dfrac{d\nu}{d\mu}-1|d\mu\right)^2\le \int\left(\dfrac{d\nu}{d\mu}-1\right)^2d\mu=\Var_{\mu}(\dfrac{d\nu}{d\mu})
\end{align*}
together with the fact that $\dfrac{d\mu_t}{d\mu}=f_t$ as shown in \Cref{lem:mixing:evolution_of_relative_density}.

The second inequality follows from \Cref{lem:mixing:variance_decay}.

The third inequality follows from the definition of variance and \Cref{lem:mixing:evolution_of_relative_density}:
\begin{align*}
    \Var_{\mu}(f_0)
    \le \int f_0^2d\mu=B \int f_0d\mu_t\le B
\end{align*}
\end{proof}

Now we can prove the main theorem of this section.
\begin{proof}[Proof of \Cref{thm:mixing:sampling_algorithm_guarantees}]\label{proof: main mix theorem}

It is easy to see that \Cref{alg:birth_death} is equivalent to \Cref{alg:continuum_glauber_chain}.

The number of iterations of \Cref{alg:continuum_glauber_chain} is bounded by \Cref{lem:mixing:iteration_bound}, it is (up to constants) stochastically dominated by $\Pois(\bm\lambda(\Lambda)T)$.
Constructing $(X_t)_{t\ge 0}$ as in \Cref{lem:implementation:chain_for_process}, the output of \Cref{alg:continuum_glauber_chain} will be
\begin{align*}
    Y_{k^*}=X_T\sim\mu_T \,.
\end{align*}
By \Cref{thm:mixing:continuous_time_mixing},
we have $d_{TV}(\mu_T,\mu)\le \dfrac{1}{2}e^{(\log B)/2-\gamma t}$.
\end{proof}

\begin{theorem}\label{thm:burnin:main_empty_start}
The mixing time of Continuum Glauber for Gibbs point processes with repulsive pair potentials with activity $\bm\lambda$ on region $\Lambda$ with initial state $\emptyset$ is
    $O(\frac{1}{\gamma}\bm\lambda(\Lambda)+(1+\frac{1}{\gamma})\log(1/\epsilon))$, where $\gamma$ is the spectral gap of continuum Glauber.
\end{theorem}

\begin{proof}
Let $\mu_0=\mathbf{1}_{\emptyset}$, so $\dfrac{d\mu_0}{d\pi}\le e^{\bm\lambda(\Lambda)}$.
For $t\ge \dfrac{1}{\gamma}(\bm\lambda(\Lambda)/2+\log(1/\epsilon))$, we have by \Cref{thm:mixing:continuous_time_mixing} that
\begin{align*}
    d_{TV}(\mu_T,\mu)\le \dfrac{1}{2}e^{\bm\lambda(\Lambda)/2-\gamma t}\le \epsilon
\end{align*}
\end{proof}

\subsection{Runtime}

In this subsection, we show bounds on the time complexity of \Cref{alg:birth_death}.

\begin{corollary}\label{cor:mixing:sampling_algorithm_runtime}
Let $\epsilon\in(0,1)$. Suppose $\nabla_{x}^+H(\eta)$ can be computed in $O(|\eta|)$ time, and we have a constant-time oracle for sampling from $\Lambda$ with density proportional to $\bm\lambda$. If we run \Cref{alg:birth_death} initialized to the empty set with $T=\frac{1}{\gamma}(\bm\lambda(\Lambda)/2+\log(1/\epsilon))$, it will output from a distribution within $\epsilon$ in TV distance of $\mu$, take $O(\frac{1}{\gamma}\bm\lambda(\Lambda)(\bm\lambda(\Lambda)+\log(1/\epsilon)))$ iterations in expectation, and have a total time complexity of $O(1+\frac{1}{\gamma^2}\bm\lambda(\Lambda)^2(\bm\lambda(\Lambda)+\log(1/\epsilon))^2)$ in expectation.
\end{corollary}
\begin{proof}
    By \Cref{thm:mixing:sampling_algorithm_guarantees}, the expected number of iterations of \Cref{alg:birth_death} will be at most (a constant times)
    \begin{align*}
        \E_{X\sim \Pois(\bm\lambda(\Lambda)T)}[X]&=O(\bm\lambda(\Lambda)T)
    \end{align*}
    Since the $i$th iteration of the algorithm takes $O(i)$ (if we naively bound $|\eta|$ at the start of iteration $i$ by $i$), the expected time complexity will be (up to constants)
    \begin{align*}
    \E_{X\sim \Pois(\bm\lambda(\Lambda)T)}[X^2]=O( \bm\lambda(\Lambda)T+(\bm\lambda(\Lambda)T)^2).
    \end{align*}

    Also, when $T\ge \frac{1}{\gamma}(\bm\lambda(\Lambda)/2+\log(1/\epsilon))$,
\begin{align*}
d_{TV}(\mu_T,\mu)\le \dfrac{1}{2}e^{\bm\lambda(\Lambda)/2-\gamma T}\le \epsilon
\end{align*}
\end{proof}

For the special case of the hard sphere model, we can use spatial data structures to optimize the runtime.

\begin{theorem}[Runtime of CG for hard spheres]\label{thm:runtime:hard_spheres}
    Let $\gamma$ be the spectral gap of continuum Glauber for the hard sphere model on $\Lambda\subseteq\R^d$ with activity $\bm\lambda:\Lambda\to[0,\infty)$. Then, \Cref{alg:birth_death}
    initialized to the empty set with $T=\frac{1}{\gamma}(\bm\lambda(\Lambda)/2+\log(1/\epsilon))$ produces a sample within $\eps$ TV distance of the hard sphere model with an expected runtime of $O(\frac{1}{\gamma}\bm\lambda(\Lambda)(\bm\lambda(\Lambda)+\log(1/\epsilon))$, assuming a constant time oracle from sampling from $\Lambda$ with density proportional to $\bm\lambda$.
\end{theorem}

\begin{proof}

\Cref{thm:mixing:sampling_algorithm_guarantees} shows that the expected number of iterations is at most
\begin{align*}
    \E_{X\sim \Pois(\bm\lambda(\Lambda)T)}[X]&=O(\bm\lambda(\Lambda)T).
\end{align*}

The cost of each iteration is dominated by the computation of $\nabla_x^+H(\eta)$.
    
To achieve constant time complexity per iteration, we employ a Spatial Hashing technique, which relies on partitioning the domain $\Lambda$ into a uniform grid structure.
We construct a grid where each cubic cell has a side length $2r$, where $r$ is the radius of the balls.

Let $C(x)$ be the cell containing the proposed point $x\in\Lambda$. We observe the following geometric property:

\begin{fact}
If an existing sphere center $x'$ is within the exclusion distance $2r$ of the proposed point $x$, i.e., $\|x - x'\| < 2r$, then $x'$ must reside in a cell that is adjacent to $C(x)$.
\end{fact}
Thus, to decide if we can add $x$, it is enough to check points in the cells $\mathcal{N}(x)$, consisting of the cell $C(x)$ and its neighbors (including diagonals). The number of such cells for a given dimension is
\[ |\mathcal{N}(x)| = 3^d \]

The maximum number of sphere centers that can fit inside the cell of size $(2r)^d$ is bounded by $k_{\max,d}=\frac{(4r)^d}{(2r/\sqrt{d})^d}=2^d{d^{d/2}}$, and hence the total cost for the non-overlap check for spheres is at most
\[ |\mathcal{N}(x)| \cdot k_{\max,d} = 3^d \cdot 2^d{d^{d/2}} = O_d(1)  \, .\]

\end{proof}

\subsection{Burn-in}\label{subsec: burn-in}

In this section, we use a burn-in argument to show mixing from any starting configuration. We note that for the purposes of sampling, one could always choose to start from the empty set, but we believe this result is interesting for its own sake.

\begin{theorem}[Mixing with Burn-in]\label{thm:burnin:main}
Let $(\mu_t)_{t\ge 0}$ be the distribution of continuum Glauber on region $\Lambda$ with initial state $S$. Then, for $t\ge \frac{1}{\gamma}(\bm\lambda(\Lambda)/2+\log(1/\epsilon))+\log(2
|S|/\epsilon)$,
\begin{align*}
    d_{TV}(\mu_t,\mu)&\le \epsilon
\end{align*}
\end{theorem}

We will prove the theorem using the following lemma, whose proof is deferred to the end of the subsection.
\begin{lemma}\label{lem:burnin:burn_in}
Let $s>0$. Suppose $\mu_0=\bm1_S$. Then, there exists a probability measure $\nu$ such that
\begin{align*}
    d_{TV}(\mu_s,\nu)&\le |S|e^{-s}&\text{and}&&\dfrac{d\nu}{d\mu}\le e^{\bm\lambda(\Lambda)}
\end{align*}
\end{lemma}

\begin{proof}[Proof of \Cref{thm:burnin:main}]
We choose $s\ge \log(2|S|/\epsilon)$.

By \Cref{lem:burnin:burn_in}, there is some $\nu$ with $\dfrac{d\nu}{d\mu}\le e^{\bm\lambda(\Lambda)}$ such that
\begin{align*}
    d_{TV}(\mu_0T_s,\nu)\le |S|e^{-s}\le \dfrac{\epsilon}{2}
\end{align*}

For $t\ge \dfrac{1}{\gamma}(\bm\lambda(\Lambda)/2+\log(1/\epsilon))$, we have by \Cref{thm:mixing:continuous_time_mixing} that
\begin{align*}
    d_{TV}(\nu T_t,\mu)\le \dfrac{1}{2}e^{\bm\lambda(\Lambda)/2-\gamma t}\le \dfrac{\epsilon}{2}
\end{align*}

\begin{align*}
    d_{TV}(\mu_0T_{s+t},\mu)&\le d_{TV}(\mu_0T_{s+t},\nu T_t)+d_{TV}(\nu T_t,\mu)\\
    &\le d_{TV}(\mu_0T_s,\nu)+d_{TV}(\nu T_t,\mu)&\text{Data processing inequality}\\
    &\le \dfrac{\epsilon}{2}+\dfrac{\epsilon}{2}=\epsilon
\end{align*}

Hence, we have $d_{TV}(\mu_0T_t,\mu)\le \epsilon$ when $t\ge \frac{1}{\gamma}(\bm\lambda(\Lambda)/2+\log(1/\epsilon))+\log(2
|S|/\epsilon)$.
\end{proof}

We keep the main proofs needed to prove \Cref{lem:burnin:burn_in} in this section and put the auxiliary lemmas in \Cref{sec: burn-in appendix}. Additionally, we provide definitions and lemmas about jump-type Markov process in \Cref{sec: jump type} for reader's convenience.

Define $U_t$ to be the finite measure such that for measurable $f:\Omega\to[0,1]$,
\begin{align*}
    \int fdU_t&=\sum_{T\subseteq S}(e^{-t})^{|T|}(1-e^{-t})^{|S|-|T|}\int f(\eta\cup T)e^{H(T)-H(\eta\cup T)}e^{\bm\lambda(\Lambda)}d\rho(\eta)
\end{align*}
where $\rho$ is the Poisson point process of intensity $\bm\lambda$.

\begin{theorem}
    For $s>0$ and measurable $f:\Omega\to[0,1]$,
    \begin{align*}
        \int fd\mu_s\le \int fdU_s
    \end{align*}
\end{theorem}

This theorem shows an upper bound on $\mu_t$. 
We will show that if we condition $\mu_t$ on the event that $T$ are the points of $S$ that survive, the relative density of this distribution with respect to $\mu$ with $T$ pinned is bounded. In the end, we will only need that the relative density of $\mu_t$ conditioned on all points of $S$ dying is bounded, but the stronger property is needed for the ``induction''.

The intuition is the following.
First, imagine that points in $S$ are not allowed to die. Then, the process behaves like continuum Glauber on the model with the points of $S$ pinned. The relative density of the initial distribution with respect to this modified stationary distribution starts off bounded by $e^{\bm\lambda(\Lambda)}$ and cannot increase. In other words, the distribution is upper bounded by a constant multiple of the modified stationary distribution. Even if the initial distribution is not $\bm1_S$, we can partition the distribution into $2^{|S|}$ parts based on which points $T\subseteq S$ are present. Each part will evolve independently, behaving like continuum Glauber on the model with the points $T$ pinned. Like before, a linear combination of modified stationary distributions (one for each possible $T\subseteq S$ pinned) will remain an upper bound.

On the other hand, suppose we ignore everything except the deaths of points in $S$. In this case, each point of $S$ lives for $\Exp(1)$ time independently, so it is not too difficult to exactly track how the distribution evolves. In particular, if a linear combination of modified stationary distributions as described in the previous paragraph is an upper bound on the distribution, this will remain true over time with the coefficients of the linear combination evolving continuously. (Pinning points changes the distribution more than simply adding the points in, but because $H$ is supermodular, the inequality is in the right direction).

The true process can be thought of as these two processes happening simultaneously. We can devise an upper bound on the measure that is preserved by the first process (continuum Glauber with points in $S$ pinned) and evolves in a well-behaved way on the second process (points in $S$ die). We formalize this intuition below.

\begin{proof}
The main idea is to show that
\begin{align*}
    \dfrac{d}{dt}\int T_{s-t}fd\mu_t&\le \dfrac{d}{dt}\int T_{s-t}fdU_t&\text{for }t\in[0,s]
\end{align*}
Together with $\int T_sfd\mu_0\le \int T_sfdU_0$, this will show that $\int fd\mu_s\le \int fdU_s$.

Since $T_tf$ does not have an easy-to-use form, we will use approximations of it instead.
To ensure the approximations do not blow up at any point, we will truncate the functions so they are zero when the number of points of their input is large. 

Define
\begin{align*}
     \mcS_h:=\{f:\Omega\to[0,1]\text{ measurable}|f(\eta)=0\text{ when }|\eta|\ge \frac{1}{h}-\bm\lambda(\Lambda)-1\}
\end{align*}

Fix $s>0$. Let $f_0\in\bigcup_{h>0}\mcS_h$. Let $k$ be a large enough integer so that $f_0\in \mcS_{s/k}$, and define
\begin{align*}
    f_{j}&:=\bm1_{|\cdot|<\frac{k}{s}-\bm\lambda(\Lambda)-1}(f_{j-1}+\frac{s}{k}\mcL f_{j-1})&\text{for $1\le j\le k$}
\end{align*}
We first show that $f_j\in\mcS_{s/k}$ for all $j\ge 0$.

\begin{claim}
    Suppose $f\in \mcS_h$ for some $h>0$. Then,
    \begin{align*}
        (f+h\mcL f)\bm1_{|\cdot|\le \frac{1}{h}-\bm\lambda(\Lambda)-1}\in \mcS_h
    \end{align*}
    Also, for $t\in[0,s]$,
    \begin{align*}
        \int \left|(f+h\mcL f)\bm1_{|\cdot|\le \frac{1}{h}-\bm\lambda(\Lambda)-1}-(f+h\mcL f)\right|d\mu_t=o(h)
    \end{align*}
    \begin{align*}
        \int \left|(f+h\mcL f)\bm1_{|\cdot|\le \frac{1}{h}-\bm\lambda(\Lambda)-1}-(f+h\mcL f)\right|dU_t=o(h)
    \end{align*}
\end{claim}
\begin{proof}

Suppose $|\eta|<\frac{1}{h}-\bm\lambda(\Lambda)$. Then,
\begin{align*}
    f(\eta)+h\mcL f(\eta)&=f(\eta)+h\left(\sum_{x\in\eta}(f(\eta\setminus\{x\})-f(\eta))+\int_{\Lambda}e^{H(\eta\cup\{x\})-H(\eta)}f(\eta\cup\{x\})-f(\eta)\bm\lambda(x)dx\right)\\
    &=(1-h(|\eta|+\int e^{H(\eta\cup\{x\})-H(\eta)}\bm\lambda(x)dx))f(\eta)+h\sum_{x\in\eta}f(\eta\setminus\{x\})\\
    &~~~+h\int_{\Lambda}e^{H(\eta\cup\{x\})-H(\eta)}f(\eta\cup\{x\})\bm\lambda(x)dx\\
    &\le (1-h(|\eta|+\int e^{H(\eta\cup\{x\})-H(\eta)}\bm\lambda(x)dx))+h|\eta|+h\int_{\Lambda}e^{H(\eta\cup\{x\})-H(\eta)}\bm\lambda(x)dx\le 1
\end{align*}
Nonnegativity follows from the second line since
\begin{align*}
    h(|\eta|+\int e^{H(\eta\cup\{x\})-H(\eta)}\bm\lambda(x)dx)&\le h(|\eta|+\bm\lambda(\Lambda))<1
\end{align*}

When $|\eta|\ge \frac{1}{h}-\bm\lambda(\Lambda)$, we have $f(\eta)+h\mcL f(\eta)=0$. Hence, $0\le f+h\mcL f\le 1$.

Thus, in the limit as $h\to\infty$,
\begin{align*}
    \int \left|(f+h\mcL f)\bm1_{|\cdot|>\frac{1}{h}-\bm\lambda(\Lambda)-1}\right|d\mu_t&\le \Pr_{\eta\sim\mu_t}(|\eta|\ge \frac{1}{h}-\bm\lambda(\Lambda)-1)\\
    &\le \Pr_{X\sim\Pois(\bm\lambda(\Lambda)s)}(|S|+X\ge \frac{1}{h}-\bm\lambda(\Lambda)-1)
\end{align*}
by the standard Poisson tail bound $\Pr_{X\sim\Pois(\lambda)}(X\ge t)\le e^{-\lambda}(e\lambda/t)^t$ for $t>\lambda$ (c.f. \cite{vershynin2018high}).

By similar calculations,
\begin{align*}
    \int \left|(f+h\mcL f)\bm1_{|\cdot|>\frac{1}{h}-\bm\lambda(\Lambda)-1}\right|dU_t&\le 
    \int \bm1_{|\cdot|>\frac{1}{h}-\bm\lambda(\Lambda)-1}dU_t\\
    &\le e^{\bm\lambda(\Lambda)}\Pr_{\eta\sim\rho}(|\eta|+|S|\ge \frac{1}{h}-\bm\lambda(\Lambda)-1)=o(h)
\end{align*}
\end{proof}

Now, observe that
\begin{align*}
    \int f_0d\mu_s-\int f_0dU_s&=\int f_kd\mu_0-\int f_kdU_0+\\
    &~~~\sum_{j=1}^k\Bigg(\left(\int f_{k-j}d\mu_{tj/k}-\int f_{k-j+1}d\mu_{t(j-1)/k}\right)\\
    &~~~~~~~~-\left(\int f_{k-j}dU_{tj/k}-\int f_{k-j+1}dU_{t(j-1)/k}\right)\Bigg)
\end{align*}

Using the fact that $\Pr_{\eta\sim\rho}(\eta=\emptyset)=e^{-\bm\lambda(\Lambda)}$, we have
\begin{align*}
    \int f_kd\mu_0&=f_k(S)\le \int f_k(\eta\cup S)e^{H(S)-H(\eta\cup S)}e^{\bm\lambda(\Lambda)}d\rho(\eta)=\int f_kdU_0
\end{align*}

Letting $f=f_{k-j}$, $h=s/k$, and $t=s(j-1)/k$,
\begin{align*}
    \int f_{k-j}d\mu_{tj/k}-\int f_{k-j+1}d\mu_{t(j-1)/k}&=\int fd\mu_{t+h}-\int (f+h\mcL f)\bm1_{|\cdot|\le \frac{1}{h}-\bm\lambda(\Lambda)-1}d\mu_t\\
    &=\int fd\mu_{t+h}-\int f+h\mcL fd\mu_t+o(h)\\
    &=o(h)&\text{by \Cref{lem:burnin:linear_approximation_of_mu_t}}
\end{align*}
where the bound on $o(h)$ is uniform over $t\in[0,s]$.

Similarly,
\begin{align*}
    \int f_{k-j}dU_{tj/k}-\int f_{k-j+1}dU_{t(j-1)/k}&=\int fdU_{t+h}-\int (f+h\mcL f)\bm1_{|\cdot|\le \frac{1}{h}-\bm\lambda(\Lambda)-1}dU_t\\
    &=\int fdU_{t+h}-\int f+h\mcL fdU_t+o(h)\\
    &\le o(h)&\text{by \Cref{lem:burnin:upperbound_evolution_inequality}}
\end{align*}
where the bound on $o(h)$ is uniform over $t\in[0,s]$.

Thus,
\begin{align*}
    \int f_0d\mu_s-\int f_0dU_s&=\int f_kd\mu_0-\int f_kdU_0+\\
    &~~~\sum_{j=1}^k\Bigg(\left(\int f_{k-j}d\mu_{sj/k}-\int f_{k-j+1}d\mu_{s(j-1)/k}\right)\\
    &~~~~~~~~-\left(\int f_{k-j}dU_{sj/k}-\int f_{k-j+1}dU_{s(j-1)/k}\right)\Bigg)\\
    &\le \sum_{j=1}^ko(s/k)=o(1)
\end{align*}
Taking $k\to\infty$, we get
\begin{align*}
    \int fd\mu_s\le \int fdU_s
\end{align*}
For any measurable $f:\Omega\to[0,1]$, we can approximate it from below by $f_n\nearrow f$ where each $f_n\in\bigcup_{h>0}\mcS_h$, and the same inequality holds for $f$ by the monotone convergence theorem.

\end{proof}

\begin{lemma}\label{lem:burnin:survivors_of_initial}
Let $\mcA_T=\{\eta\in\Omega:|\eta\cap S|=T\}$. Then, for any $s\ge 0$ and $T\subseteq S$,
\begin{align*}
    \mu_s(\mcA_T)&=(e^{-s})^{|T|}(1-e^{-s})^{|S|-|T|}
\end{align*}
\end{lemma}

\begin{proof}
    Let $\nu_t$ be the measure defined by
\begin{align*}
    \int f d\nu_t&=\sum_{T\subseteq S}(e^{-t})^{|T|}(1-e^{-t})^{|S|-|T|}f(T)
\end{align*}

The rough idea is to set $f=1_{\mcA_T}$ and show that
\begin{align*}
    \dfrac{d}{dt}\int T_{s-t}fd\mu_t&=\dfrac{d}{dt}\int T_{s-t}fd\nu_t&\text{for }t\in[0,s]
\end{align*}
Together with $\int T_sfd\mu_0=T_sf(S)=\int T_sfd\nu_0$, this will show that $\int fd\mu_s=\int fd\nu_s$.

Since $T_tf$ does not have nice formula, we will use approximation of it instead.

Fix $s>0$ and $T\subseteq S$. Let $k>s|S|$ be an integer. 
Define
\begin{align*}
    f_0&:=1_T&f_{j}&:=f_{j-1}+\frac{s}{k}\mcL f_{j-1}&\text{for $1\le j\le k$}
\end{align*}

Let $\mcC=\{f:\Omega\to[0,1]:\forall \eta\in\Omega,~f(\eta)=f(\eta\cap S)\}$. We first show that $f_j\in\mcC$ for all $j\ge 0$.

\begin{claim}
    If $f\in\mcC$ and $0<h<\frac{1}{|S|}$, then $f+h\mcL f\in\mcC$.
\end{claim}
\begin{proof}
Suppose $f\in\mcC$. Then, for all $\eta\in\Omega$, $f(\eta)=f(\eta\cap S)$, and for almost all $x\in\Lambda$, we have $f(\eta\cup\{x\})=f(\eta)$. Thus,
\begin{align*}
    \mcL f(\eta)&=\sum_{x\in\eta}(f(\eta\setminus\{x\})-f(\eta))+\int_{\Lambda}(f(\eta\cup\{x\})-f(\eta))\bm\lambda(x)dx\\
    &=\sum_{x\in\eta\cap S}(f(\eta\setminus\{x\})-f(\eta))\\
    &=\sum_{x\in\eta\cap S}(f(\eta\cap S\setminus\{x\})-f(\eta\cap S))\\
    &=\mcL f(\eta\cap S)
\end{align*}
For $\eta\subseteq S$,
\begin{align*}
    f(\eta)+h\mcL f(\eta)&=f(\eta)+h\sum_{x\in\eta}(f(\eta\setminus\{x\})-f(\eta))\\
    &=(1-h|S|)f(\eta)+h\sum_{x\in\eta}f(\eta\setminus\{x\})\in[0,1]
\end{align*}
since it is a convex combination of $f(\eta)$ and $f(\eta\setminus\{x\})$ for $x\in\eta$, which lie in $[0,1]$.

For general $\eta\in\Omega$,
\begin{align*}
    f(\eta)+h\mcL f(\eta)&=f(\eta\cap S)+h\mcL f(\eta\cap S)\in[0,1]
\end{align*}
Hence, $f+h\mcL f\in\mcC$.
\end{proof}

Now, observe that
\begin{align*}
    \int f_0d\mu_s-\int f_0d\nu_s&=\int f_kd\mu_0-\int f_kd\nu_0+\\
    &~~~\sum_{j=1}^k\Bigg(\left(\int f_{k-j}d\mu_{tj/k}-\int f_{k-j+1}d\mu_{t(j-1)/k}\right)\\
    &~~~~~~~~-\left(\int f_{k-j}d\nu_{tj/k}-\int f_{k-j+1}d\nu_{t(j-1)/k}\right)\Bigg)
\end{align*}

Letting $f=f_{k-j}$, $h=s/k$, and $t=s(j-1)/k$,
\begin{align*}
    \int f_{k-j}d\mu_{tj/k}-\int f_{k-j+1}d\mu_{t(j-1)/k}&=\int fd\mu_{t+h}-\int f+h\mcL fd\mu_t=o(h)
\end{align*}
where the bound on $o(h)$ is uniform over $t\in[0,s]$.

Also,
\begin{align*}
    \int f+h\mcL fd\nu_t&=\sum_{T\subseteq S}(e^{-t})^{|T|}(1-e^{-t})^{|S|-|T|}\left(f(T)+h\sum_{x\in T}(f(T\setminus \{x\})-f(T))\right)\\
    &=\sum_{T\subseteq S}(e^{-t})^{|T|}(1-e^{-t})^{|S|-|T|}\left((1-h|T|)f(T)+h\sum_{x\in T}f(T\setminus \{x\}))\right)\\
    &=\sum_{T\subseteq S}(e^{-t})^{|T|}(1-e^{-t})^{|S|-|T|}(1-h|T|+\frac{e^{-t}}{1-e^{-t}}h(|S|-|T|))f(T)\\
    &=\sum_{T\subseteq S}(e^{-t-h})^{|T|}(1-e^{-t-h})^{|S|-|T|}f(T)+o(h)\\
    &=\int fd\nu_{t+h}+o(h)
\end{align*}

Thus,
\begin{align*}
    \int f_0d\mu_s-\int f_0d\nu_s&=\int f_kd\mu_0-\int f_kd\nu_0+\\
    &~~~\sum_{j=1}^k\Bigg(\left(\int f_{k-j}d\mu_{sj/k}-\int f_{k-j+1}d\mu_{s(j-1)/k}\right)\\
    &~~~~~~~~-\left(\int f_{k-j}d\nu_{sj/k}-\int f_{k-j+1}d\nu_{s(j-1)/k}\right)\Bigg)\\
    &=f_0(S)-f_0(S)+\sum_{j=1}^ko(s/k)=o(1)
\end{align*}
Taking $k\to\infty$ and using $f_0=1_{\mcA_T}$, we get
\begin{align*}
    \mu_s(\mcA_T)&=(1-e^{-s})^{|S|-|T|}(e^{-s})^{|T|}
\end{align*}
\end{proof}

Now we can prove the burn-in lemma. 
\begin{proof}[Proof of \Cref{lem:burnin:burn_in}]
Let $\mcA:=\{\eta\in\Omega:\eta\cap S=\emptyset\}$.

By \Cref{lem:burnin:survivors_of_initial}, we have
\begin{align*}
    \mu_s(\mcA)=(1-e^{-s})^{|S|}
\end{align*}
For measurable $f:\Omega\to[0,1]$ such that $f(\eta)=0$ for $\eta\notin \mcA$,
\begin{align*}
    \int fd\mu_s&\le\int_{\mcA} fdU_s\\
    &=(1-e^{-s})^{|S|}\int f(\eta)e^{H(\emptyset)-H(\eta)}e^{\bm\lambda(\Lambda)}d\rho(\eta)\\
    &=(1-e^{-s})^{|S|}\int e^{H(\emptyset)-H(\eta)}e^{\bm\lambda(\Lambda)}d\rho(\eta)\int fd\mu&\text{by \Cref{eqn:gpp_ppp_density}}\\
    &\le (1-e^{-s})^{|S|}e^{\bm\lambda(\Lambda)}\int fd\mu
\end{align*}

We now choose $\nu$ to be proportional to $\mu_s|_{\mcA}$, i.e.
\begin{align*}
    d\nu&=\frac{\bm1_{\mcA}}{\mu_s(\mcA)}d\mu_s
\end{align*}

We can check that
\begin{align*}
    \int fd\nu&=\frac{1}{\mu_s(\mcA)}\int_{\mcA}f d\mu_s\\
    &\le \frac{1}{(1-e^{-s})^{|S|}}(1-e^{-s})^{|S|}e^{\bm\lambda(\Lambda)}\int fd\mu\\
    &\le e^{\bm\lambda(\Lambda)}\int fd\mu
\end{align*}
and
\begin{align*}
    d_{TV}(\mu_s,\nu)&=\frac{1}{2}\int_{\Omega}\left|1-\frac{d\nu}{d\mu_s}\right|d\mu_s\\
    &=\frac{1}{2}\left(\int_{\mcA}\left|1-\frac{1}{\mu_s(\mcA)}\right|d\mu_s+\int_{\mcA^c}1d\mu_s\right)\\
    &=\frac{1}{2}(\mu_s(\mcA)(\frac{1}{\mu_s(\mcA)}-1)+(1-\mu_s(\mcA))\\
    &=1-\mu_s(\mcA)=1-(1-e^{-s})^{|S|}\le |S|e^{-s}
\end{align*}

\end{proof}

\section{Sampling Sphere Packings of Fixed Size}\label{sec:fixed_size_sampling}

In this section, we show that if continuum Glauber has a spectral gap for $\mu_{\bm\lambda}$ (and likewise for smaller $\bm\lambda$), then \Cref{alg:canonical_sampler} can be used to efficiently (approximately) sample from the canonical ensemble for a size $k$ less than the expected number of points of a sampler from $\mu_{\bm\lambda}$ (\Cref{cor:canonical:main_with_runtime}). In particular this will imply Theorem~\ref{thm:canonical_hard_sphere_density}, improving the density up to which sampling is provably efficient for hard spheres. Our techniques and result will apply to the more general setting of repulsive pair potentials. Since we do not expect this approach to obtain the optimal polynomial, we keep things simple and do not optimize the polynomial in the time complexity. Our approach is based on \cite{davies2023approximately}, which proved the analogous theorem in the discrete setting.

Suppose that we want to sample sphere packings with size $|\eta|=k$. The idea is to first show a lower bound on $\Pr(|\eta|\le k)$ and $\Pr(|\eta|\ge k)$, and then use the spectral gap of Continuum Glauber to show a lower bound for $\Pr(|\eta|=k)$. The sampling algorithm performs a search over $t\in[0,1]$ to find $t$ where $\Pr_{\eta\sim\mu_{t\bm\lambda}}(|\eta|=k)$ is large, and performs rejection sampling from it. Note that $\mu_{t\bm\lambda}$ conditioned on $|\eta|=k$ is a uniform distribution over sphere packings, because the Gibbs weight depends only on $|\eta|$ after conditioning, and thus it is the same distribution for all $t$. The next lemma holds for all repulsive Gibbs point processes, not just hard spheres.

\begin{lemma}\label{lemma: ratio of k-1 and k}
Let $\mu_{\bm\lambda}$ be a repulsive Gibbs point process. For integer $k\ge 1$,
    \begin{align*}
        \frac{\Pr_{\eta\sim\mu_{\bm\lambda}}(|\eta|=k-1)}{\Pr_{\eta\sim\mu_{\bm\lambda}}(|\eta|=k)}&\ge \frac{k}{\bm\lambda(\Lambda)}
    \end{align*}
\end{lemma}
    To see where this ratio comes from intuitively, observe that in continuum Glauber, the jump rate from configurations of $k$ points to configurations of $k-1$ points is $k$, while the jump rate from configurations of $k-1$ points to $k$ points is at most $\bm\lambda(\Lambda)$. This intuition can be formalized using the GNZ equations.
\begin{proof}

Recall that the GNZ equations (see \Cref{thm: GNZ equations})
say that for any measurable 
$F:\Omega\times\Lambda\to[0,\infty)$
\begin{align*}
\int_{\Omega}\sum_{x\in\eta}F(\eta,x)d\mu_{\bm\lambda}(\eta)&=\int \int e^{-\nabla_x^+H(\eta)}F(\eta\cup\{x\},x)\bm\lambda(x)\,dx \, d\mu_{\bm\lambda}(\eta).
\end{align*}

Taking $F(\eta,x)=1[|\eta|=k]$, 

\begin{align*}
    \int_{\Omega}k1[|\eta|=k]d\mu_{\bm\lambda}(\eta)&=\int \int e^{-\nabla_x^+H(\eta)}1[|\eta|+1=k]\bm\lambda(x)\, dx \, d\mu_{\bm\lambda}(\eta)\\
    &\le \int\int1[|\eta|=k-1]\bm\lambda(x)\,dx\,d\mu_{\bm\lambda}(\eta).
\end{align*}
Hence,
\begin{align*}
    k\Pr_{\eta\sim\mu_{\bm\lambda}}(|\eta|=k)
    &\le\bm\lambda(\Lambda)\Pr_{\eta\sim\mu_{\bm\lambda}}(|\eta|=k-1)
\end{align*}
\end{proof}

\begin{lemma}\label{lem:canonical:maybe_at_least_expectation}
For $\eta\sim\mu_{\bm\lambda}$,
\begin{align*}
    \Pr(|\eta|\ge\lfloor\E[|\eta|]\rfloor)\ge \frac{1}{1+\bm\lambda(\Lambda)}
\end{align*}
\end{lemma}

\begin{proof}
    Let $k^*=\E[|\eta|]$. Then,
    \begin{align*}
        \Pr(|\eta|\le\lfloor k^*\rfloor-1)&=\sum_{k\le k^*-1}\Pr(|\eta|=k)\\
        &\le \sum_{k\le k^*-1}\Pr(|\eta|=k)(k^*-k)\\
        &\le \sum_{k\le k^*}\Pr(|\eta|=k)(k^*-k)\\
        &=\sum_{k>k^*}\Pr(|\eta|=k)(k-k^*)\\
        &\le\sum_{k\ge\lfloor k^*\rfloor+1}\Pr(|\eta|=k)k\\
        &\le\sum_{k\ge\lfloor k^*\rfloor+1}\Pr(|\eta|=k-1)\bm\lambda(\Lambda)\\
        &=\Pr(|\eta|\ge\lfloor k^*\rfloor)\bm\lambda(\Lambda)
    \end{align*}
    Thus,
    \begin{align*}
        1&\le \Pr(|\eta|\ge\lfloor k^*\rfloor)(\bm\lambda(\Lambda)+1)
    \end{align*}
\end{proof}

\begin{lemma}
Let $k\ge 0$ be integer. For all $t>0$,
\begin{align*}
    0\le \frac{d}{dt}\Pr_{\eta\sim\mu_{t\bm\lambda}}(|\eta|\ge k)\le \bm\lambda(\Lambda)
\end{align*}
\end{lemma}
\begin{proof}
    For integer $j\ge 0$, let
    \begin{align*}
        w_j&:=\frac{1}{j!}\int_{\Lambda^k}e^{-H(x_1,\ldots,x_j)}\prod_{i=1}^j\bm\lambda(x_i)dx
    \end{align*}
    Then,
    \begin{align*}
        \Pr_{\eta\sim \mu_{t\bm\lambda}}(|\eta|\ge k)&=\frac{\sum_{j=0}^{\infty}w_jt^j1[j\ge k]}{\sum_{j=0}^{\infty}w_jt^j}
    \end{align*}
    \begin{align*}
        \frac{d}{dt}\Pr_{\eta\sim\mu_{t\bm\lambda}}(|\eta|\ge k)&=\frac{\sum_{j=0}^{\infty}w_jt^j1[j\ge k]\cdot\frac{j}{t}}{\sum_{j=0}^{\infty}w_jt^j}-\frac{\sum_{j=0}^{\infty}w_jt^j1[j\ge k]}{\sum_{j=0}^{\infty}w_jt^j}\cdot\frac{\sum_{j=0}^{\infty}w_jt^j\frac{j}{t}}{\sum_{j=0}^{\infty}w_jt^j}\\
        &=\E_{\eta\sim \mu_{t\bm\lambda}}\left[1[|\eta|\ge k]\cdot\frac{|\eta|}{t}\right]-\E_{\eta\sim \mu_{t\bm\lambda}}[1[|\eta|\ge k]]\cdot \E_{\eta\sim\mu_{t\bm\lambda}}\left[\frac{|\eta|}{t}\right]
    \end{align*}

    For each $t$, this is the covariance of two non-decreasing functions of $|\eta|$. Hence, it is non-negative, and thus 
    \[ \frac{d}{dt}\Pr_{\eta\sim\mu_{t\bm\lambda}}(|\eta|\ge k) \geq 0.\]
    Also,
    \begin{align*}
        \frac{d}{dt}\Pr_{\eta\sim\mu_{t\bm\lambda}}(|\eta|\ge k)\le \E_{\eta\sim \mu_{t\bm\lambda}}\left[1[|\eta|\ge k]\cdot\frac{|\eta|}{t}\right]&\le \E_{\eta\sim \mu_{t\bm\lambda}}
\left[\frac{|\eta|}{t}\right] \\
        \text{ (by stochastic domination by Poisson point process)} &\le \frac{t\bm\lambda(\Lambda)}{t}=\bm\lambda(\Lambda).
    \end{align*}
\end{proof}

We assume $\bm\lambda(\Lambda)\ge 1$ to simplify the arithmetic. This is no loss as our final theorem (\Cref{thm:canonical:main}) only applies when $1\le k\le \E_{\eta\sim\bm\lambda}[|\eta|]$, and there is no such $k$ when $\bm\lambda(\Lambda)<1$. 

\begin{corollary}\label{lem:canonical:good_interval_of_activity}
    Let $\bm\lambda\ge0$ and $1 \le k\le \E_{\eta\sim\mu_{\bm\lambda}}[|\eta|]$ be a positive integer. Then, there exists $0\le a\le b\le 1$ with $b-a\ge\frac{1}{4(\bm\lambda(\Lambda))^2}$ such that for all $t\in[a,b]$,
    \begin{align*}
        \frac{1}{4\bm\lambda(\Lambda)}\le \Pr_{\eta\sim\mu_{t\bm\lambda}}(|\eta|\ge k)\le \frac{1}{2\bm\lambda(\Lambda)}.
    \end{align*}
\end{corollary}
\begin{proof}
By \Cref{lem:canonical:maybe_at_least_expectation},
\begin{align*}
    \Pr_{\eta\sim\mu_{\bm\lambda}}(|\eta|\ge k)&\ge\Pr_{\eta\sim\mu_{\bm\lambda}}(|\eta|\ge \lfloor\E[|\eta|]\rfloor)\ge \frac{1}{2\bm\lambda(\Lambda)}
\end{align*}
Also, using the fact that $k\geq 1$,
\begin{align*}
    \Pr_{\eta\sim\mu_{0\bm\lambda}}(|\eta|\ge k)&=0
\end{align*}

By continuity, there are $a,b\in[0,1]$ such that
\begin{align*}
    \Pr_{\eta\sim\mu_{a\bm\lambda}}(|\eta|\ge k)&=\frac{1}{4\bm\lambda(\Lambda)}&
    \Pr_{\eta\sim\mu_{b\bm\lambda}}(|\eta|\ge k)&=\frac{1}{2\bm\lambda(\Lambda)}
\end{align*}

Since
\begin{align*}
0\le \frac{d}{dt}\Pr_{\eta\sim\mu_{t\bm\lambda}}(|\eta|\ge k)\le \bm\lambda(\Lambda)
\end{align*}
we have
\begin{align*}
    \Pr_{\eta\sim\mu_{b\bm\lambda}}(|\eta|\ge k)-\Pr_{\eta\sim\mu_{a\bm\lambda}}(|\eta|\ge k)\le \bm\lambda(\Lambda)(b-a)
\end{align*}
Hence,
\begin{align*}
    b-a&\ge \frac{1}{4\bm\lambda(\Lambda)^2}
\end{align*}
\end{proof}

We can lower bound the probability of getting a set of size exactly $k$ using a spectral gap of continuum Glauber.

\begin{lemma}\label{lem:canonical:spectral_gap_lower_bounds_probability_of_k}
    Suppose continuum Glauber on $\mu$ has spectral gap $\gamma$. Let $k$ be a nonnegative integer. Then, for $\eta\sim\mu$,
    \begin{align*}
        k\Pr_{\eta\sim\mu}(|\eta|=k)\ge \gamma \Pr_{\eta\sim\mu}(|\eta|\ge k)\Pr_{\eta\sim\mu}(|\eta|<k).
    \end{align*}
\end{lemma}
\begin{proof}
    The spectral gap implies that for bounded measurable $f:\Omega\to\R$, we have
    \begin{align*}
        \mcE_{\mu}(f,f)\ge \gamma \Var_{\mu}(f).
    \end{align*}
    Expanding, we have
    \begin{align*}
        \int_{\Omega}\sum_{x\in\eta}(f(\eta)-f(\eta\setminus\{x\}))^2d\mu(\eta)&\ge \gamma \Var_{\eta\sim\mu}[f(\eta)].
    \end{align*}
    Plugging in $f(\eta)=1_{|\eta|\ge k}$, we have
    \begin{align*}
        \int |\eta|1_{|\eta|=k}d\mu(\eta)\ge \gamma\Var_{\eta\sim\mu}[1_{|\eta|\ge k}]
    \end{align*}
    \begin{align*}
        k\Pr_{\eta\sim\mu}(|\eta|=k)\ge \gamma \Pr_{\eta\sim\mu}(|\eta|\ge k)\Pr_{\eta\sim\mu}(|\eta|<k)
    \end{align*}
\end{proof}

We are now ready to state the algorithm.

\begin{algorithm}[H]
\caption{Sampling from canonical Gibbs distribution}
\label{alg:canonical_sampler}
\begin{algorithmic}[1]
\State Let $\gamma$ be a lower bound on the spectral gap of continuum Glauber on $\mu_{t\bm\lambda}$ for all $t\in[0,1]$
\State Let $m:=\lceil 4\bm\lambda(\Lambda)^2(\frac{8}{\gamma}\bm\lambda(\Lambda)^2\log\frac{2}{\delta}+1)\rceil$
\For{$j=0,1,\ldots,m$}
\State Sample $\eta$ from a distribution within $\frac{\delta}{2m}$ TV distance of $\mu_{j\bm\lambda/m}$.
\State If $|\eta|=k$, return $\eta$
\EndFor
\State Return $\emptyset$
\end{algorithmic}
\end{algorithm}

The following lemma shows that it is unlikely that none of the iterations were successful.
\begin{lemma}\label{lem:canonical:good_sample_likely}
Suppose $k\le\E_{\eta\sim\mu_{\bm\lambda}}[|\eta|]$ for some $\bm\lambda\ge 0$.

Suppose the spectral gap for continuum Glauber on $\mu_{t\bm \lambda}$ is at least $\gamma$ for all $t\in[0,1]$.

Let $\delta>0$. Let
\begin{align*}
    m=\lceil 4\bm\lambda(\Lambda)^2(\frac{8}{\gamma}\bm\lambda(\Lambda)^2\log\frac{1}{\delta}+1)\rceil.
\end{align*}
Then,
\begin{align*}
    \prod_{j=0}^m\Pr_{\eta\sim\mu_{j\bm\lambda/m}}(|\eta|\ne k)\le \delta.
\end{align*}
\end{lemma}
\begin{proof}
By \Cref{lem:canonical:good_interval_of_activity}, there exist $0\le a<b\le 1$ such that $b-a\ge\frac{1}{4\bm\lambda(\Lambda)^2}$ and
\begin{align*}
    \Pr_{\eta\sim\mu_{t\bm\lambda}}(|\eta|<k)\Pr_{\eta\sim\mu_{t\bm\lambda}}(|\eta|\ge k)&\ge\frac{1}{8\bm\lambda(\Lambda)}&\text{ for }t\in[a,b]
\end{align*}

By \Cref{lem:canonical:spectral_gap_lower_bounds_probability_of_k}, for any $t\in[a,b]$,
\begin{align*}
    \Pr_{\eta\sim\mu_{t\bm\lambda}}(|\eta|=k)\ge \frac{\gamma}{k}\cdot\frac{1}{8\bm\lambda(\Lambda)}\ge \frac{\gamma}{8\bm\lambda(\Lambda)^2}
\end{align*}
where the second inequality follows from $k\le \E_{\eta\sim\mu_{\bm\lambda}}[|\eta|]\le \bm\lambda(\Lambda)$.

In particular, this holds for $t\in \frac{1}{m}\Z\cap [a,b]$, and
\begin{align*}
    \left|\frac{1}{m}\Z\cap [a,b]\right|\ge \frac{8\bm\lambda(\Lambda)^2}{\gamma}\log\frac{1}{\delta}
\end{align*}

Hence,
\begin{align*}
    \prod_{j=0}^m\left(1-\Pr_{\eta\sim\mu_{j\bm\lambda/m}}(|\eta|=k)\right)\le \exp\left(-\sum_{j=0}^m\Pr_{\eta\sim\mu_{j\bm\lambda/m}}(|\eta|=k)\right)\le \exp(-\log\frac{1}{\delta})\le \delta
\end{align*}

\end{proof}

\begin{theorem}\label{thm:canonical:main}
Suppose $k\le \E_{\eta\sim\mu_{\bm\lambda}}[|\eta|]$ is a positive integer and $0<\gamma<1$ is a lower bound on the spectral gap of continuum Glauber for $\mu_{t\bm\lambda}$ for all $t\in[0,1]$. Let $\delta\in(0,\frac{1}{2})$. Then, the time complexity of
\Cref{alg:canonical_sampler} is $O(\frac{1}{\gamma}\bm\lambda(\Lambda)^4\log\frac{1}{\delta})$ times the cost of approximately sampling from $\mu_{j\bm\lambda/m}$. It outputs $\eta$ from a distribution within $\delta$ TV distance of $\mu_{\bm\lambda}$ conditioned on $|\eta|=k$.
\end{theorem}
\begin{proof}
The time complexity bound follows from observing that $m=O(\frac{1}{\gamma}\bm\lambda(\Lambda)^4\log\frac{1}{\delta})$.
    We optimally couple a run of the algorithm with an idealized run where $\eta$ is truly drawn from $\mu_{j\bm\lambda/m}$. By the union bound, the probability that the algorithm differs from the idealized run is at most $\frac{\delta}{2m}\cdot m=\frac{\delta}{2}$. Now, we just need to analyze the idealized version.

    If we condition on the algorithm stopping on iteration $j$, then the distribution of the returned $\eta$
is $\mu_{j\bm\lambda/m}$ conditioned on $|\eta|=k$, which equals $\mu_{\bm\lambda}$ conditioned on $|\eta|=k$. The probability the algorithm does not stop on some iteration $j\in\{0,\ldots,m\}$ is at most $\delta/2$ by \Cref{lem:canonical:good_sample_likely}.
Hence, the output of \Cref{alg:canonical_sampler} is within $\delta$ TV distance of $\mu_{\bm\lambda}$ conditioned on $|\eta|=k$.
\end{proof}

\begin{corollary}\label{cor:canonical:main_with_runtime}
Suppose $k\le \E_{\eta\sim\mu_{\bm\lambda}}[|\eta|]$ is a positive integer and $0<\gamma<1$ is a lower bound on the spectral gap of continuum Glauber for $\mu_{t\bm\lambda}$ for all $t\in[0,1]$. Then, there is an algorithm that outputs a sample $\eta$ from a distribution within $\delta$ TV distance of $\mu_{\bm\lambda}$ conditioned on $|\eta|=k$, which has an expected time complexity bounded by a polynomial in $\bm\lambda(\Lambda),\frac{1}{\gamma},$ and $\log\frac{1}{\delta}$.

\end{corollary}

\begin{proof}
We run \Cref{alg:canonical_sampler} and use \Cref{alg:birth_death} (initialized to $\emptyset$) to approximately sample from each $\mu_{j\bm\lambda/m}$.

By \Cref{cor:mixing:sampling_algorithm_runtime}, using \Cref{alg:birth_death} to get a sample within $\epsilon=\frac{\delta}{2m}$ TV distance of $\mu_{j\bm\lambda/m}$ will have an expected time complexity of
\begin{align*}
    O((\frac{1}{\gamma}\bm\lambda(\Lambda)(\bm\lambda(\Lambda)+\log\frac{2m}{\delta}))^2)&=O((\frac{1}{\gamma}\bm\lambda(\Lambda)(\bm\lambda(\Lambda)+\log\frac{1}{\gamma}+\log\bm\lambda(\Lambda)+\log\log\frac{1}{\delta}+\log\frac{1}{\delta}))^2)\\
    &=O((\frac{1}{\gamma^2}\bm\lambda(\Lambda)^2(\bm\lambda(\Lambda)+\log\frac{1}{\gamma}+\log\frac{1}{\delta})^2))
\end{align*}
Doing this $m=O(\frac{1}{\gamma}\bm\lambda(\Lambda)^4\log\frac{1}{\delta})$ times will result in a expected time complexity of
\begin{align*}
    &O((\frac{1}{\gamma^2}\bm\lambda(\Lambda)^2(\bm\lambda(\Lambda)+\log\frac{1}{\gamma}+\log\frac{1}{\delta})^2\cdot\frac{1}{\gamma}\bm\lambda(\Lambda)^4\log\frac{1}{\delta})\\
    &=O((\frac{1}{\gamma^3}\bm\lambda(\Lambda)^6(\bm\lambda(\Lambda)+\log\frac{1}{\gamma}+\log\frac{1}{\delta})^2\log\frac{1}{\delta})
\end{align*}
\end{proof}

\subsection{Bounds on expected cardinality}

In this section, we discuss known lower bounds on the expected size of a sample from Gibbs point process, to give a better idea of what values of $k$ are allowed in \Cref{cor:canonical:main_with_runtime}.

We start with a simpler lower bound that applies to all Gibbs point processes with repulsive pair potentials.

\begin{lemma}[\cite{MP22Analyticity} Lemma 24]\label{lem:canonical:expectation_lower_bound}
Let $\mu$ be a Gibbs point process with repulsive pair potential $\phi$ and activity $\lambda$.
Recall that
\begin{align*}
    C_{\phi}&=\sup_{y\in\Lambda}\int 1-e^{-\phi(x,y)}dx
\end{align*}

Then,
\begin{align*}
    \E_{\eta\sim\mu}[|\eta|]&\ge \frac{\lambda}{1+\lambda C_{\phi}}|\Lambda|
\end{align*}
\end{lemma}

For the hard sphere model, we can use better known bounds.
\begin{lemma}[\cite{jenssen2019hard} proof of Theorem 2, c.f. \cite{HPP21} Lemma 5.1]\label{lem:canonical:expected_density_lower_bound}
For the hard sphere model on bounded region $\Lambda\subseteq\R^d$ with activity $\lambda$,
\begin{align*}
    \frac{\E_{\eta\sim\mu_{\lambda}}[|\eta|]}{|\Lambda|}\ge\inf_{z\in\R}\max\{\lambda e^{-z},z\cdot 2^{-d}e^{-\lambda\cdot 2\cdot3^{d/2}}\}
\end{align*}

In particular, if $\lambda=\frac{c}{2^d}$, as $d\to\infty$, the right hand side is $(1-o(1))\frac{W(c)}{2^d}$ where $W(\cdot)$ is the Lambert W function, the functional inverse of $x\mapsto xe^x$ on $\R_{\ge 0}$.
\end{lemma}

We can now prove \Cref{thm:canonical_hard_sphere_density}, which we restate for the reader's convenience.

\thmcanonical*

\begin{proof}[Proof of Thm. ~\ref{thm:canonical_hard_sphere_density}.]
    Choose $c\in(0,e)$ so that $W(c)>1-\epsilon/4$. (Note that $W(\cdot)$ is continuous and increasing and $W(e)=1$). Choose $d_0$ large enough so that by \Cref{lem:canonical:expected_density_lower_bound}, for $d\ge d_0$ and $\lambda=\frac{c}{2^d}$,
    \begin{align*}
        \E_{\eta\sim\mu_{\lambda}}[|\eta|]\ge (1-\epsilon/2)\frac{1}{2^d}|\Lambda|
    \end{align*}
    Finally, apply \Cref{cor:canonical:main_with_runtime} using the lower bound on the the spectral gap of continuum Glauber for the hard sphere model with activity $\lambda=\frac{c}{2^d}$ from \Cref{thm:spectral_gap}.
\end{proof}

\bibliographystyle{alpha}
\bibliography{references}

\appendix

\section{Basic properties of point processes}
In this section, we will list some basic properties of point processes.

The following theorem can be found in a reference on point processes such as \cite{Last_Penrose_2017}.

\begin{theorem}[Mecke equation]\label{thm:mecke_equation}
For any measurable 
$F:\Omega\times\Lambda\to[0,\infty)$,
\begin{align*}
    \int \sum_{x\in\eta}F(x,\eta)d\rho_{\bm\lambda}(\eta)&=\int \int F(x,\eta\cup\{x\})d\rho_{\bm\lambda}(\eta)\bm\lambda(x) dx
\end{align*}
\end{theorem}

\begin{lemma}[Gibbs point process]
    For the Gibbs point process with activity function $\bm\lambda:\Lambda\to\R_{\ge 0}$ and pair potential $\phi:\Lambda\times\Lambda\to(-\infty,\infty]$ with a probability distribution $\mu_{\bm\lambda}$ on $\Omega$, the following holds:
    \begin{align}
        d\mu_{\bm\lambda}(\eta)&=\frac{e^{-H(\eta)}}{\int e^{-H(\xi)}d\rho_{\bm\lambda}(\xi)}d\rho_{\bm\lambda}(\eta)\label{eqn:gpp_ppp_density}
    \end{align}
    where
    \begin{align*}
        H(\eta)=\sum_{\substack{\{x,y\}\subseteq\eta\\x\ne y}}\phi(x,y).
    \end{align*}
\end{lemma}

\begin{theorem}[GNZ equations]\label{thm: GNZ equations}
For any measurable 
$F:\Omega\times\Lambda\to[0,\infty)$,
\begin{align*}
\int_{\Omega}\sum_{x\in\eta}F(\eta,x)d\mu_{\bm\lambda}(\eta)&=\int \int e^{-\nabla_x^+H(\eta)}F(\eta\cup\{x\},x)\bm\lambda(x)\,dx \, d\mu_{\bm\lambda}(\eta).
\end{align*}
\end{theorem}
This version of the GNZ equations \cite{georgii1976canonical,nguyen1977integral} can be found in \cite{Jansen19}. When $\bm\lambda(\Lambda)<\infty$, it can be shown using the Mecke equation.
\begin{proof}
Plugging the function $F(x,\eta)e^{-H(\eta)}$ into the Mecke equation, we get
\begin{align*}
    \int \sum_{x\in\eta}F(x,\eta)e^{-H(\eta)}d\rho_{\bm\lambda}(\eta)&=\int \int F(x,\eta\cup\{x\})e^{-H(\eta\cup\{x\})}d\rho_{\bm\lambda}(\eta)\bm\lambda(x) dx\\
    &=\int \int F(x,\eta\cup\{x\})e^{-\nabla_x^+H(\eta)}e^{-H(\eta)}d\rho_{\bm\lambda}(\eta)\bm\lambda(x) dx
\end{align*}
Dividing both sides by $\int e^{-H(\xi)}d\rho_{\bm\lambda}(\xi)$ and using \Cref{eqn:gpp_ppp_density} completes the proof.
\end{proof}

\begin{lemma}[Intensity vs one-point density]\label{lem:intensity_vs_one_point_density}
    \begin{align*}
        \E_{\eta\sim\mu_{\bm\lambda}}[\eta(B)]=\int_B\frac{Z(\bm\lambda e^{-\phi(x,\cdot)})}{Z(\bm\lambda)}\bm\lambda(x)dx
    \end{align*}
\end{lemma}
\begin{proof}
    Let $f:\Lambda\to\R$ be a bounded measurable function. Then, the GNZ equation with $F(\eta,x)=f(x)$ shows that
    \begin{align*}
        \E_{\eta\sim\mu_{\bm\lambda}}[\eta(f)]=\int\sum_{x\in\eta}f(x)d\mu_{\bm\lambda}(\eta)&=\int\int e^{-\nabla_x^+H(\eta)}f(x)\bm\lambda(x)dxd\mu_{\bm\lambda}\\
        &=\int\E_{\eta\sim\mu_{\bm\lambda}}[e^{-\nabla_x^+H(\eta)}]f(x)\bm\lambda(x)dx
    \end{align*}
    \begin{align*}
        \E_{\eta\sim\mu_{\bm\lambda}}[e^{-\nabla_x^+H(\eta)}]&=\frac{\int e^{-\nabla_x^+H(\eta)}e^{-H(\eta)}d\rho_{\bm\lambda}(\eta)}{\int e^{-H(\eta)}d\rho_{\bm\lambda}(\eta)}\\
        &=\frac{\int e^{-H(\eta\cup\{x\})}d\rho_{\bm\lambda}(\eta)}{\int e^{-H(\eta)}d\rho_{\bm\lambda}(\eta)}\\
        &=\frac{\sum_{k\ge 0}\frac{1}{k!}\int_{\Lambda^k}\left(\prod_{j=1}^k\bm\lambda(y_j)\right)e^{-H(\vec{y},x)} d\vec{y}}{\sum_{k\ge 0}\frac{1}{k!}\int_{\Lambda^k}\left(\prod_{j=1}^k\bm\lambda(y_j)\right)e^{-H(\vec{y})} d\vec{y}}\\
        &=\frac{\sum_{k\ge 0}\frac{1}{k!}\int_{\Lambda^k}\left(\prod_{j=1}^k(\bm\lambda(y_j)e^{-\phi(x,y_j)})\right)e^{-H(\vec{y})} d\vec{y}}{\sum_{k\ge 0}\frac{1}{k!}\int_{\Lambda^k}\left(\prod_{j=1}^k\bm\lambda(y_j)\right)e^{-H(\vec{y})} d\vec{y}}\\
        &=\frac{Z(\bm\lambda e^{-\phi(x,\cdot)})}{Z(\bm\lambda)}
    \end{align*}
\end{proof}

\begin{lemma}[c.f. \cite{MP22Strong} Lemma 13]\label{lem:gpp:intensity_bounds}
For a Gibbs point process with activity $\bm\lambda:\Lambda\to[0,\lambda]$ and repulsive pair potential $\phi$ on $\Lambda$,
    \begin{align*}
         \bm\lambda(x)e^{-\lambda C_{\phi}}&\le \zeta_{\bm\lambda}(x)\le \bm\lambda(x)
    \end{align*}
\end{lemma}
\begin{proof}
    \begin{align*}
        \zeta_{\bm\lambda}(x)&=\bm\lambda(x)\E_{\eta\sim\mu_{\bm\lambda}}[\exp(-\sum_{y\in\eta}\phi(y,x))]\le \bm\lambda(x)\\
        \zeta_{\bm\lambda}(x)&=\bm\lambda(x)\E_{\eta\sim\mu_{\bm\lambda}}[\exp(-\sum_{y\in\eta}\phi(y,x))]\\
        &\ge \bm\lambda(x)\E_{\eta\sim \rho_{\bm\lambda}}[\exp(-\sum_{y\in\eta}\phi(y,x))]&\text{stochastic domination}\\
        &=\bm\lambda(x)\exp\left(-\int_{\Lambda}(1-e^{-\phi(y,x)})\bm\lambda(x)dy\right)&\text{by \Cref{lem:ppp:product}}\\
        &\ge\bm\lambda(x)\exp(-\lambda C_{\phi})
    \end{align*}
\end{proof}

\begin{lemma}\label{lem:ppp:product} Let $f:\Lambda\to\R_{\ge 0}$ be a bounded measurable function. Then,
\begin{align*}
    \E_{\eta\sim \rho_{\bm\lambda}}[\prod_{x\in\eta}f(x)]=\exp(\int_{\Lambda} (f(x)-1)\bm\lambda(x)dx)
\end{align*}
\end{lemma}

\begin{proof}

For any $\epsilon>0$, we can partition $\Lambda$ into $A_1,\ldots,A_N$ such that
\begin{align*}
    f_i^+&=\sup_{x\in A_i}f(x)&
    f_i^-&=\inf_{x\in A_i}f(x)
\end{align*}
satisfy $|f_i^+-f_i^-|\le \epsilon$ for all $i$.
Then,
\begin{align*}
    \E_{\eta\sim \rho_{\bm\lambda}}[\prod_{x\in\eta}f(x)]&=\prod_{i=1}^N\E_{\eta\sim \rho_{\bm\lambda 1_{A_i}}}[\prod_{x\in\eta}f(x)]&\text{independence}\\
    &\le \prod_{i=1}^N\E_{\eta\sim \rho_{\bm\lambda 1_{A_i}}}[\prod_{x\in\eta}f_i^+]\\
    &=\prod_{i=1}^N\E_{X\sim \Pois(\bm\lambda(A_i))}[(f_i^+)^X]\\
    &=\prod_{i=1}^N\exp(\bm\lambda(A_i)(f_i^+-1))&\text{MGF of Poisson}\\
    &=\exp(\sum_{i=1}^N\bm\lambda(A_i)(f_i^+-1))\\
    &\le \exp(\int_{\Lambda}\bm\lambda(x)(f(x)+\epsilon-1)dx)\\
    &=\exp(\int_{\Lambda}\bm\lambda(x)(f(x)-1)dx)e^{\epsilon\bm\lambda(\Lambda)}
\end{align*}
Taking $\epsilon\to 0$ shows that
\begin{align*}
    \E_{\eta\sim \rho_{\lambda}}[\prod_{x\in\eta}f(x)]&\le \exp(\int_{\Lambda}\bm\lambda(x)(f(x)-1)dx)
\end{align*}
A similar argument shows the reverse inequality.
\end{proof}

\section{Spectral gap at low activity}\label{sec: appendix low activity}

\cite{KL03} proved a lower bound on the spectral gap of continuum Glauber for repulsive point processes on $\R^d$ at low activity. However, their proof of this applies more generally, which we justify here.

We start by proving a coercivity identity.

\begin{lemma}[Coercivity identity]\label{lem:low_activity_coercivity}
Let $\mathcal{C}$ be the class of bounded measurable $f:\Omega\to\R$ such that there is some $M>0$ and $c\in\R$ for which $f(\eta)=c$ when $|\eta|\ge M$. For $f\in\mathcal{C}$,
\begin{align*}
    \int_{\Omega}(\mcL f(\eta))^2d\mu(\eta)=\int_{\Omega}\Bigg(&\sum_{x\in\eta}(f(\eta)-f(\eta\setminus\{x\}))^2+\sum_{\substack{x,y\in\eta\\x\ne y}}(f(\eta\setminus\{x,y\}-f(\eta\setminus\{x\})-f(\eta\setminus\{y\})+f(\eta))^2\\
    &+\sum_{\substack{x,y\in\eta\\x\ne y}}(e^{\phi(x,y)}-1)(f(\eta\setminus\{x,y\})-f(\eta\setminus\{x\}))(f(\eta\setminus\{x,y\})-f(\eta\setminus\{y\}))\Bigg)d\mu(\eta)
\end{align*}
\end{lemma}
\begin{proof}
    Let $f\in\mathcal{C}$. Define
    \begin{align*}
        g(\eta):=\int e^{-\nabla_x^+H(\eta)}(f(\eta\cup\{x\})-f(\eta))\bm\lambda(x)dx
    \end{align*}
    so that $\mcL f(\eta)=g(\eta)+\sum_{x\in\eta}(f(\eta\setminus\{x\})-f(\eta))$. Note that $g$ and $\mcL f$ are bounded.
    
    Using the GNZ equations twice,
    \begin{align*}
        \int g(\eta)^2d\mu(\eta)&=\int \int e^{-\nabla_x^+H(\eta)}g(\eta)(f(\eta\cup\{x\})-f(\eta))\bm\lambda(x)dxd\mu(\eta)\\
        &=\int\sum_{x\in\eta}g(\eta\setminus\{x\})(f(\eta)-f(\eta\setminus\{x\}))d\mu(\eta)\\
        &=\int\sum_{x\in\eta}\left(\int e^{-\nabla_y^+H(\eta)+\phi(x,y)}(f(\eta\setminus\{x\}\cup\{y\})-f(\eta\setminus\{x\}))\bm\lambda(y)dy\right)(f(\eta)-f(\eta\setminus\{x\}))d\mu(\eta)\\
        &=\int\left(\int e^{-\nabla_y^+H(\eta)}\sum_{x\in\eta}e^{\phi(x,y)}(f(\eta\setminus\{x\}\cup\{y\})-f(\eta\setminus\{x\}))(f(\eta)-f(\eta\setminus\{x\})\bm\lambda(y)dy\right)d\mu(\eta)\\
        &=\int\left(\sum_{y\in\eta}\sum_{x\in\eta\setminus\{y\}}e^{\phi(x,y)}(f(\eta\setminus\{x\})-f(\eta\setminus\{x,y\}))(f(\eta\setminus\{y\})-f(\eta\setminus\{x,y\}))\right)d\mu(\eta)\\
    \end{align*}

    To justify use of GNZ equations, note that if we replace each difference of $f$s above with the absolute value of the difference, the above manipulations will hold as all terms are nonnegative. This means the relevant terms above are integrable, so we can apply \Cref{thm: GNZ equations} to the positive and negative parts.
    
    Again by the GNZ equations,
    \begin{align*}
        2\int g(\eta)\sum_{x\in\eta}(f(\eta\setminus\{x\})-f(\eta))d\mu(\eta)&=
        2\int\int e^{-\nabla_y^+H(\eta)}\sum_{x\in\eta}(f(\eta\cup\{y\})-f(\eta))(f(\eta\setminus\{x\})-f(\eta))\bm\lambda(y)dyd\mu(\eta)\\
        &=2\int\sum_{y\in\eta}\sum_{x\in\eta\setminus\{y\}}(f(\eta)-f(\eta\setminus\{y\}))(f(\eta\setminus\{x,y\})-f(\eta\setminus\{y\}))d\mu(\eta)\\
        &=\int\sum_{\substack{x,y\in\eta\\x\ne y}}\Big((f(\eta)-f(\eta\setminus\{y\}))(f(\eta\setminus\{x,y\})-f(\eta\setminus\{y\}))\\
        &\quad\quad\quad\quad+
        (f(\eta)-f(\eta\setminus\{x\}))(f(\eta\setminus\{x,y\})-f(\eta\setminus\{x\}))\Big)d\mu(\eta)
    \end{align*}
    Also,
    \begin{align*}
        \int \left(\sum_{x\in\eta}(f(\eta\setminus\{x\})-f(\eta))\right)^2d\mu(\eta)&=\int\sum_{x\in\eta}(f(\eta\setminus\{x\})-f(\eta))^2\\
        &\quad+\sum_{\substack{x,y\in\eta\\x\ne y}}(f(\eta\setminus\{x\})-f(\eta))(f(\eta\setminus\{y\})-f(\eta))d\mu(\eta)
    \end{align*}
    We get the coercivity identity by summing the equations and using that
    \begin{align*}
        (f(\eta\setminus\{x,y\}-f(\eta\setminus\{x\})-f(\eta\setminus\{y\})+f(\eta))^2&=(f(\eta\setminus\{x,y\})-f(\eta\setminus\{x\}))(f(\eta\setminus\{x,y\})-f(\eta\setminus\{y\}))\\
        &\quad+(f(\eta)-f(\eta\setminus\{y\}))(f(\eta\setminus\{x,y\})-f(\eta\setminus\{y\}))\\
        &\quad+(f(\eta\setminus\{x,y\})-f(\eta\setminus\{x\}))(f(\eta)-f(\eta\setminus\{x\}))\\
        &\quad+(f(\eta)-f(\eta\setminus\{y\}))(f(\eta)-f(\eta\setminus\{x\})).
    \end{align*}
\end{proof}

\lowactivityspectralgap*

\begin{proof}[Proof of \Cref{thm:low_activity_spectral_gap}]
    Note that
    \begin{align*}
        \sum_{\substack{x,y\in\eta\\x\ne y}}(f(\eta\setminus\{x,y\}-f(\eta\setminus\{x\})-f(\eta\setminus\{y\})+f(\eta))^2\ge 0
    \end{align*}
    and
    \begin{align*}
        &\left|\int\sum_{\substack{x,y\in\eta\\x\ne y}}(e^{\phi(x,y)}-1)(f(\eta\setminus\{x,y\})-f(\eta\setminus\{x\}))(f(\eta\setminus\{x,y\})-f(\eta\setminus\{y\}))d\mu(\eta)\right|\\
        &\le\int\sum_{\substack{x,y\in\eta\\x\ne y}}(e^{\phi(x,y)}-1)(f(\eta\setminus\{x,y\})-f(\eta\setminus\{y\}))^2d\mu(\eta)\tag{by Cauchy-Schwarz}\\
        &=\int\sum_{y\in\eta}\sum_{x\in\eta\setminus\{y\}}(e^{\phi(x,y)}-1)(f(\eta\setminus\{x,y\})-f(\eta\setminus\{y\}))^2d\mu(\eta)\\
        &=\int\int e^{-\nabla_x^+H(\eta)}\sum_{x\in\eta}(e^{\phi(x,y)}-1)(f(\eta\setminus\{x\})-f(\eta))^2\bm\lambda(y)dyd\mu(\eta)\tag{GNZ equations}\\
        &=\int\int e^{-\nabla_x^+H(\eta\setminus\{y\})}\sum_{x\in\eta}(1-e^{-\phi(x,y)})(f(\eta\setminus\{x\})-f(\eta))^2\bm\lambda(y)dyd\mu(\eta)\\
        &\le\int\int\sum_{x\in\eta}(1-e^{-\phi(x,y)})(f(\eta\setminus\{x\})-f(\eta))^2\bm\lambda(y)dyd\mu(\eta)\\
        &=\int\sum_{x\in\eta}\left(\int(1-e^{-\phi(x,y)})\bm\lambda(y)dy\right)(f(\eta\setminus\{x\})-f(\eta))^2d\mu(\eta)\\
        &\le \lambda C_{\phi}\int\sum_{x\in\eta}(f(\eta\setminus\{x\})-f(\eta))^2d\mu(\eta)\\
        &\le (1-\delta)\mcE_{\mu}(f,f)
    \end{align*}
    Hence, by \Cref{lem:low_activity_coercivity},
    \begin{align*}
        \int(\mcL f(\eta))^2d\mu(\eta)\ge \mcE_{\mu}(f,f)-(1-\delta)\mcE_{\mu}(f,f)=\delta \mcE_{\mu}(f,f).
    \end{align*}
    This can be rewritten as
    \begin{align*}
        \langle-\mcL f,-\mcL f\rangle_{L^2(\mu)}\ge \delta\langle-\mcL  f,f\rangle_{L^2(\mu)}.
    \end{align*}
    Proceeding as in \cite{KL03}, we obtain
    \begin{align*}
        \mcE_{\mu}(f,f)\ge \delta\Var_{\eta\sim\mu}[\eta(f)].
    \end{align*}
\end{proof}

\section{Jump-type Markov processes}\label{sec: jump type}

Here we review the definitions of paths and pure jump-type processes for the reader's convenience. These are used in Section \ref{sec:mixing_time}.

\begin{definition}[Right-continuous path]
    For a fixed outcome $\omega$, a path of a process is a map from time to the state space 
    \[ t \to X_t(\omega) \] 
    We say that a path is right continuous if for any time $t\geq 0$, 
    \[ \lim_{s\downarrow t}X_s = X_t\]
    where $s\downarrow t$ is a limit from the right. 
\end{definition}

\begin{definition}[Pure jump type, Chapter 13 of \cite{kallenberg2021foundations} (verbatim)]\label{def: jump type}
Say that a process $X$ in some measurable space $(S,\mcS)$ is of \emph{pure jump type} if its paths are a.s. right-continuous and constant apart from isolated jumps. In that case we may denote the jump times of $X$ by $\tau_1,\tau_2,\ldots$, with the understanding that $\tau_n=\infty$ if there are fewer than $n$ jumps.
\end{definition}

\begin{proof}
$(X_t)_{t\ge0}$ is clearly pure jump-type Markov process (interested reader can verify this fact using proof similar to Theorem 13.4 of \cite{kallenberg2021foundations}).

We remind the reader associated definitions.
\begin{definition}
    \begin{itemize}
        \item $c(x) \defeq (\E_x\tau_1)^{-1}$, where $\E_x\tau_1$ is the expected time time until the first jump from the state $x$, i.e. the smallest $t>0$ such that $X_t\ne x$;
        \item $\mu(x,B)\defeq \P_x(X_{\tau_1\in B})$;
        \item $\alpha(x,B) \defeq c(x)\mu(x,B)$, we call $\alpha$ the \textit{rate kernel}
    \end{itemize}
\end{definition}

The rate kernel can be determined as in the Proposition 13.7 of \cite{kallenberg2021foundations}.

Conditioned on $Y_0=\cdots=Y_{k-1}=x$, \[ \Pr(Y_k\ne x)=\frac{\beta(x,\{x\}^c)}{\beta(x,\Omega)},\, \frac{\gamma_k}{\beta(Y_{k-1},\Omega)}\sim\Exp(\beta(x,\Omega))\]

Thus, $\tau_1\sim\Exp(\beta(x,\{x\}^c))$ and $X_{\tau_1}\sim \frac{\beta(x,B\setminus\{x\})}{\beta(x,\{x\}^c)}$. Hence, $\alpha(x, B) = \beta(x, B\setminus \{ x\})$. 
\end{proof}

\section{Measurability of $A(t)$}\label{sec:measurability}

In this section we will show that $A(t)$ is measurable function of $X_{\le t}$. Recall the definition of $A(t)$ :

\defnA*

We now show that $A(t)$ is a measurable function of $X_{\le t}$. Fix $t$.

Recall that the sigma algebra $\mcN$ on $\mathbf{N}$ (integer-valued locally finite measures) of $\Lambda\times[0,t]\times[0,1]$ is generated by sets of the form:
\begin{align*}
    \{\mu&:\mu(B)=k\}&\forall\text{ measurable }B\subseteq\Lambda\times[0,t]\times[0,1],\forall k\in\Z_{\ge 0}
\end{align*}

We recall some standard facts from measure theory, which can be found in textbooks such as \cite{kallenberg2021foundations}.
We will use these in the sequel without comment.

\begin{fact}
    A measurable space $(X,\mcX)$ consists of a set $X$ and a $\sigma$-algebra on $X$ of measurable subsets of $X$. The complements, countable unions, and countable intersections of measurable sets are measurable.
\end{fact}

We often abuse notation and simply write $X$ when the $\sigma$-algebra is clear from context.
$\R$ and $[0,\infty]$ will be equipped with the Borel $\sigma$-algebra.
$\Omega=\{\eta\subseteq\Lambda:|\eta|<\infty\}$ will be equipped by the smallest $\sigma$-algebra for which $\eta\mapsto \eta(B)$ is measurable for all measurable $B\subseteq\Lambda$. 

Given two sets $X$ and $Y$ equipped with $\sigma$-algebras $\mcX$ and $\mcY$, we equip $X\times Y$ with the product $\sigma$-algebra, generated by (i.e. smallest $\sigma$-algebra that contains) $\{A\times B:A\in\mcX,B\in\mcY\}$. In particular, the projection maps $(x,y)\mapsto x$ and $(x,y)\mapsto y$ are measurable. The product of $\sigma$-algebras is associative.

\begin{fact}
    A function $f:X\to Y$ is measurable iff the preimage of a measurable subset of $Y$ is a measurable subset of $X$. The composition of measurable functions is measurable.
    An indicator function of a set is a measurable function iff the set is a measurable set. A simple function is a finite linear combination of indicator functions of measurable sets. Any $[0,\infty]$-valued measurable function can be written as a point-wise limit of a monotonically increasing sequence of simple functions. The sum, product, and point-wise limit of (extended) real-valued measurable functions is measurable. The reciprocal of a non-zero measurable function is measurable. 
\end{fact}

\begin{lemma}\label{lem:measurable:sum_over_points}
    Let $\Lambda\subseteq\R^d$ be a bounded measurable set.
    Suppose $\phi:\Lambda\times\Lambda\to[0,\infty]$ is measurable.

    Then,
    \begin{align*}
        (x,\eta)\mapsto \sum_{y\in\eta}\phi(x,y)
    \end{align*}
    is measurable.
\end{lemma}
\begin{proof}
    We know that the function $(x,\eta)\mapsto \eta(B)$ is measurable for any measurable $B\subseteq\Lambda$.
    
    For any measurable $A,B\subseteq\Lambda$,
    \begin{align*}
        (x,\eta)\mapsto 1_A(x)\eta(B)=\sum_{y\in\eta}1_A(x)1_B(y)
    \end{align*}
    is measurable.
    Consider
    \begin{align*}
        \mcP=\{C\subseteq\Lambda\times\Lambda\text{ measurable}:(x,\eta)\mapsto\sum_{y\in\eta}1_C(x,y)\text{ is measurable}\}
    \end{align*}
    This is a $\lambda$-system, and contains the $\pi$-system $\{A\times B:A,B\subseteq\Lambda\text{ measurable}\}$.
    Hence, by the $\pi$-$\lambda$ theorem, $\mcP$ contains the $\sigma$-algebra generated by $\{A\times B:A,B\subseteq\Lambda\text{ measurable}\}$, i.e. measurable sets of $\Lambda\times\Lambda$.
    Thus, for any simple function $\phi:\Lambda\times\Lambda\to[0,\infty]$,
    \begin{align*}
        (x,\eta)\mapsto\sum_{y\in\eta}\phi(x,y)
    \end{align*}
    is measurable. Taking monotone limits of simple functions yields the claim.
\end{proof}

\begin{lemma}\label{lem:measurable:partition_function}
If $\bm\lambda:\Xi\times \Lambda\to[0,\lambda]$ is a bounded measurable function, then $Z_{\bm\lambda}:\Xi\to\R$ defined by
\begin{align*}
Z_{\bm\lambda}=\sum_{k\ge 0}\dfrac{1}{k!}\int_{(\R^d)^k}\bm{\lambda}(x_1)\cdots\bm{\lambda}(x_k)e^{-H(x_1,\ldots,x_k)}dx_1\cdots dx_k~~~~~~~~
\text{for }H(x_1,\ldots,x_k)=\sum_{1\le i<j\le k}\phi(x_i,x_j)
\end{align*}
is measurable.
\end{lemma}

\begin{proof}
    For any $k\ge 0$, the function
    \begin{align*}
        (\xi,x_1,\ldots,x_k)\mapsto \bm\lambda(x_1)\cdots\bm\lambda(x_k)e^{-H(x_1,\ldots,x_k)}
    \end{align*}
    is measurable.
    By Lemma 1.28 of \cite{kallenberg2021foundations},
    \begin{align*}
        \xi\mapsto\int_{\Lambda^k}\bm{\lambda}(x_1)\cdots\bm{\lambda}(x_k)e^{-H(x_1,\ldots,x_k)}dx_1\cdots dx_k
    \end{align*}
    is measurable and lies in $[0,\lambda|\Lambda|]^k$. $Z_{\bm\lambda}$ is the limit of a sequence of measurable functions, and hence measurable.    
\end{proof}

\begin{lemma}\label{lem:measurable:intensity}
If $\bm\lambda:\Xi\times \Lambda\to[0,\lambda]$ is a measurable function, then $\iota_{\bm\lambda}:\Xi\times\Lambda\to\R$ defined by
\begin{align*}
\iota_{\bm\lambda}(x)=\bm\lambda(x)\frac{Z_{\bm\lambda e^{-\phi(x,\cdot)}}}{Z_{\bm\lambda}}
\end{align*}
is measurable.
\end{lemma}
\begin{proof}
The functions
\begin{align*}
    (\xi,y)\mapsto\bm\lambda_{\xi}(y)~~~~~~~~~~~~~~~~~
    (\xi,x,y)\mapsto\bm\lambda_{\xi}(y)e^{-\phi(x,y)}
\end{align*}
are measurable.
By \Cref{lem:measurable:partition_function},
\begin{align*}
    \xi\mapsto Z_{\bm\lambda} ~~~~~~~~~~~~~~~~~
    (\xi,x)\mapsto Z_{\bm\lambda e^{-\phi(x,\cdot)}}
\end{align*}
 are measurable. Also, $Z_{\bm\lambda}\ge 1$. Hence,
 \begin{align*}
     (\xi,x)&\mapsto\bm\lambda(x)\frac{Z_{\bm\lambda e^{-\phi(x,\cdot)}}}{Z_{\bm\lambda}}
 \end{align*}
is measurable. 
\end{proof}

\begin{lemma}\label{lem:measurable:kth}
    Let $\xi$ be drawn from a point process on $\Lambda\times[0,t]\times[0,1]$. Let the points of $\xi$ be $(x_1,t_1,l_1),(x_2,t_2,l_2),\ldots$ in time order $t_1<t_2<\cdots$.  Then, for any $k\ge 1$, the map
    \begin{align*}
        \xi&\mapsto (x_k,t_k,l_k)\in \Lambda\times[0,t]\times[0,1]\\
    \end{align*}
    is measurable on $E_k=\{\xi:|\xi|\ge k\}$.
\end{lemma}
\begin{proof}
    For any measurable $C\subseteq\Lambda\times[0,t]\times[0,1]$,
    \begin{align*}
        \{\xi:(x_k,t_k,l_k)\in C\} =\bigcup_{\substack{p,q\in((\Q\cap(0,t))\cup\{0,t\})\\p\le q}}(&\{\xi:\xi(\Lambda\times[0,p)\times[0,1])=k-1\}\, \cap
        \\
    &
    \{\xi:\xi(\Lambda\times[p,q]\times[0,1])=1\}\, \cap
    \\
    &
    \{\xi:\xi(C\cap(\Lambda\times[p,q]\times[0,1]))=1\})
    \end{align*}
    is a measurable subset of $E_k$.
\end{proof}

\begin{lemma}
Let $\xi$ be drawn from the Poisson point process of intensity $\lambda$ on $\Lambda\times[0,t]\times[0,1]$.

Define $A_0(\xi)=S$, and for $k=1,2,3,\ldots$
\begin{align*}
    A_k(\xi)&=\begin{cases}
        A_{k-1}(\xi)\cup\{x_k\}&\text{if $t_k\le t$ and $l_k\le \dfrac{\iota_{\mcT_{-t}\mcR_{A_{k-1}(\xi)}\bm\lambda}(x_k)}{\lambda}$}\\
        A_{k-1}(
        \xi)&\text{otherwise}
    \end{cases}
\end{align*}
Then, each $\xi\mapsto A_k(\xi)$ is measurable, and so  $A(\xi)=\bigcup_{k=1}^{\infty}A_k(\xi)$ is measurable as well (note that $A_k$ is increasing in $k$).
\end{lemma}

\begin{proof}
    We will show by induction on $k\ge 0$ that $\xi\mapsto A_k(\xi)$ is measurable.
    
    Note that $\xi\mapsto A_k(\xi)$ is measurable iff $\xi\mapsto |A_k(\xi)\cap B|$ is measurable for all measurable $B\subseteq\Lambda$.

    \begin{align*}
        |A_k(\xi)\cap B|=|A_{k-1}(\xi)\cap B|+1[x_k\in B,t_k\le t,l_k\le \dfrac{\iota_{\mcT_{-t}\mcR_{A_{k-1}(\xi)}\bm\lambda}(x_k)}{\lambda}]
    \end{align*} 
    where $1[P]$ equals one if $P$ is true and zero otherwise.

    By the inductive hypothesis, $\xi\mapsto A_{k-1}(\xi)$ is measurable. Thus, $\xi\mapsto |A_{k-1}(\xi)\cap B|$ is measurable.

    For measurable $\bm\lambda:\Lambda\to[0,\lambda]$, using \Cref{lem:measurable:sum_over_points}, we can show that the function
    \begin{align*}
        (\eta,x,t)\mapsto \bm\lambda(x)\exp(-t)\exp(-\sum_{y\in\eta}\phi(x,y))
    \end{align*}
    is measurable. Since $\xi\mapsto A_{k-1}(\xi)$ is measurable,
    \begin{align*}
        (\xi,x,t)\mapsto\bm\lambda(x)\exp(-t)\exp(-\sum_{y\in A_{k-1}(\xi)}\phi(x,y))
    \end{align*}
    is measurable. Hence, by \Cref{lem:measurable:intensity}, $(\xi,x)\mapsto\iota_{\mcT_{-t}\mcR_{A_{k-1}(\xi)}\bm\lambda}(x)$ is measurable.

    By \Cref{lem:measurable:kth}, $\xi\mapsto (x_k(\xi),t_k(\xi),l_k(\xi))$ is measurable on $E_k=\{\xi:|\xi|\ge k\}$.

    Thus, the functions
    \begin{align*}
        \xi&\mapsto 1[x_k(\xi)\in B]&\xi&\mapsto1[t_k\le t]&
        \xi&\mapsto 1[l_k\le \iota_{\mcT_{-t}\mcR_{A_{k-1}(\xi)}\bm\lambda}(x_k)]
    \end{align*}
    \begin{align*}
        \xi&\mapsto 1[x_k\in B,t_k\le t,l_k\le \dfrac{\iota_{\mcT_{-t}\mcR_{A_{k-1}(\xi)}\bm\lambda}(x_k)}{\lambda}]
    \end{align*}
    are measurable on $E_k$.

    We can make it measurable on the whole domain by defining it to be zero outside $E_k$.

    Thus,
    \begin{align*}
        \xi\mapsto |A_k(\xi)\cap B|=|A_{k-1}(\xi)\cap B|+1[x_k\in B,t_k\le t,l_k\le \dfrac{\iota_{\mcT_{-t}\mcR_{A_{k-1}(\xi)}\bm\lambda}(x_k)}{\lambda}]
    \end{align*}
    is measurable. 

    Finally, for any measurable $B\subseteq\Lambda$, since $A_k\nearrow A$,
    \begin{align*}
        \xi\mapsto |A(\xi)\cap B|=\lim_{k\to\infty}|A_k(\xi)\cap B|
    \end{align*}
    is measurable, so $\xi\mapsto A(\xi)$ is measurable as well.
\end{proof}

\section{Remaining Proofs} 

\subsection{Section \ref{sec: loc scheme}}\label{sec: append for sec 2+3}

Next we will provide the proof of \Cref{lemma: another def of tilts}. 
\lemanotherdeftilt*

\begin{proof}
    Recall that we define \begin{align*}
    \mcT_{-t}\nu\defeq g_t\nu
\end{align*}
where $g_t:\Omega\to\R$

\begin{align*}
    g_t(\eta)&=\dfrac{e^{-t|\eta|}}{\int_{\Omega}e^{-t|y|}d\nu(y)}&\eta\in\Omega
\end{align*}

We know that for a bounded measurable function $\phi: \Omega \to \R$ 
\begin{align*}
    \int \phi(\eta) d\mu_{\bm\lambda}(\eta) = C_1 \int \phi(\eta) e^{-H(\eta)} d\rho_{\bm\lambda}(\eta)
\end{align*}
where $\rho_{\bm\lambda}$ is the probability measure of a Poisson point process with intensity $\bm\lambda$ and $C_1$ is some constant. 
By the definition of tilts, 
\begin{align*}
    \int \phi(\eta) d\mcT_{-t} \mu_{\bm\lambda}(\eta) = C_2 \int \phi(\eta) e^{-t|\eta|} e^{-H(\eta)} d\rho_{\bm\lambda}(\eta) = \\
    C_2 \int \phi(\eta)  e^{-H(\eta)} dp_{e^{-t}\bm\lambda}(\eta) = C_3 \int \phi(\eta) d\mu_{e^{-t} \bm\lambda}(\eta)
\end{align*}
for some constants $C_2, C_3$.

Taking $\phi(x)\equiv 1$ above shows that $C_3=1$. Hence, $\mcT_{-t} \mu_{\bm\lambda}  = \mu_{e^{-t} \bm\lambda}$.

\end{proof}

\subsection{Subsection \ref{subsec: variance conservation}}

\begin{lemma}\label{lem:gronwall_variant}
    Suppose $(X_t)_{t\ge 0}$ is a $\R_{\ge 0}$-valued random process adapted to filtration $(\mcF_t)_{t\ge 0}$ such that for some $c>0$,
    for all $t\ge 0$,
    \begin{align*}
        \E[X_{t+h}-X_t|\mcF_t]\ge -chX_t+o(h)
    \end{align*}
    where $o(h)$ is a quantity bounded in magnitude by some function $r(h)$ depending only on $h$ such that $\lim_{h\to0}\frac{r(h)}{h}=0$.
    
    Then, for all $t\ge 0$,
    \begin{align*}
        \E[X_t]\ge e^{-ct}X_0
    \end{align*}
\end{lemma}
\begin{proof}

Let $C>c$.
There exists some $\delta_1>0$ such that for $h<\delta_1$,
\begin{align*}
    1-ch\ge e^{-Ch}
\end{align*}

For any $\eps>0$, there exists $\delta_2>0$ such that for $h<\delta_2$ and $t\ge 0$,
\begin{align*}
    \E[X_{t+h}|\mcF_t]\ge (1-ch)X_t-\epsilon h.
\end{align*}

Now, for $h<\min(\delta_1,\delta_2)$ and $t\ge 0$,
we have that
\begin{align*}
    \E[X_{t+h}|\mcF_t]&\ge (1-ch)X_t-\epsilon h\\
    &\ge e^{-Ch}X_t-\epsilon h.
\end{align*}
Thus, 
\begin{align*}
     \E[e^{C(t+h)}X_{t+h}|\mcF_t]&\ge e^{Ct}X_t-\epsilon he^{C(t+h)}\\
    \E[e^{C(t+h)}X_{t+h}+\epsilon(t+h)e^{C(t+h)}|\mcF_t]&\ge e^{Ct}X_t+\epsilon te^{Ct}.
\end{align*}

Letting $Y_t=e^{Ct}X_t+\epsilon te^{Ct}$, we have that for $h<\min(\delta_1,\delta_2)$ and $t\ge 0$,
\begin{align*}
    \E[Y_{t+h}|\mcF_t]&\ge Y_t.
\end{align*}

Choose $t_0,t_1,\ldots,t_k$ such that $0=t_0<t_1<\cdots<t_k=s$ and $t_j-t_{j-1}<\min(\delta_1,\delta_2)$. Then,
\begin{align*}
    \E[Y_{t_0}|\mcF_{t_0}]&=Y_{t_0}\\
    \E[Y_{t_{j+1}}|\mcF_{t_0}]&=\E[\E[Y_{t_{j+1}}|\mcF_{t_j}]|\mcF_{t_0}]\ge \E[Y_{t_j}|\mcF_{t_0}]&\text{for $0\le j\le k-1$}
\end{align*}
so by induction, $\E[Y_{s}|\mcF_0]\ge Y_0$, meaning that
\begin{align*}
    \E[e^{Cs}X_s]+\epsilon {s}e^{Cs}&\ge X_s.
\end{align*}
As this holds for any $\epsilon>0$, and for any $C>c$, we have
\begin{align*}
    \E[X_s]&\ge e^{-cs}X_0.
\end{align*}
\end{proof}

\subsection{Subsection \ref{subsec: bounding numbr of iterations}}

\timediffind*

\begin{proof}
Let $\mcG_i$ be the $\sigma$-algebra generated by $\gamma_1,\ldots,\gamma_i$ and $Y_0,\ldots,Y_i$.

We will show by induction on $l\ge 0$ that for all $i\ge 0$, for $E\in\mcG_i$ and $\alpha>0$,
\begin{align*}
    &\Pr(E\wedge |Y_i|=l\wedge t_{\min\{k>i:|Y_k|\ge |Y_{k-1}|\}}-t_i>\alpha)\\
    &=\Pr(E\wedge |Y_i|=l)e^{-\alpha\bm\lambda(\Lambda)}
\end{align*}

First, suppose $l=0$.

Note that $|Y_i|=0$ implies $|Y_{i+1}|\ge |Y_i|$. Thus,
\begin{align*}
    &\Pr(E\wedge |Y_i|=0 \wedge t_{\min\{k>i:|Y_k|\ge |Y_{k-1}|\}}-t_i>\alpha)\\
    &=\Pr(E\wedge |Y_i|=0\wedge t_{i+1}-t_i>\alpha)\\
    &=\Pr(E\wedge |Y_i|=0\wedge 
    \frac{\gamma_{i+1}}{\bm\lambda(\Lambda)}>\alpha)\\
    &=\Pr(E\wedge |Y_i|=0)\Pr(
    \frac{\gamma_{i+1}}{\bm\lambda(\Lambda)}>\alpha)\\
    &=\Pr(E\wedge |Y_i|=0)e^{-\alpha\bm\lambda(\Lambda)}
\end{align*}

Now, suppose $l\ge 1$. We split into two cases depending on whether we have a death or attempted birth.
\begin{align*}
    &\Pr(E\wedge |Y_i|=l \wedge t_{\min\{k>i:|Y_k|\ge |Y_{k-1}|\}}-t_i>\alpha)\\
    &=\Pr(E\wedge |Y_i|=l \wedge |Y_{i+1}|\ge |Y_i| \wedge t_{\min\{k>i:|Y_k|\ge |Y_{k-1}|\}}-t_i>\alpha)\\
    &~+\Pr(E\wedge |Y_i|=l \wedge |Y_{i+1}|=l-1 \wedge t_{\min\{k>i:|Y_k|\ge |Y_{k-1}|\}}-t_i>\alpha)
\end{align*}

In the case where we have an attempted birth (which corresponds to the first term),
\begin{align*}
    &\Pr(E\wedge |Y_i  |=l \wedge |Y_{i+1}|\ge |Y_i| \wedge t_{\min\{k>i:|Y_k|\ge |Y_{k-1}|\}}-t_i>\alpha)\\
    &=\Pr(E\wedge |Y_i|=l \wedge |Y_{i+1}|\ge |Y_i| \wedge t_{i+1}-t_i>\alpha)\\
    &=\Pr(E\wedge |Y_i|=l \wedge |Y_{i+1}|\ge |Y_i| \wedge \frac{\gamma_{i+1}}{l+\bm\lambda(\Lambda)}>\alpha)\\
    &=\Pr(E\wedge |Y_i|=l \wedge |Y_{i+1}|\ge |Y_i|)\Pr(\frac{\gamma_{i+1}}{l+\bm\lambda(\Lambda)}>\alpha)\\
    &=\Pr(E\wedge |Y_i|=l \wedge |Y_{i+1}|\ge |Y_i|)\cdot e^{-(l+\bm\lambda(\Lambda))\alpha}\\
    &=\Pr(E\wedge |Y_i|=l)\cdot\frac{\bm\lambda(\Lambda)}{l+\bm\lambda(\Lambda)}\cdot e^{-(l+\bm\lambda(\Lambda))\alpha}
\end{align*}

In the case where we have a death,
\begin{align*}
    &\Pr(E\wedge |Y_i|=l \wedge |Y_{i+1}|=l-1 \wedge t_{\min\{k>i:|Y_k|\ge |Y_{k-1}|\}}-t_i>\alpha)\\
    &=\Pr(E\wedge |Y_i|=l \wedge |Y_{i+1}|=l-1 \wedge t_{\min\{k>i+1:|Y_k|\ge |Y_{k-1}|\}}-t_{i+1}+t_{i+1}-t_i>\alpha)\\
    &=\Pr(E\wedge |Y_i|=l \wedge |Y_{i+1}|=l-1 \wedge t_{\min\{k>i+1:|Y_k|\ge |Y_{k-1}|\}}-t_{i+1}+\frac{\gamma_{i+1}}{l+\bm\lambda(\Lambda)}>\alpha)\\
    &=\sum_{m=0}^{\infty}\Pr(E\wedge |Y_i|=l \wedge |Y_{i+1}|=l-1 \wedge mh\le \gamma_{i+1}\le (m+1)h\wedge t_{\min\{k>i+1:|Y_k|\ge |Y_{k-1}|\}}-t_{i+1}+\frac{\gamma_{i+1}}{l+\bm\lambda(\Lambda)}>\alpha)\\
    &\le \sum_{m=0}^{\infty}\Pr(E\wedge |Y_i|=l \wedge |Y_{i+1}|=l-1 \wedge mh\le \gamma_{i+1}\le (m+1)h\wedge t_{\min\{k>i+1:|Y_k|\ge |Y_{k-1}|\}}-t_{i+1}+\frac{(m+1)h}{l+\bm\lambda(\Lambda)}>\alpha)\\
    &=\sum_{m=0}^{\infty}\Pr(E\wedge |Y_i|=l \wedge |Y_{i+1}|=l-1 \wedge mh\le \gamma_{i+1}\le (m+1)h)\Pr\left(\Exp(\bm\lambda(\Lambda))+\frac{(m+1)h}{l+\bm\lambda(\Lambda)}>\alpha\right)\tag{by the inductive hypothesis}\\
    &=\sum_{m=0}^{\infty}\Pr(E\wedge |Y_i|=l \wedge |Y_{i+1}|=l-1)\Pr(mh\le \gamma_{i+1}\le (m+1)h)\Pr\left(\Exp(\bm\lambda(\Lambda))+\frac{(m+1)h}{l+\bm\lambda(\Lambda)}>\alpha\right)\\
    &=\Pr(E\wedge |Y_i|=l)\cdot\frac{l}{l+\bm\lambda(\Lambda)}\cdot\sum_{m=0}^{\infty}\Pr(mh\le \Exp(1)\le (m+1)h)\Pr\left(\Exp(\bm\lambda(\Lambda))+\frac{(m+1)h}{l+\bm\lambda(\Lambda)}>\alpha\right)\\
\end{align*}

Taking $h\to 0$ along a sequence where $h$ divides $\alpha(l+\bm\lambda(\Lambda))$ evenly,
\begin{align*}
    &\sum_{m=0}^{\infty}\Pr(mh\le \Exp(1)\le (m+1)h)\Pr\left(\Exp(\bm\lambda(\Lambda))+\frac{(m+1)h}{l+\bm\lambda(\Lambda)}>\alpha\right)\\
    &=\Pr(\Exp(1)>\alpha(l+\bm\lambda(\Lambda)))+\sum_{m=0}^{\frac{\alpha(l+\bm\lambda(\Lambda))}{h}-1}(e^{-mh}-e^{-(m+1)h})\exp\left(-\bm\lambda(\Lambda)\alpha+\frac{\bm\lambda(\Lambda)(m+1)h}{l+\bm\lambda(\Lambda)}\right)\\
    &=e^{-\alpha(l+\bm\lambda(\Lambda))}+e^{-\alpha\bm\lambda(\Lambda)}(e^h-1)\sum_{m=0}^{\frac{\alpha(l+\bm\lambda(\Lambda))}{h}-1}\exp\left(\frac{\bm\lambda(\Lambda)(m+1)h}{l+\bm\lambda(\Lambda)}-(m+1)h\right)\\
    &=e^{-\alpha(l+\bm\lambda(\Lambda))}+e^{-\alpha\bm\lambda(\Lambda)}(e^h-1)\sum_{m=0}^{\frac{\alpha(l+\bm\lambda(\Lambda))}{h}-1}e^{-\frac{l}{l+\bm\lambda(\Lambda)}(m+1)h}\\
    &=e^{-\alpha(l+\bm\lambda(\Lambda))}+e^{-\alpha\bm\lambda(\Lambda)}(e^h-1)e^{-\frac{l}{l+\bm\lambda(\Lambda)}h}\frac{1-e^{-l\alpha}}{1-e^{-\frac{l}{l+\bm\lambda(\Lambda)}h}}\\
    &\xrightarrow[h\to0]{}e^{-\alpha(l+\bm\lambda(\Lambda))}+e^{-\alpha\bm\lambda(\Lambda)}\cdot\frac{l+\bm\lambda(\Lambda)}{l}(1-e^{-l\alpha}).
\end{align*}

Hence,
\begin{align*}
    &\Pr(E\wedge |Y_i|=l \wedge |Y_{i+1}|=l-1 \wedge t_{\min\{k>i:|Y_k|\ge |Y_{k-1}|\}}-t_i>\alpha)\\
    &\le\Pr(E\wedge |Y_i|=l)\cdot\frac{l}{l+\bm\lambda(\Lambda)}\cdot(e^{-\alpha(l+\bm\lambda(\Lambda))}+e^{-\alpha\bm\lambda(\Lambda)}\cdot\frac{l+\bm\lambda(\Lambda)}{l}(1-e^{-l\alpha}))\\
    &=\Pr(E\wedge |Y_i|=l)\left(\frac{l}{l+\bm\lambda(\Lambda)}e^{-\alpha(l+\bm\lambda(\Lambda))}+e^{-\alpha \bm\lambda(\Lambda)}-e^{-\alpha(l+\bm\lambda(\Lambda))}\right)\\
    &=\Pr(E\wedge |Y_i|=l)\left(e^{-\alpha \bm\lambda(\Lambda)}-\frac{\bm\lambda(\Lambda)}{l+\bm\lambda(\Lambda)}e^{-\alpha(l+\bm\lambda(\Lambda))}\right)
\end{align*}
A similar argument shows the opposite inequality, so we in fact have equality. 

Thus,
\begin{align*}
    &\Pr(E\wedge |Y_i|=l \wedge t_{\min\{k>i:|Y_k|\ge |Y_{k-1}|\}}-t_i>\alpha)\\
    &=\Pr(E\wedge |Y_i|=l)\cdot\frac{\bm\lambda(\Lambda)}{l+\bm\lambda(\Lambda)}\cdot e^{-(l+\bm\lambda(\Lambda))\alpha}+\Pr(E\wedge |Y_i|=l)\left(e^{-\alpha \bm\lambda(\Lambda)}-\frac{\bm\lambda(\Lambda)}{l+\bm\lambda(\Lambda)}e^{-\alpha(l+\bm\lambda(\Lambda))}\right)\\
    &=\Pr(E\wedge|Y_i|=l)e^{-\alpha\bm\lambda(\Lambda)}
\end{align*}

This completes the induction.
Summing over all $l\ge 0$, we get that for all $i\ge 0$, $E\in\mcG_i$, and $\alpha>0$,
\begin{align}
    &\Pr(E\wedge t_{\min\{k>i:|Y_k|\ge |Y_{k-1}|\}}-t_i>\alpha)=\Pr(E)e^{-\alpha\bm\lambda(\Lambda)}\label{eqn:diff_exponential}
\end{align}
hence, $\{t_{\min\{k>i:|Y_k|\ge |Y_{k-1}|\}}-t_i>\alpha\}$ is independent of $E$. 
Also, this implies that
\begin{align*}
    \Pr(t_{\min\{k>i:|Y_k|\ge |Y_{k-1}|\}}-t_i>\alpha)=e^{-\alpha\bm\lambda(\Lambda)}
\end{align*}
Thus, $t_{\min\{k>i:|Y_k|\ge |Y_{k-1}|\}}-t_i\sim\Exp(\bm\lambda(\Lambda))$.

\end{proof}

\subsection{Subsection \ref{subsec:cont_time_MC}}\label{sec: appendix cont MC}

\begin{definition}[Generator of semigroup \cite{Engel2000OneParameter}]

Let $(T_t)_{t\ge 0}$ be a strongly continuous semigroup on $L^2(\mu)$.
Its generator is the operator
\begin{align*}
\mcL f=\lim_{t\to0^+}\frac{T_tf-f}{t}
\end{align*}
whose domain $D(\mcL)$ is all $f\in L^2(\mu)$ for which this limit exists in $L^2(\mu)$.
\end{definition}

\begin{theorem}[\cite{Engel2000OneParameter} Chapter II Lemma 1.3; \cite{Ethier1986Markov} Chapter 1 Proposition 1.5] 

Let $(T_t)_{t\ge 0}$ be a strongly continuous semigroup on $L^2(\mu)$.
Then, $\mcL:D(\mcL)\subseteq L^2(\mu)\to L^2(\mu)$ is a linear operator.

For $f\in D(\mcL)$, we have $T_tf\in D(\mcL)$, and
\begin{align*}
\frac{d}{dt}T_tf&=T_t\mcL f=\mcL T_tf&t\ge 0
\end{align*}

In particular, for $t\ge 0$,
\begin{align*}
    \lim_{h\to 0}\|\frac{T_{t+h}f-T_tf}{h}-\mcL T_t f\|_{\mu}=0
\end{align*}
 where we take $h\to 0+$ if $t=0$.

\end{theorem}

\cite{Engel2000OneParameter} and \cite{Ethier1986Markov} actually state results for arbitrary Banach spaces, but we will only need it for the Banach space $L^2(\mu)$.

\begin{lemma}
\label{lem:mixing:bounded_functions_in_domain}
For all bounded measurable $f:\Omega\to\R$, we have $f\in D(\mcL)$, and
\begin{align*}
    \mcL f(\eta)&=\int f(\xi)-f(\eta)\alpha(\eta,d\xi)\\
    &=\sum_{x\in\eta}(f(\eta\setminus\{x\})-f(\eta))+\int_{\Lambda}e^{-\nabla_x^+H(\eta)}(f(\eta\cup\{x\})-f(\eta))\bm\lambda(x)dx
\end{align*}
\end{lemma}
\begin{proof}
By scaling, it suffices to assume that $f:\Omega\to[-1,1]$.
By \Cref{thm:mixing:jump_process_first_order},
\begin{align*}
    \left|\frac{T_hf(\eta)-f(\eta)}{h}-\int f(\xi)-f(\eta)\alpha(\eta,d\xi)\right|
    &\le h(\alpha(\eta,\Omega)^2+\int \alpha(\xi,\Omega)\alpha(\eta,d\xi))\\
    &\le h(|\eta|+\bm\lambda(\Lambda))^2+h(|\eta|+\bm\lambda(\Lambda))(|\eta|+1+\bm\lambda(\Lambda))
\end{align*}
Then, since $\mu$ is stochastically dominated by a Poisson point process of intensity $\bm\lambda$,
\begin{align*}
    \int\left(\frac{T_hf(\eta)-f(\eta)}{h}-\int f(\xi)-f(\eta)\alpha(\eta,d\xi)\right)^2d\mu&\le \int h^2(|\eta|+\bm\lambda(\Lambda))^2(2|\eta|+1+2\bm\lambda(\Lambda)))^2d\mu\\
    &\le h^2\E_{X\sim\Pois(\bm\lambda(\Lambda))}[(|\eta|+\bm\lambda(\Lambda))^2(2|\eta|+1+2\bm\lambda(\Lambda))^2]
\end{align*}
Since $\Pois(\bm\lambda(\Lambda))$ has finite moments, the expectation is finite, and so
\begin{align*}
    \lim_{h\to0+}\int\left(\frac{T_hf(\eta)-f(\eta)}{h}-\int f(\xi)-f(\eta)\alpha(\eta,d\xi)\right)^2d\mu&=0
\end{align*}
Thus, $f\in D(\mcL)$, and
\begin{align*}
    \mcL f(\eta)&=\int f(\xi)-f(\eta)\alpha(\eta,d\xi)\\
    &=\sum_{x\in\eta}(f(\eta\setminus\{x\})-f(\eta))+\int_{\Lambda}e^{-\nabla_x^+H(\eta)}(f(\eta\cup\{x\})-f(\eta))\bm\lambda(x)dx
\end{align*}

\end{proof}

\begin{definition}
    Let $\mcB\subseteq L^2(\mu)$ be the space of bounded measurable functions.
\end{definition}

\begin{lemma}\label{lem:mixing:dirichlet_form_and_generator}
For $f,g\in\mcB$,
\begin{align*}
    \mcE(f,g)&=-\int f\mcL gd\mu
\end{align*}
In particular,
\begin{align*}
    \int f\mcL gd\mu=\int g\mcL fd\mu
\end{align*}
\end{lemma}
This was shown for a slightly different class of functions in \cite{KL03}. For the reader's convenience, we detail the argument here.

\begin{proof}
Recall that the GNZ equations (see \Cref{thm: GNZ equations}) say that for any measurable $F:\Omega\times\Lambda\to[0,\infty)$
\begin{align*}
\int_{\Omega}\sum_{x\in\eta}F(\eta,x)d\mu(\eta)&=\int \int e^{-\nabla_x^+H(\eta)}F(\eta\cup\{x\},x)\bm\lambda(x)dxd\mu(\eta)
\end{align*}

Note that this also holds for any bounded measurable $F:\Omega\times\Lambda\to\R$, since we can split it into the positive and negative parts, and the integrals of both will be finite since
\begin{align*}
    \int \int e^{-\nabla_x^+H(\eta)}\bm\lambda(x)dxd\mu(\eta)\le \int \bm\lambda(\Lambda)d\mu(\eta)=\bm\lambda(\Lambda)<\infty
\end{align*}

Using the formula of $\mcL$ (\Cref{lem:mixing:bounded_functions_in_domain}), we have for bounded measurable $f,g:\Omega\to\R$,
\begin{align*}
    -\int f\mcL gd\mu&=\int\sum_{x\in\eta}f(\eta)(g(\eta)-g(\eta\setminus\{x\})d\mu(\eta)+\int\int e^{-\nabla_x^+H(\eta)}f(\eta)(g(\eta)-g(\eta\cup\{x\}))\bm\lambda(x)dxd\mu(\eta)
\end{align*}
Applying the GNZ equations (\Cref{thm: GNZ equations}) to the second summand with the bounded measurable function
\begin{align*}
    F(\eta,x)&=f(\eta\setminus\{x\})(g(\eta\setminus\{x\})-g(\eta))\\
    F(\eta\cup\{x\},x)
    &=f(\eta)(g(\eta)-g(\eta\cup\{x\}))&\text{ for almost all }x\in\Lambda\text{ for any given $\eta\in\Omega$}
\end{align*}
we get
\begin{align*}
    \int\int e^{-\nabla_x^+H(\eta)}f(\eta)(g(\eta)-g(\eta\cup\{x\}))\bm\lambda(x)dxd\mu(\eta)&=\int \sum_{x\in\eta}f(\eta\setminus\{x\})(g(\eta\setminus\{x\})-g(\eta))d\mu(\eta)
\end{align*}

Returning to the original expression, we have
\begin{align*}
    -\int f\mcL gd\mu&=\int\sum_{x\in\eta}(f(\eta)-f(\eta\setminus\{x\}))(g(\eta)-g(\eta\setminus\{x\})d\mu(\eta)\\
    &=\mcE(f,g)
\end{align*}

\end{proof}

\begin{lemma}
For bounded measurable $f:\Omega\to\R$,
\begin{align*}
    \int \mcL fd\mu&=0
\end{align*}
\end{lemma}
\begin{proof}
\begin{align*}
    \int \mcL fd\mu=\int f\mcL 1d\mu=0
\end{align*}
since $\mcL1=0$.
\end{proof}

\begin{lemma}\label{lem:mixing:inner_product_is_continuous_in_l2}
    If $f_n\to f$ in $L^2(\mu)$ and $g_n\to g$ in $L^2(\mu)$, then
    \begin{align*}
        \int f_ng_nd\mu\to\int fgd\mu
    \end{align*}
\end{lemma}
\begin{proof}
As $n\to \infty$,
    \begin{align*}
        \left|\int (f_n-f)g_nd\mu+\int f(g_n-g)d\mu\right|\le \|f_n-f\|_{\mu}\|g_n\|_{\mu}+\|f\|_{\mu}\|g_n-g\|_{\mu}\to 0
    \end{align*}
\end{proof}

\begin{lemma}[Reversibility]\label{lem:mixing:reversible}
For $f,g\in \mcB$ and all $t\ge 0$,
\begin{align}
    \int f(\eta)T_tg(\eta)d\mu(\eta)&=\int g(\eta)T_tf(\eta)d\mu(\eta)\label{eqn:reversible}
\end{align}

\end{lemma}

\begin{proof}
Fix $f,g\in \mcB$ and $t>0$. Then, $T_sf,T_sg\in \mcB$ for all $s\ge 0$, and
\begin{align*}
    \dfrac{d}{ds}\int T_sfT_{t-s}gd\mu&=\int (\mcL T_sf)(T_{t-s}g)-(T_sf)(\mcL T_{t-s}g)d\mu=0
\end{align*}

More precisely, restricting the limits below to $h\to 0+$ if $s=0$ and $h\to 0-$ if $s=t$,
\begin{align*}
    \lim_{h\to 0}\int \frac{T_{s+h}fT_{t-s-h}g-T_sfT_{t-s}g}{h}d\mu&=\lim_{h\to 0}\int\frac{T_{s+h}f-T_sf}{h}T_{t-s-h}g
-T_sf\frac{T_{t-s}g-T_{t-s-h}g}{h}
    d\mu\\
    &=\int \mcL T_sfT_{t-s}g-T_sf\mcL T_{t-s}g d\mu=0
\end{align*}
The second equality holds by \Cref{lem:mixing:inner_product_is_continuous_in_l2}
since each of the four terms converges in $L^2(\mu)$. The last equality holds by \Cref{lem:mixing:dirichlet_form_and_generator}.

Integrating over $s\in[0,t]$ proves the claim.
\end{proof}

\subsection{Subsection \ref{subsec: burn-in}}\label{sec: burn-in appendix}

\begin{lemma}\label{lem:burnin:linear_approximation_of_mu_t}

Suppose $\mu_0=\bm 1_S$. For $f:\Omega\to[0,1]$ and $t\in[0,s]$,
\begin{align*}
    \int fd\mu_{t+h}&=\int fd\mu_t+h\int \mcL f(\eta)d\mu_t(\eta)+o(h)
\end{align*}
where $o(h)$ is a quantity that is bounded in magnitude by some $r_{\bm\lambda(\Lambda),s,S}(h)$ not depending on $t$ such that $\lim_{h\to 0}\frac{r_{\bm\lambda(\Lambda),s,S}(h)}{h}=0$.
\end{lemma}
\begin{proof}
By \Cref{thm:mixing:jump_process_first_order},
\begin{align*}
    \left|T_hf(\eta)-f(\eta)-h\int(f(\xi)-f(\eta))\alpha(\eta,d\xi)\right|&\le \frac{1}{2}h^2\alpha(\eta,\Omega)^2+\frac{1}{2}h^2\int\alpha(\xi,\Omega)\alpha(\eta,d\xi)
\end{align*}

\begin{align*}
    &\left|\int fd\mu_{t+h}-\int fd\mu_t-h\int \mcL f(\eta)d\mu_t(\eta)\right|\\
    &=\left|\int T_hf(\eta)-f(\eta)-h\int(f(\xi)-f(\eta))\alpha(\eta,d\xi)d\mu_t(\eta)\right|\\
    &\le \int \left|\frac{1}{2}h^2\alpha(\eta,\Omega)^2+\frac{1}{2}h^2\int\alpha(\xi,\Omega)\alpha(\eta,d\xi)\right|d\mu_t(\eta)\\
    &\le \frac{1}{2}h^2\int (|\eta|+\bm\lambda(\Lambda))^2+(|\eta|+\bm\lambda(\Lambda))(|\eta|+\bm\lambda(\Lambda)+1)d\mu_t(\eta)\\
    &\le h^2\E_{X\sim\Pois(\bm\lambda(\Lambda)s)}[(X+|S|+\bm\lambda(\Lambda)+1)^2] &\text{by \Cref{lem:mixing:cardinality_bounded_by_births}}\\
    &=h^2((\bm\lambda(\Lambda)s+|S|+\bm\lambda(\Lambda)+1)^2+\bm\lambda(\Lambda)s)
\end{align*}
\end{proof}

\begin{lemma}\label{lem:burnin:generator_almost_null}
For $f:\Omega\to[0,1]$,
\begin{align*}
    \int \mcL fdU_t&=\sum_{x\in\eta\cap S}\int f(\eta\setminus\{x\})-f(\eta)dU_t
\end{align*}
\end{lemma}
\begin{proof}
Recall that
\begin{align*}
    \mcL f(\eta)&=\sum_{x\in \eta}(f(\eta\setminus\{x\})-f(\eta))+\int_{\Lambda} e^{H(\eta)-H(\eta\cup\{x\})}(f(\eta\cup\{x\})-f(\eta))\bm\lambda(x)dx
\end{align*}
Define
\begin{align*}
    \mcL^{(S)}f(\eta)&=\sum_{x\in \eta\setminus S}(f(\eta\setminus\{x\})-f(\eta))+\int_{\Lambda} e^{H(\eta)-H(\eta\cup\{x\})}(f(\eta\cup\{x\})-f(\eta))\bm\lambda(x)dx
\end{align*}

It suffices to show that
\begin{align*}
    \int \mcL^{(S)}f(\eta)dU_t(\eta)&=0
\end{align*}

For $T\subseteq S$, we have
\begin{align*}
    &\int \mcL^{(S)}f(\eta\cup T)e^{-H(\eta\cup T)}d\rho(\eta)\\
    &=\int \sum_{x\in (\eta\cup T)\setminus S}(f((\eta\cup T)\setminus\{x\})-f(\eta\cup T))e^{-H(\eta\cup T)}d\rho(\eta)\\
    &~~~+\int \left(\int_{\Lambda} e^{H(\eta\cup T)-H(\eta\cup T\cup\{x\})}(f(\eta\cup T\cup\{x\})-f(\eta\cup T))\bm\lambda(x)dx\right)e^{-H(\eta\cup T)}d\rho(\eta)\\
    &=\int \sum_{x\in \eta}e^{-H(\eta\cup T)}(f((\eta\cup T)\setminus\{x\})-f(\eta\cup T))d\rho(\eta)\\
    &~~~+\int \left(\int_{\Lambda} e^{-H(\eta\cup T\cup\{x\})}(f(\eta\cup T\cup\{x\})-f(\eta\cup T))\bm\lambda(x)dx\right)d\rho(\eta)
\end{align*}

Recall that the Mecke equations (\Cref{thm:mecke_equation}) say that
\begin{align*}
    \int \sum_{x\in\eta}F(x,\eta)d\rho(\eta)&=\int \int F(x,\eta\cup\{x\})d\rho(\eta)\bm\lambda(x)dx
\end{align*}
Plugging in $F(x,\eta)=e^{-H(\eta\cup T)}(f(\eta\cup T)-f((\eta\cup T)\setminus\{x\}))$, we get
\begin{align*}
    &\int \sum_{x\in \eta}e^{-H(\eta\cup T)}(f(\eta\cup T)-f((\eta\cup T)\setminus\{x\}))d\rho(\eta)\\
    &=\int \left(\int_{\Lambda} e^{-H(\eta\cup T\cup\{x\})}(f(\eta\cup T\cup\{x\})-f((\eta\cup T)\setminus\{x\}))\bm\lambda(x)dx\right)d\rho(\eta)\\
    &=\int \left(\int_{\Lambda} e^{-H(\eta\cup T\cup\{x\})}(f(\eta\cup T\cup\{x\})-f(\eta\cup T)\bm\lambda(x)dx\right)d\rho(\eta)&\text{$x\notin \eta$ for almost all $x\in\Lambda$}
\end{align*}
Hence,
\begin{align*}
    \int \mcL^{(S)}f(\eta\cup T)e^{-H(\eta\cup T)}d\rho(\eta)&=0
\end{align*}
Taking a linear combination over $T\subseteq S$, we get
\begin{align*}
    \int \mcL^{(S)}f(\eta)dU_t(\eta)&=0
\end{align*}
\end{proof}

\begin{lemma}\label{lem:burnin:upperbound_evolution}
For $t\ge 0$ and $h>0$,
\begin{align*}
(e^{-t-h})^{|T|}(1-e^{-t-h})^{|S|-|T|}&=(e^{-t})^{|T|}(1-e^{-t})^{|S|-|T|}(1-h|T|+\frac{e^{-t}}{1-e^{-t}}h(|S|-|T|))+o(h)
\end{align*}
where $o(h)$ means a term bounded in magnitude by $f(h)$ depending only on $h$ such that $\lim_{h\to0+}\frac{f(h)}{h}=0$.
\end{lemma}
\begin{proof}
By the product, power, and chain rules,
\begin{align*}
    \frac{d}{dt}(e^{-t})^{|T|}(1-e^{-t})^{|S|-|T|}&=(e^{-t})^{|T|}(1-e^{-t})^{|S|-|T|}\left(-|T|+(|S|-|T|)\frac{e^{-t}}{1-e^{-t}}\right)
\end{align*}
Also, observe that $(e^{-t})^{|T|}(1-e^{-t})^{|S|-|T|}$ can be written as a finite linear combination of terms of the form $e^{-\alpha t}$ for $\alpha\ge 0$. Thus, the second derivative is bounded in magnitude, and so by Taylor's remainder theorem, we can bound the $o(h)$ term by a quadratic that does not depend on $t$.
\end{proof}

\begin{lemma}\label{lem:burnin:upperbound_evolution_inequality}
For $f:\Omega\to[0,1]$,
\begin{align*}
    \int f+h\mcL fdU_t&\le \int f dU_{t+h}+o(h)
\end{align*}
\end{lemma}

\begin{proof}
\begin{align*}
\int(f+h\mcL f)(\eta)dU_t(\eta)&=\int f(\eta)+h\sum_{x\in\eta\cap S}( f(\eta\setminus\{x\})-f(\eta))dU_t(\eta)\\
&=\int (1-h|\eta\cap S|)f(\eta)dU_t(\eta)+\int \sum_{x\in\eta\cap S}f(\eta\setminus\{x\})dU_t(\eta)
\end{align*}

We expand the first summand as follows.
\begin{align*}
&\int (1-h|\eta\cap S|)f(\eta)dU_t(\eta)
\\
&=\sum_{T\subseteq S}(e^{-t})^{|T|}(1-e^{-t})^{|S|-|T|}\int (1-h|(\eta\cup T)\cap S|)f(\eta\cup T)e^{H(T)-H(\eta\cup T)}e^{\bm\lambda(\Lambda)}d\rho(\eta)\\
&=\sum_{T\subseteq S}(e^{-t})^{|T|}(1-e^{-t})^{|S|-|T|}\int (1-h|T|)f(\eta\cup T)e^{H(T)-H(\eta\cup T)}e^{\bm\lambda(\Lambda)}d\rho(\eta)&&\text{as $\eta\cap S=\emptyset$ almost surely}
\end{align*}

We now expand and bound the second summand. Observe that for repulsive pair potentials, $H$ is supermodular:
\begin{align}
    H(T)-H(\eta\cup T)&\le H(T\setminus\{x\})-H(\eta\cup T\setminus\{x\})&\text{ for }x\in T,\eta\cap T=\emptyset\label{eq:burnin:submodular}
\end{align}
Thus,
\begin{align*}
    &\int \sum_{x\in\eta\cap S}f(\eta\setminus\{x\})dU_t(\eta)\\
    &=\sum_{T\subseteq S}(e^{-t})^{|T|}(1-e^{-t})^{|S|-|T|}\int \sum_{x\in(\eta\cup T)\cap S}f((\eta\cup T)\setminus\{x\})e^{H(T)-H(\eta\cup T)}e^{\bm\lambda(\Lambda)}d\rho(\eta)\\
    &=\sum_{T\subseteq S}(e^{-t})^{|T|}(1-e^{-t})^{|S|-|T|}\sum_{x\in T}\int f(\eta\cup T\setminus\{x\})e^{H(T)-H(\eta\cup T)}e^{\bm\lambda(\Lambda)}d\rho(\eta)&&\text{as $\eta\cap S=\emptyset$ almost surely}\\
    &\le\sum_{T\subseteq S}(e^{-t})^{|T|}(1-e^{-t})^{|S|-|T|}\sum_{x\in T}\int f(\eta\cup T\setminus\{x\})e^{H(T\setminus\{x\})-H(\eta\cup T\setminus\{x\})}e^{\bm\lambda(\Lambda)}d\rho(\eta)&&\text{by (\ref{eq:burnin:submodular})}\\
    &=\sum_{T\subseteq S}(e^{-t})^{|T|+1}(1-e^{-t})^{|S|-|T|-1}|S\setminus T|\int f(\eta\cup T)e^{H(T)-H(\eta\cup T)}e^{\bm\lambda(\Lambda)}d\rho(\eta)
\end{align*}

Substituting these back into the original expression,
\begin{align*}
&\int(1+h\mcL)f(\eta)dU_t(\eta)\\
&=\int (1-h|\eta\cap S|)f(\eta)dU_t(\eta)+h\int \sum_{x\in\eta\cap S}f(\eta\setminus\{x\})dU_t(\eta)\\
&\le \sum_{T\subseteq S}(e^{-t})^{|T|}(1-e^{-t})^{|S|-|T|}\left(1-h|T|+h(|S|-|T|)\frac{e^{-t}}{1-e^{-t}}\right)\int f(\eta\cup T)e^{H(T)-H(\eta\cup T)}e^{\bm\lambda(\Lambda)}d\rho(\eta)\\
&=\sum_{T\subseteq S}(e^{-t-h})^{|T|}(1-e^{-t-h})^{|S|-|T|}\int f(\eta\cup T)e^{H(T)-H(\eta\cup T)}e^{\bm\lambda(\Lambda)}d\rho(\eta)+o(h)
\tag{\Cref{lem:burnin:upperbound_evolution}}\\
&=\int f(\eta)dU_{t+h}(\eta)+o(h)
\end{align*}
\end{proof}

\section{Spectral gap implies spectral independence}
\label{appendix:sg_implies_si}
It was shown in \cite{anari2024universality} that, in the discrete setting, a spectral gap for Glauber dynamics implies spectral independence. The same argument shows that an analogous statement is true for Gibbs point processes with repulsive pair potentials, which we include here for completeness.

\begin{theorem}
    Let $\mu$ be a Gibbs point process on finite volume $\Lambda$ with activity $\bm\lambda:\Omega\to\R_{\ge 0}$ and repulsive pair potential $\phi$, $\iota$ be its intensity measure, and $\Psi$ be its influence operator. If the spectral gap of continuum Glauber for $\mu$ is at least $\gamma$, then  $\rho(\Psi)\le \frac{1}{\gamma}$.
\end{theorem}

\begin{proof}
The spectral gap of continuum Glauber being at least $\gamma$ means that
\begin{align*}
    \mcE_{\mu}(F,F)&\ge \gamma \Var_{\eta\sim\mu}(F(\eta))&\text{for bounded measurable }F:\Omega\to\R.
\end{align*}
Recall that for bounded measurable $F,G:\Omega\to\R$,
\begin{align*}
\mcE_{\mu}(F,G)&=\int_{\Omega}\sum_{x\in\eta}(F(\eta\setminus\{x\})-F(\eta))(G(\eta\setminus\{x\})-G(\eta))d\mu(\eta).
\end{align*}
Thus, for all bounded measurable $F:\Omega\to\R$,
\begin{align*}
    \gamma \Var_{\eta\sim\mu}(F(\eta))\le \int_{\Omega}\sum_{x\in\eta}(F(\eta\setminus\{x\})-F(\eta))^2d\mu(\eta).
\end{align*}

For any bounded measurable $f:\Lambda\to\R$, setting
\begin{align*}
    F(\eta)=\max(-L,\min(L,\sum_{x\in\eta}f(x)))
\end{align*}
above 
and taking $L\to\infty$ and using dominated convergence yields
\begin{align*}
    \gamma\Var_{\eta\sim\mu}(\eta(f))\le \E_{\eta\sim\mu}[\sum_{x\in\eta}f(x)^2]=\E_{\eta\sim\mu}[\eta(f^2)]=\iota(f^2).
\end{align*}

By \Cref{lem:influence_covariance}, $\langle f,\Psi f\rangle_{\iota}=\Var_{\eta\sim\mu}[\eta(f)]$, so
\begin{align*}
    0\le \langle f,\Psi f\rangle_{\iota}&\le \frac{1}{\gamma}\langle f,f\rangle_{\iota}&\text{for bounded measurable }f:\Lambda\to\R.
\end{align*}
By a limiting argument, this in fact holds for $f\in L^2(\iota)$.
Finally, since $\Psi$ is a bounded self-adjoint operator, by \cite{reed1981functional} Theorem VI.6 and \cite{hunter2001applied} Lemma 8.26, 
\begin{align*}
    \rho(\Psi)=\|\Psi\|_{\iota\to\iota}=\sup_{\|f\|_\iota=1}|\langle f,\Psi f\rangle_{\iota}|\le \frac{1}{\gamma}.
\end{align*}
\end{proof}

\end{document}